\documentclass[reqno,12pt,letterpaper]{amsart}
\usepackage{amsmath,amssymb,amsthm,graphicx,mathrsfs,url}
\usepackage[usenames,dvipsnames]{color}
\usepackage[colorlinks=true,linkcolor=Red,citecolor=Green]{hyperref}
\usepackage{amsxtra}
\usepackage{wasysym} 
\usepackage{graphicx}
\usepackage{subcaption}

\usepackage{graphicx,color}

\def\arXiv#1{\href{http://arxiv.org/abs/#1}{arXiv:#1}}
\usepackage{array}
\newcolumntype{P}[1]{>{\centering\arraybackslash}m{#1}}

\def\wrtext#1{\relax\ifmmode{\leavevmode\hbox{#1}}\else{#1}\fi}

\def\abs#1{\left|#1\right|}

\def\?[#1]{\textbf{[#1]}\marginpar{\Large{\textbf{??}}}}

\let\epsilon=\varepsilon 

\def\norm#1{||\,#1\,||}

\newtheorem{theo}{Theorem}
\newtheorem{prop}{Proposition}[section]

\newtheorem{lemm}[prop]{Lemma}

\numberwithin{equation}{section}

\DeclareMathOperator{\ad}{ad}

\DeclareMathOperator{\Spec}{Spec}

\let\Im=\Imag

\let\Re=\Real
\DeclareMathOperator{\sgn}{sgn}

\DeclareMathOperator{\Ad}{Ad}
\DeclareMathOperator{\neigh}{neigh}

\def\indic{\operatorname{1\hskip-2.75pt\relax l}}
\usepackage{scalerel}

\newcommand\reallywidehat[1]{\arraycolsep=0pt\relax%
\begin{array}{c}
\stretchto{
  \scaleto{
    \scalerel*[\widthof{\ensuremath{#1}}]{\kern-.5pt\bigwedge\kern-.5pt}
    {\rule[-\textheight/2]{1ex}{\textheight}} 
  }{\textheight} %
}{0.5ex}\\           
#1\\                 
\rule{-1ex}{0ex}
\end{array}
}

\begin{document}

\title[Overdamped QNM]{Overdamped QNM for Schwarzschild black holes}

\author{Michael Hitrik}
\address{Department of Mathematics, University of California,
Los Angeles, CA 90095, USA.}
\email{hitrik@math.ucla.edu}


\author{Maciej Zworski}
\address{Department of Mathematics, University of California,
Berkeley, CA 94720, USA.}
\email{zworski@math.berkeley.edu}

\begin{abstract}
We show that the number of quasinormal modes (QNM) for Schwarz\-schild and Schwarzschild--de Sitter
black holes in a disc of radius $ r $ is bounded from below by $ c r^3 $, proving that the recent upper bound by J\'ez\'equel \cite{jez} is sharp. The argument is an application of a spectral asymptotics result for non-self-adjoint operators which provides a finer description of QNM, explaining the emergence of a distorted lattice and generalizing the lattice structure in strips described in \cite{saz} (see Figure \ref{f:qnm}). As a by-product we obtain an exponentially accurate Bohr--Sommerfeld quantization rule for one dimensional problems. The resulting description of QNM allows their accurate evaluation ``deep in the complex" where numerical methods break down due to pseudospectral effects (see Figure \ref{f:pse}).
\end{abstract}

\maketitle

\section{Introduction and statement of the results}

Quasinormal modes remain a topic of interest in physics and mathematics -- see for instance
the following recent references \cite{card} and \cite{gawa} respectively; one can find there
pointers to the vast literature on the subject. In this note we consider a purely mathematical question
of global lower  bounds on the number of quasinormal modes for Schwarzschild (S) or Schwarzschild--de Sitter (SdS)
black holes.
\begin{figure}
\includegraphics[width=16.5cm]{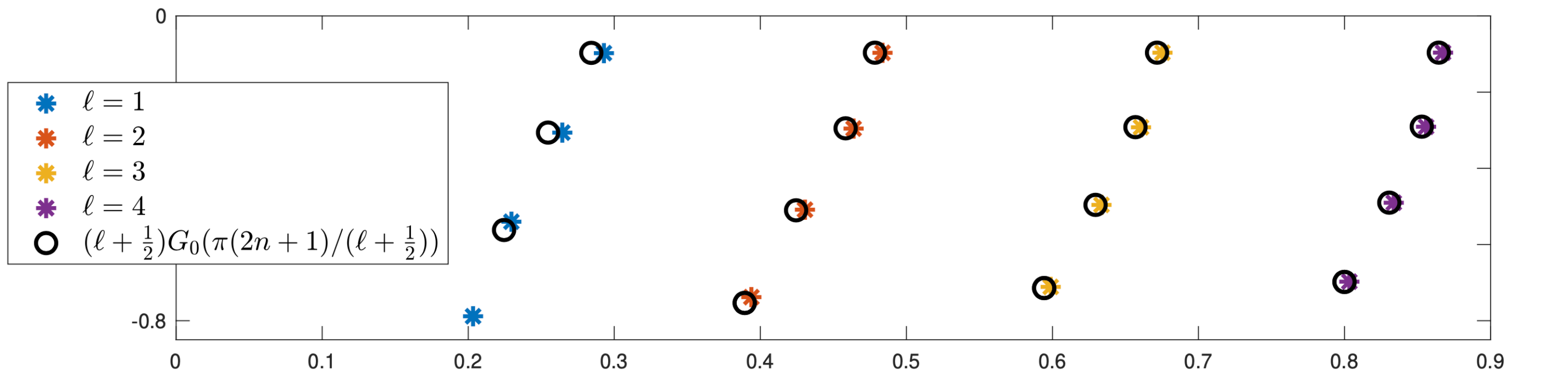}
\caption{\label{f:qnm}
A comparison of numerically computed QNM in \cite[Figure 6]{jan}
for the case of $ m = 1$, $ \Lambda = 0 $ and the leading term in the semiclas\-sical expression for QNM \eqref{eq:defG} given by
\eqref{eq:G0}. Even at large values of $ h = (\ell + \frac12)^{-1} = \frac23, \frac25, \frac27, \frac29$,
one sees a reasonable agreement as well as the emergence of the distorted lattice deeper in the complex
(as opposed to the lattice in strips described in S\'a Barreto--Zworski \cite{saz}).}
\end{figure}
Our motivation comes from a recent paper by J\'ez\'equel \cite{jez} who proved a global upper bound
for the number of quasinormal modes for SdS black holes: if we denote by
$ {\rm{QNM}}(m,\Lambda ) $ the set of quasinormal modes for an SdS black hole of mass $ m$ with
the cosmological constant $ \Lambda $ (see \S \ref{s:rewh}), then we have
\begin{equation}
\label{eq:jez}  |  {\rm{QNM}}(m,\Lambda ) \cap D ( 0 , r )  | \leq C(m, \Lambda ) r^3, \ \ \
0 < \Lambda < \frac{1}{9m^2}. \end{equation}
We recall from \cite{saz} that for $ \Lambda > 0 $, $ {\rm{QNM}}(m,\Lambda ) $ is a discrete set, while
for $ \Lambda = 0 $ (S black hole) the situation is more complicated and was described only recently by Stucker \cite{stuck}. However we have
\begin{theo}
\label{t:1}
For $ 0 < t \ll 1 $ and $ 0 <  \Lambda < 1/9m^2 $, let
$ A_t ( r ) := \{ \lambda : 1 \leq |\lambda | \leq r , \arg \lambda > -  t \}  $. Then,
there exist   $ c ( t, m, \Lambda ) > 0 $, such that, as $ r \to \infty $,
\begin{equation}
\label{eq:defN}  | {\rm{QNM}}(m,\Lambda ) \cap  A_t ( r ) |   =   c ( t, m , \Lambda ) r^3  + o (r^3 )
  \end{equation}
When $ \Lambda = 0 $ the asymptotic equality in \eqref{eq:defN} should be replaced by $ \geq $.
 \end{theo}
In particular, this shows that J\'ez\'equel's estimate is sharp. Theorem \ref{t:1} is a consequence of a more precise result about the distribution of quasinormal modes obtained from a one dimensional spectral result in \S \ref{s:gen}. Theorem \ref{t:2} below generalizes the result of S\'a Barreto--Zworski \cite{saz}, which was based on the work of Sj\"ostrand \cite{Sj86} (applied in the special 1D case).   The constant $ c ( t, m , \Lambda ) $ can be computed using the function $ G_0 $ appearing in Theorem \ref{t:2}  -- see \eqref{eq:const}.

A finer description of quasinormal modes is based on the spectral analysis
of the complex scaled (see \S \ref{s:cs})
Regge--Wheeler potential (see \S \ref{s:rewh}). The main point is that for each angular momentum $ \ell$ we obtain the description of quasinormal modes,
$ {\rm QNM}_\ell ( m , \Lambda ) $,  with $ \Im \lambda \geq - t \ell $, $ 0 < t \ll 1 $
and with $ e^{ - c |\lambda| } $ accuracy (though that level of finery is irrelevant for Theorem \ref{t:1}).
In particular, this gives
a mathematical description of the deformed lattice seen in numerical calculations.
For smaller neighbourhoods of the real axis ($
\Im \lambda > - |\lambda|^{1-\delta} $ for any $ \delta > 0 $), Iantchenko \cite{Ia} used the
results of Kaidi--Kerdelhu\'e \cite{N} to describe quasinormal modes with $ \mathcal O ( \langle \lambda\rangle^{-\infty} )$ accuracy.

\begin{theo}
\label{t:2}
For $ 0 \leq \Lambda < 1/9 m^2  $ and  $ \ell \gg 1 $,  the quasinormal modes (of multiplicity
$ 2 \ell + 1 $), with $ |\lambda | > c_0 \ell $ (for any positive $ c_0$) and $ \arg \lambda > - \theta $,
$ 0 < \theta \ll 1 $, are given by
\begin{equation}
\label{eq:defG}
\begin{gathered}
 \lambda_{\ell, n } = h^{-1}  G ( 2 \pi ( n + \tfrac12 ) h;h ),  \ \  h:=  ( \ell + \tfrac12 )^{-1}  , \ \ \ n\in \mathbb N,\ \ \
 G ( x; h ) \sim \sum_{ j=0}^\infty G_j ( x ) h^j ,
\end{gathered} \end{equation}
where the functions $ G_j $ are holomorphic near $ 0 $ and satisfy bounds $ |G_j | \leq A^{j+1} j! $, $ A > 0$
(that is, $ G $ is an analytic symbol),  $ G_0 $ is the inverse of a complex action (see Theorem
{\rm \ref{t:action}})
and $ G_1 \equiv 0 $.
\end{theo}


\noindent
{\bf Remark.} Just as in the self-adjoint case for Schr\"odinger operators
with $ h$-independent potentials (see \cite{CdV} for a clear
presentation) we expect $ G_{ 2j + 1} \equiv 0 $  but as our proofs work for more general operators we restrict ourselves to showing
that $ G_1 \equiv 0 $. That holds for all operators with $ 0 $ subprincipal symbol.

The definition of $ G_0 $ as the inverse of a complex action (see Theorem \ref{t:action}) gives
\begin{equation}
\label{eq:G00}    G_0 ( x ) = \frac{ ( 1 - 9 \Lambda m^2)^{\frac12} }{ 3 \sqrt 3 m } \left( 1 - \frac{i x}{ 2 \pi }  + \mathcal O (x^2 ) \right) , \end{equation}
which describes the lattice of QNM from \cite{saz}. The constant in \eqref{eq:defN} is given by
\begin{equation}
\label{eq:const}  c ( t, m, \Lambda ) =  \pi^{-1} ( 1 - 9 \Lambda m^2)^{-\frac32} 3^{\frac72} m^3 | \{ x \in \mathbb R_+ :
\arg G_0 ( x ) > - t \} |,
\end{equation}
where $ | E| $ is the length of an interval $ E $. For $ \Lambda = 0 $ and $ m = 1 $ additional terms in the expansion of $ G_0 $ are calculated by
Mathematica as follows:
\begin{equation}
\label{eq:G0}
G_0 ( x ) = \frac{1}{3\sqrt{3}} - \frac{i  x}{6 \sqrt{3} \pi} - \frac{5  x^2}{432\sqrt{3}\pi^2}
-\frac{235 i x^3}{93312\sqrt{3}\pi^3} + \frac{17795  x^4}{40310784 \sqrt{3}\pi^4} +
\cdots
\end{equation}

Despite the mathematical assumption on the smallness of the semiclassical parameter $ h$, the asymptotic formula (\ref{eq:defG})
is effective starting with $ \ell = 1$ -- see Figure \ref{f:qnm}. The terms $ G_j $ are in principle computable,
and using different methods and conventions similar expansions have been proposed in the physics literature as early as  \cite{iywi}. For an interesting recent physics perspective, see \cite{seiwi} and reference given there. Here we remark that as $ G $  in \eqref{eq:defG} is an analytic symbol,
the usual truncation of the expansion at $ j = [ ( eAh)^{-1}] $ approximates $ \lambda_{n,\ell} $ up to
errors of the size $ e^{ - c_0 \ell }$ -- see \S \ref{s:anas}.

\begin{figure}
\includegraphics[width=16cm]{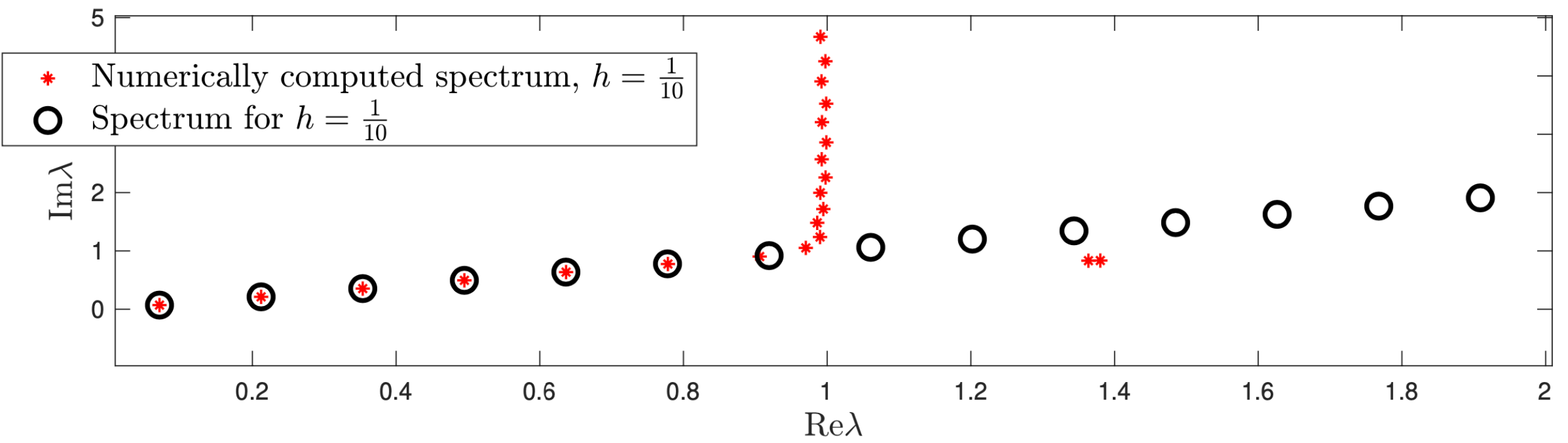}
\includegraphics[width=16cm]{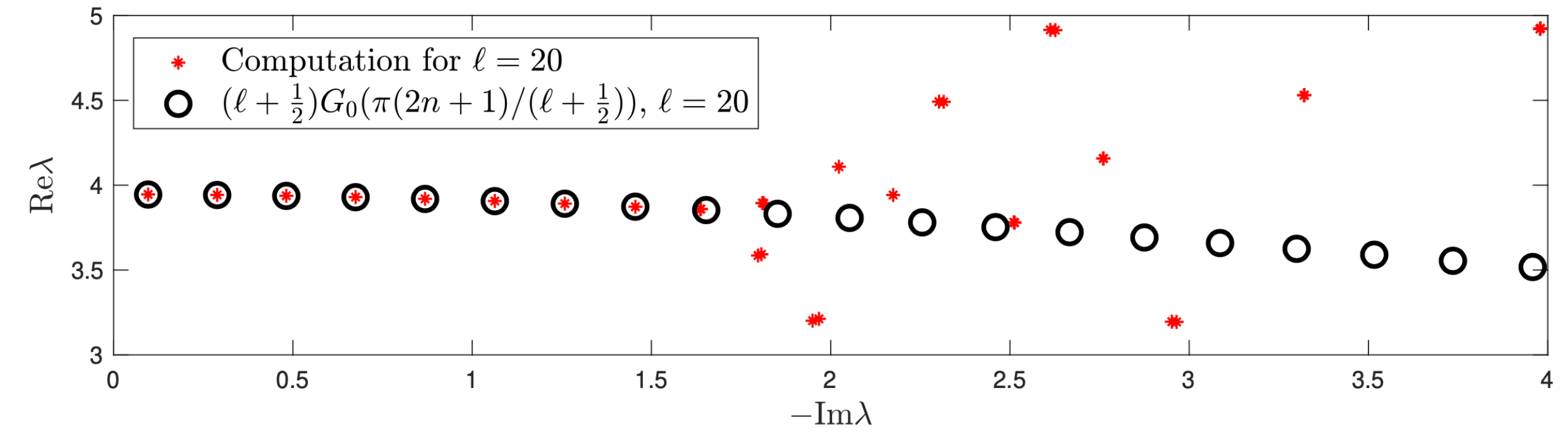}
\caption{\label{f:pse}
The top plot shows the spectrum of the rotated harmonic oscillator $ -h^2 \partial_x^2 + i x^2 $
used by Davies and Embree--Trefethen \cite[I.5]{trem} to illustrate spectral instability for non-normal
pseudodifferential operators. The eigenvalues (explicitly given by $ e^{\pi i /4 } h ( 2n + 1 )) $
are com\-pared to eigenvalues computed numerically
using the basis of the first 151 Hermite functions (eigenfunctions of $ - h^2 \partial_x^2 + x^2 $).
This illustrates the fragility of eigenvalues ``deep in the complex".
The bottom plot shows a calculation of QNM for $ m = 1 $, $ \Lambda = 0 $ and $ \ell = 20 $ using the Mathe\-matica code \cite{jan}. We again see divergence from the (mathematically) established result
(modulo the issues of the size of the region).}
\end{figure}

Theorem \ref{t:2} is an immediate consequence of a more general semiclassical Theorem \ref{t:gen}  about
spectra of analytic non-self-adjoint operators satisfying suitable ellipticity conditions near infinity. It differs from the
results of Hitrik \cite{Hi04} by changing the class of operators and by providing analyticity of the expansion.
Rather than using exact WKB methods specialized to one (complex) dimensional holomorphic equations,
we use general methods of analytic microlocal theory \cite{HiSj15}, \cite{sam} and the essential one dimensional
feature is the use of Vey's theorem \cite{vey} (holomorphic and symplectic version of Morse lemma) in (phase space) dimension two. It then provides a geometric interpretation of $ G_0 $ in (\ref{eq:defG}) as the inverse
of a holomorphic action.

For regular energy levels, the recent thesis by Duraffour \cite{Dur} provides Bohr--Sommer\-feld quantization rules
with analytic symbols in the same way as is done in Theorems \ref{t:2} and \ref{t:gen} for fixed neighbourhoods of critical levels. We also remark that using methods developed by Hagedorn and Joye, Toloza \cite{tolo} obtained an exponentially accurate description of a fixed number of lowest eigenvalues for Schr\"odinger operators with real analytic potentials having non-degenerate minima.

\noindent
{\bf Outline of the proof.} In \S \ref{s:rewh} we use the complex scaling method
(see \cite[\S 4.8]{res} for pointers to the literature) as presented in \cite[\S 5]{SjL}, to reduce the
problem of finding quasinormal modes to finding eigenvalues of an analytic non-self-adjoint
operator $ Q ( x, h D_x; h ) $ which is elliptic at $ 0 $ energy (in the semiclassical sense) except at $ ( x, \xi ) = (0,0) $
(see Lemma \ref{l:ell}). Near that point, the principal symbol of this operator is (essentially) given by
\begin{equation}
\label{eq:qxx}  q ( x, \xi ) =  \mu ( \xi^2 +x^2 ) + \mathcal O ( |x|^3 + |\xi|^3 ) , \ \ \ \mu := - i .
\end{equation}
 The result then follows from
a general result presented in \S \ref{s:gen} which holds under more general assumptions \eqref{eq:assH}.
We present the main steps in the case of \eqref{eq:qxx}, avoiding various technical issues which
require subtle modifications. The general strategy follows Helffer--Sj\"ostrand \cite{HelSj} (self-adjoint case)
and Hitrik \cite{Hi04} (description of the spectrum mod $ \mathcal O ( h^\infty ) $ and slightly stronger assumptions) but we try to make all the arguments as explicit as possible.

To that end, \S \ref{s:anas}
presents the needed results about analytic symbols (an example of one is given by
$ G ( x; h ) $ in \eqref{eq:defG}) and their quantization as formal pseudodifferential operators. The
key results are Proposition \ref{p:ave} which provides a variant of the method of averaging (with
\cite[Proposition 3.2]{HitSj3a} being a crucial component) and Lemma \ref{l:F2f} which, in a suitable
sense, writes quantization of $ g ( x^2 + \xi^2; h ) $ as $ G ( (hD_x)^2 +x^2; h ) $ (this ultimately allows
the quantization rule akin to \eqref{eq:defG}).

All the analysis is ultimately done ``on the FBI transform side", that is, after applying
an FBI transform
\begin{gather*}  T_h u ( z ) := \int e^{ \frac i h \varphi ( z, x ) } u ( x ) dx , \ \ \ T : L^2 ( \mathbb R ) \to
H_{\Phi_0} ( \mathbb C ) , \ \ \ \Phi_0 ( z ) := \tfrac12 |z|^2  \\
\varphi ( z, x ):= i\left(\tfrac{1}{2} z^2- \sqrt{2} zx + \tfrac{1}{2} x^2 \right), \ \
 H_{\Phi_0} ( \mathbb C ) := L^2 ( \mathbb C, e^{ -  |z|^2 /h  } dL ( z ) ) \cap \mathscr O ( \mathbb C ) .
\end{gather*}
As reviewed in \S \ref{s:global} we then have
\[  T_h Q ( x, h D_x; h  )u = P ( z, h D_z; h ) T_h u, \ \ \  P ( z, h D_z; h ) : H_{\Phi_0}(\mathbb C) \to H_{\Phi_0} (\mathbb C) . \]
(We neglect here the Sobolev orders, see \eqref{eq2.6.1}.) The operator $ P = P(z,hD_z;h)$ quantizes $ P ( z, \zeta; h ) $,
$ ( z, \zeta ) \in \Lambda_{\Phi_0 } := \{ ( z, (2/i) \partial_z \Phi_0 ( z ) ); z\in \mathbb C \} $ which is the image of
$ T^* \mathbb R $ under the complex canonical transformation,
$ ( x, - \varphi'_x(z,x)) \mapsto ( z , \varphi'_z(z,x) ) $, associated to $ T_h $. From \eqref{eq:qxx}
and the standard (and easy to check) fact
\[   T_h ( ( h D_x)^2 + x^2 ) u  = ( 2 zh \partial_z + h ) T_h u , \]
we see that the principal symbol of $ P $ is given by
\[ p ( z, \zeta ) = 2 \mu i z \zeta + \mathcal O ( |z|^3 +
|\zeta|^3 ), \ \ \ \  ( z, \zeta ) \in  \Lambda_{\Phi_0} . \]
We also note that the spectra of
$ P $ and $ Q $ on $ H_{\Phi_0}(\mathbb C) $ and $ L^2 ( \mathbb R) $, respectively, are the same.  The fact that
$ Q ( x, \xi; h ) $ extends to a complex neighbourhood of $ T^* \mathbb R $ (the key analyticity assumption)
shows that the operator $ P ( z, h D_z; h ) $ is bounded, uniformly as $h\rightarrow 0^+$, on $ H_{\Phi}(\mathbb C)$ where $ \Phi $ is a compactly supported perturbation of $ \Phi_0 $. Since there is no change near infinity, the spaces $ H_\Phi(\mathbb C) $ and
$ H_{\Phi_0}(\mathbb C) $ are the same but with different ($h$-dependent) norms. All of this is presented in complete
detail in \S \ref{s:global}.

The key component now is Vey's Morse Lemma \cite{vey}, \cite{CV} which, in our special case, is proved in the appendix.  It gives a holomorphic canonical transformation, $ \kappa $,  defined near $ 0 \in \mathbb C^2 $, such that $ \kappa^* p ( z, \zeta ) = g(2iz \zeta)$. To quantize it we need a theory of {\em local}
Fourier integral operators on the FBI transform side and that is reviewed in \S \ref{s:local} (necessarily without
all the proofs and relying on specific references to \cite{sam}, \cite{HiSj15}, and \cite{leb}). In this outline
we will {\em pretend} that $ \kappa $ is defined globally. Then
\begin{equation}
\label{eq:kappa0}     \kappa ( \Lambda_{\Phi_0 } ) = \Lambda_{\Phi } , \ \ \kappa^* p ( z, \zeta ) = g(2iz \zeta), \end{equation}
for some strictly convex weight $ \Phi $ which is close to $ \Phi_0 $ (only the behaviour near $ 0 \in \mathbb C $ ultimately matters). The canonical transformation $\kappa$ is {\em quantized} by a Fourier integral operator $ A : H_{\Phi_0 } \to H_\Phi $ (see Proposition \ref{p:AB}) which comes with a
 (microlocal) inverse $ B : H_{\Phi } \to H_{\Phi_0 }$. In view of \eqref{eq:kappa0}, the principal symbol of
 $ B P A: H_{\Phi_0 } \to H_{\Phi_0 } $ is given by $g(2iz \zeta)$. Using Proposition \ref{p:ave}
 and Lemma \ref{l:F2f} we can modify $ A $ and $ B $ so that, modulo exponentially small errors and
 microlocally near $ 0 \in \Lambda_{\Phi_0 }  $, $ B P A $ is equivalent to $
 G (  2 z h \partial_z + h; h ) $, for an analytic symbol $ G = g + \mathcal O(h)$ -- see Proposition \ref{p:norm} for the
 precise statement.  If these were global, we would immediately conclude that the spectrum of $ P $ and hence of $ Q $, is given by $ G (( 2 n + 1) h; h ) $. In practice, localization
 arguments are needed: since we are in a non-self-adjoint situation, construction of quasi-modes does not
 suffice for the existence of eigenvalues. That is carried out in \S \ref{s:gen} using a Grushin problem obtained
 from eigenstates of $ 2 z h \partial z+ h $, $ z^n $, $ n \leq c_0/h  $. This gives an exponentially accurate
 description of the spectrum of $ Q $ in a fixed neighbourhood of $ 0 $. In the special case of the
 (complex scaled) Regge--Wheeler operator we obtain Theorem \ref{t:2} and consequently Theorem \ref{t:1}.

\noindent
{\bf Acknowledgements.} We would like to thank Justin Holmer and  Zhongkai Tao for their generous help with Mathematica which led to the calculation \eqref{eq:G0} and to Figures~\ref{f:qnm}, \ref{f:pse},  Georgios Mavrogiannis whose question about global lower bounds for QNM provided additional motivation for this project, and San V\~u Ng\d{o}c for informing us of \cite{Dur}. The first named author is also grateful to Reid Johnson for very helpful discussions. The second named author gratefully acknowledges partial support by the NSF grant DMS-1901462. Partial support of both authors by the Simons Foundation under a ``Moir\'e Materials Magic" grant is also most gratefully acknowledged.

\section{Regge--Wheeler potentials}
\label{s:rewh}

In this section we review the definition of quasinormal modes for Schwarzschild and  Schwarzschild--de Sitter black holes using the method of complex scaling.

\subsection{Complexified Regge--Wheeler potentials}
\label{s:complexi}
We recall the Regge--Wheeler potential defined for $ 0 \leq \Lambda < 1/9 m^2$,~\cite[\S 4]{saz},
\begin{equation}
\label{eq:VRW}
\begin{gathered} V_\ell ( x ) := \alpha (r)^2r^{-2} ( \ell ( \ell + 1 ) + r \partial_r ( \alpha ( r)^2 ) ),  \\ x' ( r ) = \alpha(r)^{-2} , \ \ \alpha( r)^2 :=  1- \frac{2m}r - \tfrac13 \Lambda r^2.
\end{gathered} \end{equation}
We recall also that
\[  \alpha ( r)^2 = \left\{ \begin{array}{ll} \ \ \ \ \ \ \ \ \ \ \ r^{-1} ( r - 2m ), & \Lambda = 0 , \\
\tfrac13 \Lambda^2 r^{-1} ( r - r_0 ) ( r - r_- ) ( r_+ - r ), & \Lambda > 0 , \end{array} \right. \]
where $ r_0 < 0 < r_- < r_+ $. We have $ x : ( 2m , \infty ) \to (-\infty , \infty ) $ ($\Lambda = 0 $)
and $ x : ( r_- , r_+ ) \to ( -\infty, \infty ) $ ($ \Lambda > 0 $) and the transformations are
unique up to additive constants. Explicitly,  for $ \Lambda = 0 $,
\begin{equation}
\label{eq:xLa0}
\begin{gathered}
x ( r ) = r + 2m \log ( r - 2m ) ,  \ \ 2m < r < \infty,   \\
r ( x ) := 2m + 2m W_0 \left(  e^{ ({x}/{2m}) - 1 }/2m  \right),
\end{gathered}
\end{equation}
where $ W_0 $ is the first branch of the Lambert W-function,
$ W_0 ( t e^t ) = t$,   $ t > 0 $. For $ \Lambda > 0 $,
\[ x( r ) = a_0 \log ( r -r_0 ) + a_- \log ( r - r_- ) + a_+ \log ( r_+ - r ) ,  \ \ r_- < r < r_+ ,
\]
where $  \mp a_\pm, a_0  > 0 $ and we choose the principal branch of the logarithm (real on the positive real axis).

The transformation $ r \mapsto x ( r ) $ is holomorphic and conformal for $ r $ in open neighbourhoods of $ ( 2m, \infty ) \subset
\mathbb C \setminus \{ 2m \}  $ and
$ ( r_- , r_+ ) \subset \mathbb C \setminus \{r_-, r_+  \}   $, for $\Lambda = 0$ and $\Lambda > 0$, respectively, and we have a unique inverse $ r = r ( x ) $, for $ x $ in a neighbourhood of $ \mathbb R \subset \mathbb C $.
For $ \Lambda = 0 $ and $ |r | \gg 2m $,
\[
x ( r ) = r + 2m\log r + 2m\log \left( 1 - \frac{2m}{r}\right) = r + 2m\log r - \frac{4m^2}{r} + \cdots , \ \
   | \arg r | \leq \theta < \pi ,
\]
which produces a holomorphic inverse
\[ r(x) = x - 2m \log x + \sum_{ k=1 }^\infty \sum_{\ell = 0}^\infty c_{k\ell} \frac{ \log^\ell x }{ x^k }, \ \
|x| \gg 1 , \ \  | \arg x | \leq \theta < \pi .  \]

Let $ \Lambda_{\underline r } $ denote the logarithmic plane branched at $ \underline r $, that is, the Riemann surface of $ r  = f( z ) :=  \underline{r} + \exp z $, $ f:  \mathbb C  \to \Lambda_{\underline {r}} $ bi-holomorphically. Suppose $ \pi : \Lambda_{\underline r}
\to \mathbb C \setminus 0 $ is the natural projection and consider $
D_{\rho}  \subset \Lambda_{\underline r} $ defined by $ | \pi ( r ) | < \rho $.

If we take $ \underline r = 2m $ and $ \Lambda = 0 $, then
for $ \rho $ sufficiently small $ x : D_\rho \to \Omega_\rho \supset H_\rho := \{ x : \Re x < c_0 \log \rho \} $
is a bi-holomorphic map. We then have an inverse $ r = r ( x ) $ defined in $ H_\rho$.
We note that $ \alpha ( r )^2 r^{-2} $ is holomorphic in $ D_\rho $ and goes to $ 0 $ as
$ |r | \to 0 $. We conclude that
\[   \alpha ( r ( x ) )^2 r ( x )^{-2}  \in \mathscr O ( H_\rho ) , \ \ \
\alpha ( r ( x ) )^2 r ( x )^{-2}   \to 0 , \ \ \  \Re x \to - \infty. \]
(In fact the decay is exponential though this is not relevant in this paper.) The same property
is valid for $ r^{-1}  \partial_r ( \alpha ( r )^2 ) |_{ r = r ( x ) }$.

The same arguments work for $ \Lambda > 0 $ where we consider the logarithmic planes $ \Lambda_{r_\pm }$ with $ r_\pm $ corresponding to $ x= \pm \infty $. We summarize these results in the following lemma, see also~\cite[Proposition 4.1]{saz},
\begin{lemm}
\label{l:RWC}
In the notation of \eqref{eq:VRW} define the semiclasssical Regge--Wheeler potential as
\begin{equation}
\label{eq:SRW}
W ( x, h ) := ( \ell + \tfrac12)^{-2} V_\ell ( x ) = W_0 ( x ) + h^2 W_1 ( x ) , \ \ \ h = ( \ell + \tfrac12)^{-1} .
\end{equation}
Then for $ 0 \leq \Lambda < 1/9 m^2 $, the potentials $ W_j ( x )$  have holomorphic
extensions to $ \Omega_{\theta, \delta} := \{ z: |\Im z | < \delta , |z| < 2/\delta_0  \} \cup \{ z:  |\arg ( \pm z ) | < \theta,
 |z| > 1/\delta_0 \}$, and
\begin{equation}
\label{eq:RWC}
W_j \in \mathscr O ( \Omega_{\theta, \delta } ) , \ \  W_j ( z ) \to 0 , \  | z| \to \infty, \ z \in
\Omega_{\theta, \delta} ,
\end{equation}
where  $ 0 < \theta <  \pi/2 $ and $ 0 < \delta < \delta_0 $ (for some $ \delta_0 > 0$).
\end{lemm}

\noindent
{\bf Remark.} We did not evaluate $ \delta_0 $ as in the application we will use Lemma \ref{l:RWC} only for small $ \theta $. But it might be interesting to investigate large angle global complex scaling for $ W ( x, h )$.

\subsection{Complex scaling}
\label{s:cs}

Quasinormal modes, or scattering resonances, for the operator $ P =  (hD_x)^2 + W(x, h ) $
are defined by considering $ P $ as a holomorphic operator in $ \Omega_{ \theta_0 , \delta} $,
$ 0 < \theta_0 < \pi/2 $ (see Lemma \ref{l:RWC}) and restricting it to a contour $ \Gamma_\theta $. Resonances of $ P $ in $ \arg \lambda > - 2\theta$ are then the eigenvalues of $ P_\theta := P|_{\Gamma_\theta} $.
We refer to \cite[\S 5]{SjL} and \cite[\S 7.2]{Sj02} for the presentation of the method in a way applicable to our operator.

The choice of $ \Gamma_\theta $ is tailored to the structure of the potential and we start by
putting
\begin{equation}
\label{eq:defV1} V ( x ) = \alpha ( r ( x + x_0  ) )^2 r(x + x_0 )^{-2 } -
 \alpha ( r(x_0 )   )^2 r(x_0 ) ^{-2 },
 \end{equation}
  where $ r ( x_0 ) $ is the unique critical point of
 $ \alpha ( r)^2 r^{-2} $. We then have
\begin{equation}
\label{eq:propV}   V( x ) = - c_0 x^2 + \mathcal O ( x^3) , \ \ c_0 >  0, \ \ \  x V' ( x) < 0 , \ \ x \neq 0 .
\end{equation}
With $ \Omega = \Omega_{ \theta, \delta } $ of Lemma \ref{l:RWC}
we have
 $V \in \mathscr O (\Omega)$, $V|_{\mathbb R}$ is real valued,  and
\begin{equation}
\label{eq1.1}
V(x) \rightarrow - E_0, \quad \abs{x}\rightarrow \infty, \quad x\in \Omega, \ \ \
E_0 := \alpha ( r(x_0 )   )^2 r(x_0 ) ^{-2 } > 0 .
\end{equation}
We put
\begin{equation}
\label{eq4}
\Gamma_{\theta} = \left\{(1 + i\theta)x; x\in \mathbb R\right\}\subseteq \mathbb C, \quad 0 < \theta \ll 1,
\end{equation}
and notice that
\[
T^* \Gamma_{\theta} = \{\left((1 + i\theta)x, (1+i\theta)^{-1}\xi\right); (x,\xi) \in \mathbb R^2\} \subseteq \mathbb C^2.
\]
Letting
\[
p(x,\xi) = \xi^2 + V(x) , \quad (x,\xi) \in \Omega \times \mathbb C,
\]
we shall be concerned with ellipticity properties of the scaled symbol
\begin{equation}
\label{eq5}
p_{\theta}(x,\xi) = p|_{T^*\Gamma_{\theta}}  = ((1+i\theta)^{-1}\xi)^2 + V(x + i\theta x) = ((1-i\theta) \eta)^2 + V(x + i\theta x) ,
\end{equation}
for $(x,\xi) \in \mathbb R^2$, where
\begin{equation}
\label{eq5.1}
\eta := (1 + \theta^2)^{-1}\xi.
\end{equation}

\medskip
\noindent
To this end, we first make the following well known observation concerning the global ellipticity of $p_{\theta}$ near infinity. First, in the region where
\begin{equation}
\label{eq6}
x\in \mathbb R,\quad \abs{\xi} \geq 2 \left(1 + \|V\|_{L^{\infty}(\Omega)}\right)^{1/2},
\end{equation}
we have using (\ref{eq5}), (\ref{eq5.1}),
\begin{equation*}
\begin{split}
\abs{p_{\theta}(x,\xi)} & \geq \eta^2 - \| V \|_{L^{\infty}(\Omega)} 
\geq \tfrac12 {\xi^2} - \| V \|_{L^{\infty}(\Omega)} \geq \tfrac14 {\xi^2}
\\ & \geq  \tfrac{1}{8}\left(1 + \xi^2\right).
\end{split}
\end{equation*}
We next choose $ c_0 > 0$ independent of $ \theta $ so that
\begin{equation}
\label{eq9}
\abs{(1-i\theta)^2 \eta^2 - E_0} \geq c_0 {\theta} (1 + \xi^2), \quad \xi \in \mathbb R.
\end{equation}
This is clearly possible for $\abs{\xi} \geq 4(1 + E_0)$, and in the region where $\abs{\xi}\leq 4(1 + E_0)$, we write
\begin{equation*}
\begin{split}
\abs{(1-i\theta)^2 \eta^2 - E_0} & = \abs{1-i\theta}^2 \abs{\eta^2 - E_0 (1-i\theta )^{-2} }  \geq
\abs{\eta^2 - E_0 (1 + i\theta + \mathcal O(\theta^2))^2} \\
& = \abs{\eta^2 - E_0 - 2i\theta E_0 - \mathcal O(\theta^2)} \geq c_0 {\theta}, 
\end{split}
\end{equation*}
showing (\ref{eq9}).

Now, let $ M(\theta ) $ be chosen large enough so that
\[ \abs{x} > M (\theta) \Longrightarrow \abs{V(x + i\theta x) + E_0} \leq \tfrac12 c_0 \theta.\]
Then we have for $ \xi \in \mathbb R $ and $ |x| > M ( \theta ) $,
\begin{equation*}
\begin{split}
\abs{p_{\theta}(x,\xi)} & \geq \abs{(1-i\theta)^2 \eta^2 - E_0} - \abs{V(x+i\theta x) + E_0} \geq c_0 \theta \left(1 + \xi^2\right) - \abs{V(x+i\theta x) + E_0}
\\ &
\geq \tfrac12 c_0 {\theta}\left(1 + \xi^2\right).
\end{split}
\end{equation*}
 We conclude that there exist constants $L>0$, independent of $0 < \theta \ll 1$, and $M (\theta) > 0$, depending on $\theta$, such that
letting
\begin{equation*}
K(\theta):= \{ ( x, \xi )\in \mathbb R^2; \abs{x}\leq M(\theta), \,\,\abs{\xi}\leq L  \} ,
\end{equation*}
we have
\begin{equation*}
\abs{p_{\theta}(x,\xi)} \geq \tfrac12 c_0 \theta (1 + \xi^2), \quad (x,\xi) \in \mathbb R^2 \setminus K(\theta).
\end{equation*}

\bigskip
\noindent
We next claim that we can choose $0 < \theta \ll 1$ small enough so that
\begin{equation}
\label{eq13}
p_{\theta}(x,\xi) = 0, \,\,(x,\xi) \in K(\theta)  \Longrightarrow  (x, \xi) = (0, 0 ).
\end{equation}
To verify (\ref{eq13}), we use Taylor's formula for $x\in \mathbb R$, 
\begin{equation}
\label{eq14}
V(x+i\theta x) = V(x) + V'(x)i\theta x + \frac{1}{2}V''(x) (i\theta x)^2 + \frac{(i\theta x)^3}{2!} \int_0^1 (1-t)^2 V'''(x + it\theta x)\, dt,
\end{equation}
and therefore
\begin{equation}
\label{eq15}
\Im \left(V(x + i\theta x)\right) = \theta x V'(x) + \mathcal O(\theta^3 \abs{x}^3),
\end{equation}
where the implicit constant is independent of $\theta$. We get, combining (\ref{eq5}), (\ref{eq5.1}), and (\ref{eq15}),
\begin{equation}
\label{eq16}
\Im \left(p_{\theta}(x,\xi)\right) = -2\theta(1 + \mathcal O(\theta^2))\xi^2 + \theta x V'(x) + \mathcal O(\theta^3 \abs{x}^3).
\end{equation}
Here we know, in view of (\ref{eq:propV}), that
\begin{equation}
\label{eq17}
xV'(x) = -2c_0 x^2 + \mathcal O(\abs{x}^3),\quad x\rightarrow 0,
\end{equation}
and we get, using (\ref{eq16}), (\ref{eq17}),
\begin{equation}
\label{eq18}
{\rm Im}\, \left(p_{\theta}(x,\xi)\right) = -2\theta (1 + \mathcal O(\theta^2))\xi^2 -2\theta c_0 x^2 + \mathcal O(\theta \abs{x}^3), \quad x\rightarrow 0.
\end{equation}
It is therefore clear that there exists $\delta > 0$ small enough, independent of $0 < \theta \ll 1$, such that we have for $\abs{x} \leq \delta$, $\xi \in \mathbb R$,
\begin{equation}
\label{eq19}
\abs{p_{\theta}(x,\xi)} \geq \abs{{\rm Im}\, p_{\theta}(x,\xi)} \geq c_1 {\theta}(x^2 + \xi^2),
\end{equation}
for some $ c_1 > 0 $ independent of $ \theta $.

It only remains to check that no zero of $p_{\theta}$ can occur in the region $\delta \leq \abs{x} \leq M(\theta)$, $\abs{\xi} \leq L$. To this end assume, seeking a contradiction, that $p_{\theta}(x,\xi) = 0$ for some $(x,\xi)$ with $(x,\xi) \in K(\theta)$, $\abs{x} \geq \delta$. It follows that, using (\ref{eq14}) again,
\begin{equation}
\label{eq20}
\begin{split}
0 & = \Im \left(p_{\theta}(x,\xi)\right)
\\& = -2\theta(1 + \mathcal O(\theta^2))\xi^2 + \theta x V'(x) - \frac{\theta^3 x^3}{2!} \int_0^1 (1-t)^2 {\rm Re}\,\left(V'''(x + it\theta x)\right) dt.
\end{split}
\end{equation}
It follows from the Cauchy estimates that the last term on the right hand side of (\ref{eq20}) is $\mathcal O(\theta ^3)$, uniformly for $x\in \mathbb R$, $\theta > 0 $, small enough, and we get
\begin{equation}
\label{eq21}
2\theta(1 + \mathcal O(\theta^2))\xi^2 = \theta x V'(x) + \mathcal O(\theta^3) \leq \mathcal O(\theta^3) \Longrightarrow \abs{\xi} = \mathcal O(\theta).
\end{equation}
Here we have also used (\ref{eq:propV}). We next write, using (\ref{eq5}) and (\ref{eq14}),
\begin{equation}
\label{eq22}
0 = \Re \left(p_{\theta}(x,\xi)\right) = (1-\theta^2)\eta^2 + V(x) - \tfrac{1}{2} V''(x) (\theta x)^2 + \mathcal O(\theta^3),
\end{equation}
and using the Cauchy estimates again, we get
\begin{equation}
\label{eq23}
0 = \Re \left(p_{\theta}(x,\xi)\right) = (1-\theta^2)\eta^2 + V(x) + \mathcal O(\theta^2).
\end{equation}
We obtain, using (\ref{eq5.1}), (\ref{eq21}), and (\ref{eq23}),
\begin{equation}
\label{eq24}
 - V(x) = (1-\theta^2)\eta^2 + \mathcal O(\theta^2) = \mathcal O(\theta^2).
\end{equation}
Here we recall that $\abs{x} \geq \delta$, for some $\delta > 0$ independent of $\theta$. It follows that
\begin{equation}
\label{eq25}
0 < \inf_{\abs{y}\geq \delta} \left( - V(y)\right) \leq \mathcal O(\theta^2),
\end{equation}
which gives a contradiction for $\theta > 0$ small enough.

We proved the following result. 
\begin{lemm}
\label{l:ell}
Let $V$ be given in \eqref{eq:defV1}
There exists $0 < \theta \ll 1$ small enough such that the scaled symbol
\begin{equation}
\label{eq26}
p_{\theta}(x,\xi) = ((1+i\theta)^{-1}\xi)^2 + V(x + i\theta x) , \quad (x,\xi) \in \mathbb R^2,
\end{equation}
satisfies
\[ p_{\theta}(0,0) = 0, \ \ \ dp_{\theta}(0,0) = 0, \ \ - \Im p_{\theta}''(0,0)\, \text{   positive definite}, \]
and for each $\varepsilon > 0$ there exists $\delta > 0$ such that
\begin{equation}
\label{eq27}
\abs{(x,\xi)} > \varepsilon \Longrightarrow \abs{p_{\theta}(x,\xi)} \geq \delta (1 + \xi^2), \quad (x,\xi) \in \mathbb R^2.
\end{equation}
\end{lemm}

\subsection{Resonance free regions}
\label{s:free}
The quasinormal modes $ \lambda_{\ell, n} $ in Theorem \ref{t:2} are given by the square roots of the eigenvalues of $ P_\theta $,
\begin{equation}
\label{eq:Ptheta}
 P_\theta = P|_{\Gamma_\theta} ,   \ \ \ P :=(hD_x)^2 + W ( x, h ) , \ \ \ h = ( \ell + \tfrac12)^{-1},
 \end{equation}
multiplied by $h^{-1}$, see Lemma \ref{l:RWC} and \S \ref{s:cs}. We now explain why to obtain the theorem we only need find eigenvalues in
the region $ \{ | z- W_0 (x_0 ) | < \delta \}$, for $\delta > 0$ sufficiently small but fixed, where $ x_0$ is the critical point of $ W_0 $.
Suppose that $ 0 < E $ satisfies $ | W_0 ( x_0 ) - E | > \delta $. Then, as $ W_0 $ has a unique maximum at $x_0$, the energy level $ E $ is {\em non-trapping} for $ p := \xi^2 + W_0 (x ) $ in the sense that $ p^{-1} ( E ) \subset T^* \mathbb R $ is connected and unbounded.
Since $ W_0 $ is globally analytic it follows that
\begin{equation}
\label{eq:resfr}
\exists \,  \varepsilon_0 , h_0 \ \ \
\Spec ( P_\theta ) \cap \{ z:  |z - E | < \varepsilon_0  \} = \emptyset , \ \ \  0 < h < h_0 ,
\end{equation}
see \cite[\S 12.5]{Sj02} for a self-contained presentation and references. Taking the square root and rescaling by $ h =1/( \ell + \tfrac12)$ implies that for any (small) $ c_0 , \delta_0 > 0$  there exists $ t_0 $ such that
\[  \begin{gathered}
 {\rm QNM}_\ell ( m , \Lambda ) \cap U (c_0, \delta_0, t_0 ) = \emptyset, \ \ \  E_0 = (W_0 ( x_0 ))^{1/2}, \\
U(c_0, \delta_0, t_0 ) :=    \{ \lambda : \Re \lambda \in
 [ c_0 \ell, ( E_0 - \delta_0 ) \ell ] \cup [  ( E_0 + \delta_0 ) \ell , \ell /c_0 ] , \ \Im
\lambda_0 > - t_0 \Re \lambda \} .\end{gathered}  \]
To see that
\begin{equation}
\label{eq:c0} {\rm QNM}_\ell ( m , \Lambda ) \cap \{ \lambda :  \Re \lambda \in
  [   \ell /c_0 , \infty ) , \ \Im
\lambda_0 > - t_0 \Re \lambda \} = \emptyset , \end{equation}
we need a different choice of $ h $ in the passage from the Regge--Wheeler equation
$ D_x^2  + V_\ell ( x ) $ (see \eqref{eq:VRW}) to $ (hD_x)^2 + W ( x, h ) $. For that we
write
\[ P - z =  h^2 ( D_x^2 + V_\ell ( x ) - \lambda^2 ) = ( hD_x)^2 + h^2 ( \ell + \tfrac12)^2 W_0 ( x ) +
h^2 W_1 ( x ) - z  , \ \ \ z := (\lambda h)^2 . \]
If $ h < (\ell + \frac12 )^{-1} \sqrt E_0 ( 1 - \delta)$, $ \delta > 0$, then the energy level $
z = 1 $ is non-trapping for $ P$. Hence, as in \eqref{eq:resfr}, there exists $ h_0 $
and $ \varepsilon_0 $, such that if
$ h < \ell + \frac12 )^{-1} E_0 ( 1 - \delta)$. $ \ell \gg 1 $,  then $ P$ has no resonances $ z $
(elements of the spectrum of $ P_\theta $) for $ |z -1 | < \varepsilon_0 $. We now take
$ 1/h = \Re \lambda > ( \ell + \frac12) / \sqrt (E_0(1-\delta)) $ and that give \eqref{eq:c0}.

This shows that to obtain Theorem \ref{t:2} it is sufficient to describe the eigenvalues of
$ P_\theta $ near $ W_0 ( x_0 ) $, the maximal value of $ W_0$.

When $ \Lambda > 0 $ we can replace $ c_0 $ by $ 0 $ and then Theorem \ref{t:2} implies
\eqref{eq:defN} in Theorem \ref{t:1}. This follows from \cite[Proposition 4.4, (4.9)]{saz} and
the fact that for $ \Lambda > 0 $ quasinormal modes form a discrete set in $ \mathbb C $
(see \cite[\S 2]{saz} and also the careful presentation in \cite{boha}).

\noindent
{\bf Remark.} The same conclusion should be valid for the case of $ \Lambda = 0 $, but as far as the authors are aware, there is no mathematical argument for the discreteness of $ {\rm QNM} ( m , 0 ) $ in $ \{ \lambda : \Im \lambda >  - t_0 | \Re \lambda | \} $ and for
\[  {\rm QNM }_\ell ( m, 0 ) \cap \{ | \Re \lambda | < c_0 \ell , \Im \lambda >  -t_0 \Re \lambda \} =
\emptyset, \ \ \  0 < c_0, t_0 \ll 1 . \]
The latter would be needed to obtain \eqref{eq:defN} for $ \Lambda = 0 $.

\section{Analytic symbols}
\label{s:anas}

A (formal) classical analytic symbol in an open set $ \Omega \subset \mathbb C^d $ is a (formal) expression
\begin{equation}
\label{eq:classs}
\begin{gathered}
a ( z; h ) := \sum_{ k = 0 }^\infty h^k a_k ( z ) , \ \  a_k \in \mathscr O ( \Omega ) , \\
 \forall \, K \Subset \Omega \,\,
\exists \, C = C ( K ) \  \ | a_k ( z ) | \leq C^{ k+1} k^k , \ z \in K,\,\,k= 0,1,2,\ldots.
\end{gathered}
\end{equation}
Since all the symbols we use are classical, from now on we will drop that specification.

For open $ \Omega_1 \Subset \Omega $ we have a {\em realization} of $ a (z; h ) $ on $\Omega_1$ given by
the following holomorphic function:
\begin{equation}
\label{eq:reali}
a_{\Omega_1 } ( z; h ) := \sum_{k=0}^{ [( e C( \Omega_1) h )^{-1} ]} a_k ( z ) h^k.
\end{equation}
If $\Omega_1 \subset \Omega_2 \Subset \Omega$, we have, assuming, as we may, that $C(\Omega_1) < C(\Omega_2)$,
\[
\abs{\ a_{\Omega_1}(z;h) - a_{ \Omega_2 }(z;h)} \leq C(\Omega_1) \frac{e}{e-1} e^{-1/eC(\Omega_2)h},\,\,z\in \Omega_1.
\]
(See \eqref{eq:Borel} for another way to {\em realize} $ a $ as a holomorphic function.)
We refer to  \cite[\S 1]{sam} or \cite[\S 2.2]{HiSj15} for a detailed account.

For $ a = a ( x, \xi; h) $ and $ b = b ( x, \xi; h) $,     analytic symbols in $ \Omega
\subset \mathbb C_x^n \times \mathbb C_\xi^n $ we define
\begin{equation}
\label{eq:prod}  a \# b ( x, \xi; h ) := \sum_{\alpha \in \mathbb N^n } \frac{1} { \alpha!} \partial_\xi^{\alpha} a ( x, \xi; h )
( h D_x )^\alpha b ( x, \xi; h ) . \end{equation}
Then $ a \# b $ is an analytic symbol in $ \Omega $.
We now recall the fundamental result of Boutet de Monvel--Kr\'ee: if
$ a \neq 0 $ on $ \Omega $ and $ \Omega_0 \Subset \Omega $ is an open set, then
the formal   symbol $ b $ defined by
\begin{equation}
\label{eq:BdMK}
a \# b = 1
\end{equation}
is a formal   analytic symbol in $ \Omega_0 $.
The proof given by Sj\"ostrand in \cite{sam} uses norms on formal pseudodifferential operators
\[ \hat a  = a ( x,  h D_x; h ) , \  \ \ \ a  (x, \xi; h ) = \sum_{ \ell=0}^\infty a_\ell ( x, \xi ) h^\ell , \ \ \ \
 \hat a  \circ \hat b = \widehat { a \# b } . \]
To recall the definition of the norms we assume for simplicity that
$ \Omega = \{ (x, \xi )  : x_j , \xi_j \in \mathbb C, |x_j| < r_j, |\xi_j | < \rho_j  \} \subset \mathbb C^{2n}
$ is an open polydisc. We then define a family of polydiscs by
\begin{equation}
\label{eq:polydisc}
\Omega (t) := \{ (x, \xi )  : x_j , \xi_j \in \mathbb C, |x_j| < r_j-t, |\xi_j | < \rho_j-t  \}  , \ \ \
t < t_0 :=  \min_j \min (r_j, \rho_j ) .
\end{equation}
Then $ \Omega (t)  \subset \Omega ( s ) $ if $ t > s $ and by the Cauchy estimates,
\[
\sup_{ \Omega(t) } |  D^\alpha u | \leq \frac{\alpha!}{( t -s )^{|\alpha| }}\,\sup_{ \Omega(s) } | u |, \ \ \
s< t < t_0 , \ \ \ u \in (\mathscr O \cap L^{\infty})(\Omega_s) . \]
This allows the following definition for $\rho > 0$:
\begin{equation}
\label{eq:norhm}  \begin{gathered} \| \hat a \|_\rho = \sum_{k=0}^\infty \rho^k f_k ( \hat a ) , \ \ \ \ \ \
f_k ( \hat a ) := \sup_{ 0 \leq s < t < t_0} \frac{ ( t - s)^k \| A_k \|_{t,s} }{ k^k } , \\
 A_k ( x, \xi; D_x ) := \sum_{ | \alpha | + \ell = k } \frac{1}{ \alpha!} \partial_\xi^\alpha a_\ell ( x, \xi ) D_x^\alpha , \ \ \ \ \
 \| B \|_{t, s} = \sup_{0\neq u \in (\mathscr O \cap L^{\infty}) ( \Omega_s )} \frac{ \sup_{ \Omega(t) } | B u  |}{ \sup_
 {\Omega(s) } |u| } .
 \end{gathered}   \end{equation}
 (Formally, $ a ( x, \xi + h D_x;h) = \sum_{ k=0}^\infty h^k A_k(x,\xi;D_x) $.)

The key property is given by
\begin{equation}
\label{eq:key}   \| \hat a  \circ \hat b \|_\rho \leq \| \hat a \|_\rho \| \hat b \|_{\rho }.
\end{equation}
One immediate consequence is that the formal symbol $ b $ defined by
$  \hat b =\exp{\hat a} $ is an analytic symbol.

When working with the Weyl quantization of analytic symbols, rather than (\ref{eq:prod}), we define
\begin{equation}
\label{eq:new1}
a \# b ( x, \xi; h ) := \left(e^{\frac{ih}{2}\sigma(D_x,D_{\xi}; D_y, D_{\eta})} (a(x,\xi;h)\, b(y,\eta;h))\right)|_{y=x, \eta = \xi},
\end{equation}
where
\[
\sigma(D_x,D_{\xi}; D_y, D_{\eta}) = D_{\xi} \cdot D_y - D_{\eta} \cdot D_x.
\]
The formal differential operator of infinite order associated to $a$ is then given by
\begin{multline}
\label{eq:new2}
A(x,\xi,D_x,D_{\xi};h) u(x,\xi) = {\rm Op}_a(x,\xi,D_x,D_{\xi};h)u(x,\xi) \\
= \sum_{k=0}^{\infty} \frac{h^k}{k!} \left(\left(\frac{i}{2} \sigma(D_x,D_{\xi}; D_y, D_{\eta})\right)^k a(x,\xi;h) u(y,\eta)\right)\biggl |_{y=x, \eta = \xi},
\end{multline}
and we have $A(x,\xi,D_x,D_{\xi};h) = \sum_{k=0}^{\infty} h^k A_k(x,\xi, D_x,D_{\xi})$, where
\begin{equation}
\label{eq:new3}
A_k(x,\xi, D_x,D_{\xi}) = \sum_{\abs{\alpha + \beta} + \ell = k} \frac{(-1)^{\abs{\beta}}}{(2i)^{\abs{\alpha + \beta}} \alpha! \beta!}
\partial_{\xi}^{\alpha} \partial_x^{\beta} a_{\ell}(x,\xi) \partial_x^{\alpha} \partial_{\xi}^{\beta}.
\end{equation}
We have the corresponding definition of $\norm{\hat{a}}_{\rho}$, analogous to (\ref{eq:norhm}), and (\ref{eq:key}) still holds. See also~\cite[Section 3]{HitSj3a}. As observed in~\cite[Appendix A]{HitSj3a}, a key advantage of working with the Weyl quantization in this context is due to its metaplectic invariance. Let $\kappa: \mathbb C^{2n} \rightarrow \mathbb C^{2n}$ be a complex linear canonical transformation and let $\kappa^*$, $\kappa_*$ stand for the operations of pull-back and push-forward of functions on $\mathbb C^{2n}$, respectively. It follows from (\ref{eq:new1}) that $(\kappa_* a) \# (\kappa_* b) = \kappa_* (a\#b)$, and therefore using (\ref{eq:new2}) we get
\[
\kappa_* \circ {\rm Op}_a \circ \kappa^* = {\rm Op}_{\kappa_* a}.
\]
This implies that if the polydiscs $\Omega(t)$ in (\ref{eq:polydisc}) are invariant under $\kappa$, then we have
\begin{equation}
\label{eq:new4}
\norm{\widehat{\kappa_* a}}_{\rho} = \norm{\hat{a}}_{\rho}.
\end{equation}

We also use the following technical result of Hitrik--Sj\"ostrand \cite[Proposition 3.2]{HitSj3a}, which holds both for the classical and the Weyl quantization. (We stress that the two page proof can be read independently of the rest of the paper and depends only on the definitions
reviewed above.)
\begin{lemm}
\label{l:hitsj}
For analytic symbols $ a $ and $ b$, there exists a constant $ C ( b ) $ (independent of  $h$ and
$ \rho $) such that
\[   \| [ \hat a , \hat b ] \|_{\rho} \leq C (b) \rho \| \hat a \|_\rho . \]
\end{lemm}
We need the following application of this lemma, where we write $\ad_{\hat a} \hat q = [\hat a, \hat q]$:
\begin{equation}
\label{eq:hitsj}   \| \ad_{\hat a }^p \hat q \|_\rho \leq
\rho 2^p C(q )  \| \hat a\|_{\rho}^p . \end{equation}
\begin{proof}[Proof of \eqref{eq:hitsj}]
We proceed by induction on $ p $, with $ p = 1 $ given by the lemma. Then
\[ \| \ad_{\hat a }^{p+1} \hat q \|_\rho = \| \hat a \circ \ad_{\hat a }^p \hat q -
\ad_{\hat a }^p \hat q  \circ \hat a \| \leq 2 \| \hat a \|_{\rho}
\| \ad^p_{\hat a } \hat q \|_\rho \leq 2^{p+1} C( q) \rho \| a \|_{\rho}^{p+1}, \]
completing the proof. \end{proof}

The key result we need is an analytic symbol version of an averaging result.
We present it in a special case needed here:

\begin{prop}
\label{p:ave}
Suppose that $ q = q ( z, \zeta; h ) $ is an analytic symbol in a neighbourhood of the origin in $\mathbb C^2$,
\begin{equation}
\label{eq:defq}  q = \sum_{ \ell=0}^\infty h^\ell q_\ell, \ \ \  q_0 ( z, \zeta ) = g ( z \zeta ) , \ \ \
q_1 ( z, \zeta ) = g_1 ( z\zeta ), \end{equation}
where $g(0) = 0$, $g'(0) \neq 0$. Then there exist analytic symbols $ a ( z, \zeta; h ) $ and
\[  G ( w; h ) = \sum_{ \ell = 0}^\infty h^\ell G_\ell ( w ) , \ \ \  G_0 ( w ) = g( w ) , \]
such that, in the sense of formal asymptotic expansions using Weyl quantization, we have
\begin{equation}
\label{eq:conj}
e^{\hat a } \circ \hat q \circ e^{-\hat a } = \hat b, \ \ \ b ( z, \zeta; h ) = G ( z\zeta; h ) .
\end{equation}
\end{prop}

\noindent
{\bf Remark.} The assumption on $ q_1 $ in \eqref{eq:defq} is not a serious restriction. Using Lemma \ref{l:ave} below, as well as (\ref{eq:Ad}),   we can can find an analytic $h$--independent symbol $ a_0 $ such that
\begin{equation}
\label{eq:remarque}
e^{\hat a_0 } \circ \hat q \circ e^{-\hat a_0 } = \hat b_0 +
h^2 \hat b_1, \ \  b_0 ( z, \zeta; h ) = g ( z\zeta ) + h g_1 ( z \zeta  ) 
,
\end{equation}
where $ b_1 $ is an analytic symbol. The point of Proposition \ref{p:ave} is that this can be iterated
to infinite order with analytic symbols. 

The proof uses the following lemma. We note that if $  g ( w ) = w $, then \eqref{eq:defa0} below is of the same form as
\begin{equation}
\label{eq:pois1}
\{ g ( z \zeta ), b ( z, \zeta ) \} = a ( z, \zeta ) + a_0 ( z\zeta ) ,
\end{equation}
where $ \{ f, g \} := \partial_\zeta f \partial_z g - \partial_z f \partial_\zeta g $ is the Poisson
bracket with respect to the complex symplectic form on $\mathbb C^2$,
\[  \sigma :=  d \zeta \wedge d z . \]

\begin{lemm}
\label{l:ave}
Let $ \Omega ( \delta ) = D ( 0 , \delta ) \times D ( 0 , \delta ) \subset \mathbb C^{2} $ and
suppose that $ r \in \mathscr O (\Omega ) $, $ \Omega ( \delta ) \Subset \Omega \subseteq \mathbb C^2$. Then
there exist $ a \in \mathscr O  ( \Omega( \delta ) ) $
and $ r_0 \in \mathscr O ( D ( 0 , \delta^2 ) )$ such that for an absolute constant $ K $, we have
\begin{equation}
\label{eq:defa0}  \begin{gathered}  i   ( z \partial_z - \zeta \partial _\zeta ) a ( z, \zeta ) = - r ( z, \zeta ) + \langle r \rangle ( z,  \zeta ) , \ \ \ \langle r \rangle ( z, \zeta ) = r_0 ( z \zeta ) \\
\sup_{ D ( 0, \delta^2 ) } |r_0| \leq K \sup_{ \Omega ( \delta ) } |r| , \ \
\  \sup_{ \Omega( \delta)  } | a | \leq  K  \sup_{ \Omega ( \delta ) } |r|.
\end{gathered} \end{equation}
\end{lemm}
\begin{proof}
We note that by scaling, we can assume that $ \delta = 1 $ and that we can write
$$  i r ( z, \zeta ) = \sum_{ m,n \in \mathbb N } r_{mn} z^m \zeta^n ,
\ \ \ \  |r_{mn} | \leq \sup_{\Omega ( 1 ) } |{r}|.
$$
It follows that we can take
\[ \begin{gathered}    a ( z, \zeta ) := \sum_{ m \neq n } r_{mn} ( m - n )^{-1} z^m \zeta^n , \ \ \
i r_0 ( w ) := \sum_{ m \in \mathbb N} r_{mm} w^m ,
\\  | a ( z , \zeta ) | \leq \sup_{\Omega ( 1 ) } |r| ( 1 - |z|)^{-1} ( 1 - |\zeta |)^{-1} , \ \ \
 | r_0 ( w ) | \leq \sup_{\Omega ( 1 ) } |r| ( 1 - | w| )^{-1} .
\end{gathered}
\]
This shows that solutions $a$ and $r_0$ exist in  $ \Omega ( 1 ) $ and $ D (0,1 ) $, respectively, and it only remains for us to verify that the bounds given in (\ref{eq:defa0}) hold. To this end, introducing polar coordinates and writing $ z = r e^{i \theta } $, $ \zeta = \rho e^{ i \varphi } $, $r , \rho < 1 $, we see that (\ref{eq:defa0}) takes the form
\[
\left(\frac{i(r\partial_r - \rho \partial_{\rho}) + (\partial_\theta - \partial_\varphi)}{2}\right)A = -R + R_0,
\]
with the Fourier series expansions
\[ \begin{gathered}
iR(r,\rho, \theta, \varphi ) = \sum_{m,n\in \mathbb N} r_{mn} r^m \rho^n e^{ i (m \theta + n \varphi ) } , \ \
A(r,\rho, \theta, \varphi ) = \sum_{m\neq n} \frac{r_{mn}}{m-n} r^m \rho^n e^{ i (m \theta + n \varphi) } , \\
iR_0(r,\rho, \theta, \varphi )  = \sum_{m\in \mathbb N} r_{mm} r^m \rho^m e^{ im ( \theta + \varphi )} .
\end{gathered} \]
It is now easy to check that $ A $ and $ R_0 $ have the following integral representations
\begin{equation}
\label{eq:new5}
A (r,\rho, \theta, \varphi ) = -\frac{1}{2\pi} \int_0^{2 \pi} (t-\pi) R (r,\rho, \theta + t , \varphi - t ) dt ,
\end{equation}
\begin{equation}
\label{eq:new6}
R_0 (r,\rho, \theta, \varphi ) = \frac{1}{2 \pi} \int_0^{2\pi}  R (r,\rho, \theta + t , \varphi - t ) dt.
\end{equation}
This gives the uniform bounds
\[  \sup_{T(r, \rho )} | a (z, \zeta ) | \leq K \sup_{ \Omega ( 1 ) } |r|
 \ \  \sup_{T(r, \rho ) } | r_0(z \zeta ) | \leq K \sup_{ \Omega ( 1 ) } |r|, \ \ \
 0 < r < 1, \ 0 < \rho < 1 . \]
By continuity, the bounds remain valid for $ 0 \leq r, \rho < 1 $, and \eqref{eq:defa0} follows.
\end{proof}

\noindent
{\bf Remark.} For future reference, it will be convenient to rewrite (\ref{eq:new5}), (\ref{eq:new6}) as follows,
\begin{equation}
\label{eq:new7}
r_0 = \langle{r\rangle} = \frac{1}{2\pi} \int_0^{2\pi} r\circ \exp(itH_{z\zeta})\, dt,
\end{equation}
\begin{equation}
\label{eq:new7.1}
a = -\frac{1}{2\pi} \int_0^{2\pi} (t-\pi)\, r\circ \exp(itH_{z\zeta})\, dt = -\frac{1}{2\pi} \int_0^{2\pi} (t-\pi)\, (r - r_0) \circ \exp(itH_{z\zeta})\, dt.
\end{equation}
Here $H_{z\zeta} = z\partial_z - \zeta \partial_{\zeta}$ is the Hamilton vector field of $z\zeta$.

\begin{proof}[Proof of Proposition {\rm \ref{p:ave}}]
We modify (with a simpler goal) the proof in  \cite[\S 3]{HitSj3a}, \cite[Appendix A]{HitSj3a}. We first recall standard notation:
\begin{equation}
\label{eq:Ad}
\Ad_{\hat a } \hat q := e^{ \hat a }\circ \hat q \circ e^{- \hat a }, \ \ \
\ad_{\hat a } \hat q := [ \hat a , \hat q ], \ \ \
\Ad_{\hat a } \hat q = \sum_{ k=0}^{\infty} \frac{1}{k!} \ad_{\hat a }^k \hat q  ,
\end{equation}
where all the expressions are understood as compositions of formal pseudodifferential
operators quantizing formal analytic symbols, using the Weyl quantization.

We start by reducing the problem to the case of $g(\tau) = \tau$ in (\ref{eq:defq}). For that we find a holomorphic function $ f = g^{-1}$ near $ 0 $. We then define an analytic symbol
\begin{equation}
\label{eq:new8}
q_{\rm{new}}  = f ( 0 ) + f' ( 0 ) q + \tfrac12 f''(0) q \# q + \cdots,
\end{equation}
where, as formal analytic pseudodifferential operators, $ f ( \hat q ) = \widehat{ q_{\rm{new}}} $, and
\begin{equation}
\label{eq:new9}
q_{\rm{new} } ( z, \zeta ) = z \zeta + h q_{{\rm{new}}, 1} ( z, \zeta ) + h^2 q_{{\rm{new}}, 2} ( z, \zeta ) +
\cdots.
\end{equation}
In what follows, when proving \eqref{eq:conj}, we shall assume therefore that
\begin{equation}
\label{eq:q02zz}
q_0 ( z, \zeta ) =  z \zeta .
\end{equation}
It will be easy to return to the original $\hat q$, since we have, in the sense of asymptotic expansions,
\[ f ( Ad_{\hat a} \hat q ) = \Ad_{\hat a} f ( \hat q ).
\]
Performing a preliminary conjugation, as explained in the remark after Proposition \ref{p:ave}, we may also assume that $q_{{\rm{new}}, 1} ( z, \zeta ) = g_{{\rm{new}},1} ( z \zeta )$ in (\ref{eq:new9}).

We now have the following consequence of \eqref{eq:hitsj} in which we consider $ q $ as fixed and have all the constants depending on $ q $,
\begin{equation}
\label{eq:abc}
\begin{gathered}  \Ad_{ \hat b_0 + \hat b_1 } \hat q =
\Ad_{ \hat b_0 } \hat q + \ad_{ \hat b_1} \hat q  + R ( \hat b_0, \hat b_1 ) , \\
\| R ( \hat b_0, \hat b_1 ) \|_\rho \leq C \rho\,\| \hat b_1 \|_\rho  \max(\| \hat b_0 \|_\rho, \| \hat b_1 \|_\rho ) e^{ 4
\max(\| \hat b_0\|_\rho , \|\hat b_1\|_\rho)}.
\end{gathered}
\end{equation}
To see this we put $ J ( k) = \{ 0 , 1 \}^k \setminus (0, \ldots, 0 ) $ and write, using (\ref{eq:Ad}),
\begin{equation}
\label{eq:R}
R ( \hat b_0 , \hat b_1 ) = \sum_{ k=2}^{\infty} \frac1 {k!} \sum_{ j=(j_1,\ldots ,j_k) \in J(k)   }
\ad_{ \hat b_{j_1} } \cdots \ad_{ \hat b_{ j_k }}  \hat q .
\end{equation}
After $k-1$ applications of the inequality
\[
\norm{{\rm ad}_{\hat a}\, \hat b}_{\rho} = \norm{\hat a \circ \hat b - \hat b \circ \hat a}_{\rho} \leq 2\norm{\hat a}_{\rho} \norm{\hat b}_{\rho}
\]
and one application of Lemma \ref{l:hitsj}, we see that
\begin{equation}
\label{eq:ad}
\|  \ad_{ \hat b_{j_1} } \cdots \ad_{ \hat b_{ j_k } } \hat q \|_\rho \leq 2^{k-1} \|\hat b_{j_1}\|_{\rho} \cdots \|\hat b_{j_{k-1}}\|_{\rho}\left(C(q)\rho \|\hat b_{j_k}\|_{\rho}\right),
\end{equation}
and hence, using (\ref{eq:R}) and (\ref{eq:ad}) we get
\[ \begin{split} \| R ( \hat b_0 , \hat b_1 ) \|_\rho & \leq C(q) \rho \sum_{ k = 2}^\infty
\frac{2^{k-1}}{k!}  \sum_{ j=(j_1,\ldots ,j_k) \in J(k)} \|\hat b_{j_1}\|_{\rho} \cdots \|\hat b_{j_{k}}\|_{\rho} \\
&  \leq C(q) \rho \sum_{ k = 2}^\infty \frac{2^{k-1}}{k!} \left(2^k -1\right) \| \hat b_1 \|_\rho (\max( \| \hat b_0 \|_\rho , \| \hat b_1 \|_\rho))^{k-1}.
\end{split}
\]
This implies \eqref{eq:abc}, since
\[
\sum_{k=2}^{\infty} \frac{2^{k-1}}{k!} \left(2^k -1\right) \lambda^{k-1} \leq \sum_{k=2}^{\infty} \frac{4^k}{(k-2)!} \lambda^{k-1} = 16\lambda e^{4\lambda},\quad \lambda \geq 0
\]

Given an analytic symbol $r$, it follows from \eqref{eq:q02zz} and Lemma \ref{l:ave} that there exists an analytic symbol $a$, given in (\ref{eq:new7.1}), such that
\begin{equation}
\label{eq:homol}
\frac{1}{i} \{ a , q_0 \} = - r + \langle r \rangle  , \ \ \langle{r\rangle} ( z, \zeta ) =r_0  ( z \zeta).
\end{equation}
On the level of operators, we shall use the notation $\hat a = \mathcal L ( \hat r ) = \mathcal L(\hat r - \hat r_0)$, and it follows from (\ref{eq:new4}), (\ref{eq:new7}), (\ref{eq:new7.1}) and the triangle inequality for the $\rho$--norm that we have
\begin{equation}
\label{eq:Lbound}
\| \mathcal L ( \hat r ) \|_\rho \leq \left(\frac{1}{2\pi}\int_0^{2\pi} \abs{t-\pi}\, dt\right) \| \hat r - \hat r_0 \|_\rho = \frac{\pi}{2} \| \hat r - \hat r_0 \|_\rho, \quad \| \hat r_0 \|_\rho \leq \| \hat r \|_\rho.
\end{equation}
Since $ q_0 $ is quadratic, we have an exact formula
\begin{equation}
\label{eq:defL}
\ad_{ \mathcal L (\hat r ) } \hat q_0 = \frac{h}{i} \widehat{\{a,q_0\}} = - h \hat r + h \langle \hat r \rangle .
\end{equation}
Here and below, we simplify the notation by writing $ \langle \hat r \rangle $ for $ \widehat {\langle r \rangle }$.

In the notation of \eqref{eq:defq} we write
\[
q = q_0 + hr_0,
\]
where $r_0$ is an analytic symbol such that $r_0 - \langle{r_0\rangle} = \mathcal O(h)$ in the sense of usual symbols, and let us define
\[
\hat a_0 := \mathcal L ( \hat r_0 ) , \ \ \    \| \hat a_0  \|_\rho \leq  \frac{\pi}{2} \| \hat r_0 - \langle{\hat r_0 \rangle}\|_\rho.
\]
Here we have used (\ref{eq:Lbound}). We have
\[
\left\|(\hat r_0 - \langle{\hat r_0\rangle})/h\right\|_{\rho_0} < \infty,
\]
for some $\rho_0 > 0$. Hence, for $0 < \rho \leq \rho_0$,
\begin{equation}
\label{eq:a_0}
\| \hat a_0\|_{\rho} \leq \tfrac{1}{2}\pi h\,\left \|(\hat r_0 - \langle{\hat r_0\rangle})/{h}\right \|_\rho \leq K h,\quad ,\quad K := \tfrac{1}{2}\pi \, \left \|(\hat r_0 - \langle{\hat r_0\rangle})/{h}\right \|_{\rho_0}.
\end{equation}
Here following~\cite[\S 3]{HitSj3a}, we view $h$ as an independent parameter. We then use \eqref{eq:abc} with $ b_0  = 0 $, $ b_1 = a_0 $, as well as (\ref{eq:defL}), to obtain
\[  \begin{split}
\Ad_{ \hat a_0 }  \hat q &  = \hat q + \ad_{\hat a_0  } \hat q + R ( 0, \hat a_0 )
= \hat q_0 + h \hat r_0 + \ad_{\hat a_0} \hat q_0 + \ad_{\hat a_0}(\hat q - \hat q_0) + R ( 0 , \hat a_0 ) \\& =
\hat q_0 + h \langle \hat r_0\rangle +h  \hat r_1.
\end{split} \]
Here $h\hat r_1 = \ad_{\hat a_0}(\hat q - \hat q_0) + R ( 0 , \hat a_0 )$, so that
\begin{equation}
\label{eq:r1}
\hat r_1 = \ad_{\hat a_0} \hat r_0 + \frac{1}{h} R ( 0 , \hat a_0 ).
\end{equation}
We get, using Lemma~\ref{l:hitsj}, (\ref{eq:abc}), and (\ref{eq:a_0}),
\[
\|\hat r_1 \|_\rho \leq C\rho \| \hat a_0 \|_\rho + \frac{1}{h} C \rho \|  \hat a_0 \|_\rho^2 e^{ 4 \|   \hat a_0  \|_\rho } \leq B\rho h.
\]
Here $ B $ is a large constant depending on $q$, which will be chosen later to close the argument.

Arguing inductively, assume that we found $ a_\ell $, $ 0 \leq \ell \leq j $, such that
\begin{equation}
\label{eq:aell}
\begin{gathered}
\Ad_{ \hat a_0 + \cdots + \hat a_j } \hat q  = \hat q_0 + h \sum_{\ell = 0}^{j} \langle \hat r_\ell  \rangle  + h  \hat r_{j + 1}, \\
\| \hat a_\ell \|_{\rho} \leq K B^{\ell}  \rho^{\ell} h, \ \ \ \| \hat r_{\ell + 1 }  \|_\rho \leq B^{\ell+1} \rho^{\ell+1} h , \ \ \
0 \leq \ell \leq j.
 \end{gathered}
 \end{equation}
We now define, assuming, as we may, that the constant $K$ in (\ref{eq:a_0}) satisfies $\pi \leq K$,
\begin{equation}
\label{eq:aj+1}
\hat a_{j + 1} = \mathcal L ( \hat r_{j + 1} ) , \ \ \ \| \hat a_{j + 1} \|_\rho \leq \pi \| \hat r_{j + 1} \|_\rho \leq K B^{j + 1} \rho^{j+1}  h.
\end{equation}
Here we have used (\ref{eq:Lbound}). We get therefore, by an application of (\ref{eq:abc}), and using (\ref{eq:aell}),
\[ \begin{split}
\Ad_{ \hat a_0 + \cdots + \hat a_{j+1}  } \hat q & = \Ad_{ \hat a_0 + \cdots + \hat a_j } \hat q
 + \ad_{\hat a_{j+1}} \hat q + R (\hat a_0 + \cdots + \hat a_{j} , \hat a_{j+1} )\\
  &= \hat q_0 + h \sum_{\ell = 0}^{j+1} \langle \hat r_\ell  \rangle + h \hat r_{j+2}, \ \
\end{split} \]
where, similarly to (\ref{eq:r1}),
\[
\hat r_{j+2} = \ad_{\hat a_{j+1}} \hat r_0 + \frac{1}{h} R (\hat a_0 + \cdots + \hat a_{j} , \hat  a_{j+1} ).
\]
Using Lemma \ref{l:hitsj} \eqref{eq:abc}, (\ref{eq:aell}) and (\ref{eq:aj+1}), we obtain that
\[ \begin{split}
\| \hat r_{j+2} \|_\rho & \leq C\rho \| \hat a_{j+1}\|_{\rho} + \frac{1}{h} C \rho \| \hat a_{j+1} \|_\rho
\left( \sum_{\ell=0}^{j+1} \| \hat a_\ell \|_\rho \right) e^{ 4 \sum_{\ell=0}^{j+1} \|  \hat a_\ell \|_\rho   } \\ &
\leq CKB^{j+1} \rho^{j+2} h + CKB^{j+1} \rho^{j+2} 2Kh e^{8Kh} \leq  B^{j+2} \rho^{j+2} h ,
 \end{split} \]
provided that we choose $ B $ large enough depending on $ K $ and $C$, and let $0 < \rho$ be small enough so that $B\rho \leq 1/2$.

Hence, by induction, we have constructed $a_{\ell}$, $\ell \geq 0$, satisfying $\| \hat a_{\ell}\|_{\rho} \leq K2^{-\ell}h$, such that $a = \sum_{\ell = 0}^{\infty} a_{\ell}$ is an analytic symbol for which \eqref{eq:conj} holds.
\end{proof}

We now need the following fact which we first establish on a formal level. For more on functional calculus for quadratic symbols see Derezi\'nski \cite{Der} and H\"ormander \cite{hor}.
\begin{lemm}
\label{l:F2f}
Suppose $ F ( z, \zeta; h ) $ is an analytic symbol of the form
\begin{equation}
\label{eq:F2f} F ( z, \zeta;  h ) = f ( z \zeta; h ) , \ \ \ f ( w; h ) = \sum_{k=0}^\infty h^k f_k ( w ) , \ \  \sup_{ D ( 0, \delta ) } |f_k| \leq A^{k+1} k^k.  \end{equation}
Then there exists an analytic symbol $ g ( w; h ) = \sum_{ k=0}^\infty h^k g_k ( w ) $ such that, as formal pseudodifferential operators,
\begin{equation}
\label{eq:F2f1}
 F^{\rm{w}} ( z, hD_z ;  h ) = g ( z hD_z + h/2i; h) .
 \end{equation}
\end{lemm}

\noindent
{\bf Remark.} Functional calculus of formal analytic pseudodifferential operators is reviewed in a self-contained way in \cite[\S a.2]{HelSj}. In our case, we have
\[
g(zhDz + h/2i; h) := \frac{1}{2\pi i } \int_{\partial D( 0, \varepsilon ) } \left( \lambda - z hD_z - \frac{h}{2i}\right)^{-1} g ( \lambda; h ) d\lambda,
\]
where $\varepsilon > 0$ is small enough independent of $ h $ and $(\lambda - z hD_z - h/2i )^{-1}$, for $\lambda \in \partial D ( 0, \varepsilon)$, is a formal analytic pseudodifferential operator, whose Weyl symbol $R(\lambda,z,\zeta;h)$ is an analytic symbol for $ ( z, \zeta ) \in \neigh_{\mathbb C^2 } ( 0 ) $, satisfying
\[
R(\lambda,z,\zeta;h) \# (\lambda - z\zeta) = 1,
\]
using the Weyl composition given by (\ref{eq:new1}).

\begin{proof}[Proof of Lemma {\rm \ref{l:F2f}}]
We shall prove \eqref{eq:F2f1} for the classical quantization and for a modified $ g $:
\begin{equation}
\label{eq:F2g}
F ( z, hD_z ; h ) = g ( zhD_z; h ) .
\end{equation}
A modification for the Weyl quantization follows from \cite[Theorem 4.13]{z12} and the fact that $ \exp( ih D_z D_\zeta/2 ) f ( z\zeta; h ) =
f_{1/2} ( z \zeta; h ) $, where $ f_{1/2}$ is an analytic symbol, by analytic stationary phase~\cite[Chapter 2]{sam}.

Let us first assume that $f$ in (\ref{eq:F2f}) is independent of $h$ and that $f(0) = 0$, since the constant term contribution is obvious. Writing
\[
f(w) = \sum_{n=1}^{\infty} f_{n} w^n,\quad \abs{f_{n}}\leq B^{n+1},\quad n=1,2,\ldots,
\]
we obtain formally,
\begin{equation}
\label{eq:L.3.4.1}
F(z,hD_z;h) = \sum_{n=1}^{\infty} f_{n}i^{-n} z^n\, (h\partial_z)^n.
\end{equation}
An induction argument based on $ ( z \partial_z - n ) ( z^n \partial^n_z ) = z^{n+1} \partial^{n+1}_z $ shows that
\begin{equation}
\label{eq:L.3.4.2}
z^n \partial_z^n = p_n ( z \partial_z ) , \ \  p_n ( t ) = \prod_{j=0}^{n-1} ( t - j  ) =
\sum_{k =0}^{n-1} a_{k,n} t^{n-k}, \ \ \  |a_{k,n}| \leq  2^n n^{k},
\end{equation}
and we get, combining (\ref{eq:L.3.4.1}) and (\ref{eq:L.3.4.2}),
\begin{equation}
F(z,hD_z;h) = \sum_{n=1}^{\infty} h^n f_{n}i^{-n} p_n(z\partial_z) = \sum_{n=1}^{\infty} f_n i^{-n} \left(\sum_{k=0}^{n-1} h^k a_{k,n} (zh\partial_z)^{n-k}\right).
\end{equation}
We obtain therefore (\ref{eq:F2g}), with
\[ \begin{split}
g (w; h ) & = \sum_{n=1}^\infty f_{n} \left(\sum_{ k=0}^{n-1} h^k a_{k,n} i^{-k} w^{n-k}\right)
=  \sum_{k=0}^\infty h^k \left(\sum_{n=k+1}^\infty f_n a_{k,n} i^{-k} w^{n-k}\right) \\
& = \sum_{ k=0}^\infty
h^k g_k (w ) , \ \ \  g_k (w) = \sum_{ p=1}^\infty  f_{ k+p} a_{ k, k+p} i^{-k} w^p.
\end{split} \]
Here we estimate, using (\ref{eq:L.3.4.2}),
\[  \begin{split}
|g_k ( w ) | &  \leq \sum_{ p=1}^\infty B^{k+p+1} 2^{k+p} (k+p)^{k} |w|^p  \leq
(2B)^{k+1} k^k \sum_{p=1}^{\infty} (2Be |w| )^p \\
&\leq (2B)^{k+1} k^{k} , \ \ \  |w| \leq ( 4 e B )^{-1} .
\end{split} \]
Since \eqref{eq:F2f} implies that $ f_k ( w ) = \sum_{p=0}^\infty f_{kp} w^p $, $|f_{kp}| \leq \delta^{-p} A^{k+1} k^k $, the general case follows from the $ h$--independent case.
\end{proof}

\noindent
{\bf Remark}. An alternative approach to Lemma \ref{l:F2f} proceeds as follows. Let $q_0 = z\zeta$, so that $q_0^{\rm w}(z,hD_z) = zhD_z + {h}/{2i}$, and let us observe that
\[
[F^{\rm w}(z,hD_z;h),q_0^{\rm w}(z,hD_z)] = \frac{h}{i} {\rm Op}^{\rm w}_h(\{F,q_0\}) = 0,
\]
in view of (\ref{eq:F2f}). To conclude that $F^{\rm w}(z,hD_z;h)$ is a function of $q_0^{\rm w}$, we may argue as in~\cite[Appendix A]{HitSj3a}: put $z = e^s$ and work locally near $z = r > 0$ for some fixed small $r > 0$. Then $zhD_z = hD_s$ and using the fact that the class of analytic pseudodifferential operators is preserved under analytic changes of variables, we conclude that in the new coordinates,
\[
F^{\rm w} = r^{\rm w}(s,hD_s;h),
\]
where $r$ is an analytic symbol. Using that $[r^{\rm w}(s,hD_s;h),hD_s] = 0$ we get $r^{\rm w} = r^{\rm w}(hD_s;h)$, and returning to the $z$-coordinates gives
\begin{equation}
\label{eq:monodr}
F^{\rm w}(z,hD_z;h) = g(q_0^w;h),
\end{equation}
where $g$ is an analytic symbol. An argument of analytic continuation around the origin shows that the representation (\ref{eq:monodr}) is valid in a full neighborhood of $0$.

For completeness we also include the following, essentially well known, result, see also~\cite[Lemma 4.3.17]{Dur}.
Its relevance comes from the fact that in the WKB literature (see for instance \cite{seiwi})
one typically computes the quantum action $ S ( E; h )$ while the eigenvalues (or resonances)
are given in terms of its inverse $ G ( z; h ) $. In this paper we construct $ G $ directly
but of course the connection with actions is there -- see Theorem \ref{t:action} for an indication of this.

\begin{prop}
\label{p:inverse}
Suppose that $ S ( x; h ) $, $ x \in \neigh_{\mathbb C} ( 0 ) $, is an analytic symbol
such that $ S_0 ( 0 ) = 0 $, $ S'_0 ( 0 ) \neq 0 $. Then there exists an analytic
symbol $ G ( x; h ) $ such that
\begin{equation}
\label{eq:inv}    S (  G ( x; h ); h ) \equiv x , \ \  x \in \neigh_{\mathbb C} ( 0 ) .
\end{equation}
The equivalence is in the sense of formal analytic symbols.
\end{prop}
\begin{proof}
We define a simpler version of the $\rho$-norms of analytic symbols,
\begin{equation}
\label{eq:normr} \| A \|_{ \rho} := \sum_{ k = 0 }^\infty \frac{ \rho^k} {k!} \| A_k \|, \ \
A (x;h) = \sum_{ k=0}^\infty A_k ( x ) h^k, \ \ \   \| f \| := \sup_{x \in\Omega } |f ( x ) |,
\end{equation}
where $ \Omega $ is a fixed sufficiently small neighbourhood of $ 0$. This norm is finite
for {\em some} $ \rho > 0$ if and only if $ A $ is an analytic symbol. We recall from
\eqref{eq:key} (or rather from a much simpler special case) that
\begin{equation}
\label{eq:Hell}  \| H^\ell  \|_\rho \leq \| H \|_{\rho}^\ell, \ \ \
 H ( x, h )^\ell = \sum_{ p=0}^\infty H_{\ell p} ( x ) h^p. \end{equation}
 If $ S ( x; h ) = \sum_{ j=0}^\infty S_j ( x ) h^j $, then we can find $ G_0 ( x ) $ such that
$ S_0 ( G_0 (x ) ) = x $ near $ 0 $. Replacing $ S ( x; h ) $ by $ S ( G_0 ( x ); h ) $ we
can assume that $ S_0 ( x ) = x $. We then write $ S ( x; h ) = x - h F ( x; h ) $, and postulate the form of $ G ( x; h ) $ to
be
\[ x + \sum_{ \ell = 1}^\infty h^\ell G_\ell ( x ) =: x + h H ( x; h ) = x + h \left( \sum_{ \ell = 0 }^\infty
h^\ell H_\ell ( x )  \right). \]
Hence, \eqref{eq:inv} becomes
\begin{equation}
\label{eq:G2H} H ( x; h ) =  F ( x + h H ( x; h ); h) =  \sum_{ \ell=0}^\infty \frac{h^\ell }{ \ell!} F^{(\ell)} ( x; h )
H( x; h )^\ell. \end{equation}
Here
\begin{equation}
\label{eq:Fell}   \frac{1 }{ \ell!} F^{(\ell)} ( x, h ) = \sum_{ m = 0 }^\infty F_{ m \ell} ( x )  h^m ,  \ \ \
F_{m\ell} ( x ) := \frac{1 }{ \ell!} F_m ^{(\ell)} ( x), \ \ \
\| F_{ m \ell} \| \leq A^{m + \ell +1 } m!.
\end{equation}
It follows from (\ref{eq:G2H}) and (\ref{eq:Fell}) that the terms in the expansion of $ H ( x; h ) $ are obtained from an iteration procedure based on
\eqref{eq:G2H} (note that $ H_{0p} ( x) = \delta_{0p} $):
\begin{equation}
\label{eq:iterH}
H_k ( x ) = \sum_{ p + m + \ell = k } F_{ m \ell } ( x ) H_{ \ell p } ( x ) =
F_{ k0 } ( x ) + \sum_{\substack{p + m + \ell = k\\ \ell \geq 1}} F_{ m \ell } ( x ) H_{ \ell p } ( x ).
\end{equation}

We now define
\[  M_n  := \sum_{k=0}^n \frac{ \rho^k}{k!} \| H_k \| , \ \ \ \    \| H \|_\rho = \sup_{ n } M_n .
\]
Since $ H_{\ell k} $, $ k < n $, depend only on $ H_p $ with $ p < n $, we see that \eqref{eq:Hell}
implies
\begin{equation}
\label{eq:lk2l} \sum_{ k=0}^{n-1} \frac{\rho^k}{k!} \| H_{ \ell k } \| \leq M_ { n-1}^\ell ,
\end{equation}
We now want to estimate $ M_n  $ in terms of $ M_{n-1} $: for that
we use  \eqref{eq:Fell} and \eqref{eq:lk2l} in \eqref{eq:iterH} to obtain, with $ B = A ( 1 - \rho_0 A )^{-1} $,
$ 0 < \rho < \rho_0 < 1/A $,
\begin{multline}
M_n  \leq \sum_{k=0}^n \frac{\rho^k}{k!}\| F_{k0}\| + \sum_{k= 1}^n \frac{\rho^k}{k!} \sum_{\substack{p + m + \ell = k\\ \ell \geq 1}}
\| F_{m\ell } \| \| H_{ \ell p} \| \\
\leq B + \sum_{\substack{1\leq p + m + \ell \leq n\\ \ell \geq 1}} \frac{ \rho^{p + m +\ell} \| F_{m\ell } \| \| H_{ \ell p} \| }
{ ( p + m +\ell ) ! } \leq B +  \sum_{m = 0}^{n-1}\sum_{\ell = 1}^{n} \sum_{p = 0 }^{n-1} \frac{ \rho^{p + m +\ell} A^{m+\ell +1}m! \| H_{ \ell p} \| }{ ( p + m +\ell ) ! } \\
\leq B + \sum_{m = 0}^{n-1}\sum_{\ell = 1}^{n} \sum_{p = 0 }^{n-1} \frac{A (A\rho)^{m+\ell}}{\ell!} \frac{\| H_{ \ell p} \| \rho^p}{p!}
\leq B + \sum_{m = 0}^{n-1}\sum_{\ell = 1}^{n} \frac{A (A\rho)^{m+\ell}}{\ell!} M_{n-1}^{\ell} \\
\leq B + B \sum_{ \ell = 1}^n \frac{ (A\rho)^{\ell} M_{ n-1}^\ell} {\ell !} \leq B \exp \left ( A \rho M_{ n-1 }  \right).
\end{multline}
Putting $ a_n := M_{ n } / B $ we get the following iterative inequality
$  a_n \leq \exp \left( \rho AB a_{n-1}   \right) $, $a_0 \leq 1$.
It has been known since Euler that the sequence $(a_n)$ is bounded if $ \rho AB \leq 1/e $. For $ \rho > 0$ small enough we obtain that
$ \| H \|_\rho < \infty$, proving that $ H$, and hence $ G $, are   analytic symbols.
\end{proof}

\section{Review of pseudodifferential and Fourier integral operators}
\label{s:rev}

Here we recall, with proofs or precise references, the needed results from microlocal analytic theory.

\subsection{Global pseudodifferential operators}
\label{s:global}

We describe the action of pseudodifferential
operators defined by globally analytic symbols on weighted spaces of holomorphic functions following
\cite[Chapter 1]{HiSj15} and \cite[\S 12]{Sj02}  (see also \cite[Chapter 13]{z12} for a gentle introduction).

We start with $\Phi_0$,  a strictly plurisubharmonic quadratic form on $\mathbb C^n $, and the Hilbert space of weighted holomorphic functions:
\begin{equation}
\label{eq:2.4}  H_{\Phi_0} = H_{\Phi_0}  (\mathbb C^n ) := \{ u \in \mathscr O ( \mathbb C^n ) :
\| u \|_{\Phi_0}^2 := \int_{\mathbb C^n } | u ( x ) |^2 e^{ - 2 \Phi_0 ( x )/h } dL ( x) < \infty \} ,\end{equation}
where $ dL ( x) $ is the Lebesgue measure on $ \mathbb C^n = \mathbb R^{2n} $.
To this space we associate a geometric object, a real linear subspace of $ \mathbb C^{2n} $ given by
\begin{equation}
\label{eq2.1}
\Lambda_{\Phi_0} = \left\{\left(x,\tfrac{2}{i}{\partial_x \Phi_0}  (x)\right); \, x\in \mathbb C^n \right\} .
\end{equation}
A function  $1\leq m \in C^{\infty}(\Lambda_{\Phi_0})$ is called an order function on $\Lambda_{\Phi_0}$
if for some $C_0 > 0$, $N_0 > 0$, we have
\begin{equation}
\label{eq2.1.1}
m(X) \leq C_0 \langle{X-Y\rangle}^{N_0} m(Y),\quad X,Y \in \Lambda_{\Phi_0}.
\end{equation}
Since  $\pi:  \Lambda_{\Phi_0} \rightarrow \mathbb C^n $, $ \pi ( x, \xi ) = x $, identifies $ \Lambda_{\Phi_0 }
$ with $ \mathbb C^n $ we can consider $ m $ as a function on $ \mathbb C^n $. To the order function $ m$ we associate
a more general class of weighted spaces generalizing Sobolev spaces:
\begin{equation}
\label{eq2.6.1}
\begin{gathered}
H_{\Phi_0, m } = H_{\Phi_0,m}(\mathbb C^n) := \mathscr O ( \mathbb C^n ) \cap L^2_{\Phi_0, m } , \ \ \ \
L^2_{\Phi_0, m }  := \{ u  :
\| u \|_{\Phi_0, m} < \infty \} ,\\
\| u \|_{\Phi_0, m}^2 := \int_{\mathbb C^n } m ( x )^2  | u ( x ) |^2 e^{ - 2 \Phi_0 ( x )/h } dL ( x).
\end{gathered}  \end{equation}

Let $P = P(x,\xi;h)$ be a holomorphic function in  $\Lambda_{\Phi_0} + B_{\mathbb C^{2n}} ( 0, \delta ) $,
$ \delta > 0 $, and assume that
\begin{equation}
\label{eq2.1.2}
\abs{P(x,\xi;h)} \leq C m(x),\quad (x,\xi) \in \Lambda_{\Phi_0} + B_{\mathbb C^{2n}} ( 0, \delta ) ,
\end{equation}
and that there is a complete asymptotic expansion
\begin{equation}
\label{eq2.2}
P(x,\xi;h) \sim \sum_{k=0}^{\infty} h^k p_k(x,\xi)  , \ \ \ p := p_0,
\end{equation}
in the space of holomorphic functions satisfying (\ref{eq2.1.2}). Eventually, we strengthen this to
analytic symbol expansion in the sense of \S \ref{s:anas}. We also demand that $ P $ is elliptic near infinity:
\begin{equation}
\label{eq2.3}
\abs{p(x,\xi)} \geq {m(x)}/{C},\quad (x,\xi) \in \Lambda_{\Phi_0} + B_{\mathbb C^{2n}} ( 0, \delta) , \ \ \
\abs{(x,\xi)} \geq R.
\end{equation}
The $h$-Weyl quantization of $ P $, $P^{\rm{w}}(x,hD_x;h)$, is an unbounded operator on $ H_{\Phi_0 }$
defined by
\begin{equation}
\label{eq2.5}
P^{\rm{w}}(x,hD_x;h)u(x) = \frac{1}{(2\pi h)^n} \int\!\!\!\int_{\Gamma_{\Phi_0}(x)} e^{\frac{i}{h}(x-y)\cdot \theta} P\left(\frac{x+y}{2},\theta;h\right) u(y)\, dy\,d\theta,
\end{equation}
where the contour of integration $\Gamma_{\Phi_0}(x) \subset \mathbb C^{2n}_{y,\theta}$ is given by
\begin{equation}
\label{eq2.6}
y \mapsto \theta = \frac{2}{i}\frac{\partial \Phi_0}{\partial x}\left(\frac{x+y}{2}\right) + \frac{i \delta }{2} \frac{\overline{(x-y)}}{\langle{x-y\rangle}} .
\end{equation}
(We note that $ ( (x+y)/2 , \theta ) \in \Lambda_{\Phi_0} + B_{\mathbb C^{2n}} ( 0, \delta) $. For a discussion of contours of
integration in this context see \cite[\S 13.2.1]{z12}.) If $ \widetilde m $ is another order function then
an application of Schur's lemma \cite[(A.5.3)]{res} to (\ref{eq2.5}), (\ref{eq2.6}), gives
\begin{equation}
\label{eq2.7}
P^{\rm{w}}(x,hD_x;h) = \mathcal O(1): H_{\Phi_0,\widetilde{m}}\rightarrow H_{\Phi_0,{\widetilde{m}}/{m}},
\end{equation}
see~\cite{HiSj15},
~\cite[Section 12.2]{Sj02}. In particular, for the domain of $ P^{\rm{w}} $
as an unbounded operator on $ H_{\Phi_0}$ we can take $ H_{\Phi_0, m } $.

The holomorphy of $P$ in a tubular neighbourhood of $\Lambda_{\Phi_0}$ allows us to consider other weights as well:

\begin{prop}
\label{p:neww}
Suppose that $\Phi = \Phi_0 + f \in C^{\infty}(\mathbb C^n)$ satisfies
\begin{equation}
\label{eq2.8}
f \in L^{\infty}(\mathbb C^n),\quad \|{\nabla^k f}\|_{L^{\infty}({\mathbb C^n})} \leq \varepsilon  ,\quad k = 1,2,
\end{equation}
where $ \varepsilon > 0$ is sufficiently small depending on $ \delta $ in \eqref{eq2.1.2}. Then
\begin{equation}
\label{eq2.9}
P^{\rm{w}}(x,hD_x;h) = \mathcal O(1): H_{\Phi,\widetilde{m}} \rightarrow H_{\Phi,{\widetilde{m}}/{m}},
\end{equation}
where the exponentially weighted spaces are defined as in {\rm (\ref{eq2.6.1})}, replacing $\Phi_0$ by $\Phi$.
\end{prop}
\begin{proof}
Defining a new contour
\begin{equation}
\label{eq2.10}
\Gamma_{\Phi}(x):  y \mapsto \theta = \frac{2}{i}\frac{\partial \Phi}{\partial x}\left(\frac{x+y}{2}\right) + \frac{i
\delta }{2} \frac{\overline{(x-y)}}{\langle{x-y\rangle}},
\end{equation}
and performing a contour deformation in (\ref{eq2.5}) we obtain a different formula for the action of $P^{\rm{w}}(x,hD_x;h)$ adapted to the weight $\Phi$,
\begin{equation}
\label{eq2.11}
P^{\rm{w}}(x,hD_x;h)u(x) = \frac{1}{(2 \pi h)^n } \int\!\!\!\int_{\Gamma_{\Phi}(x)} e^{\frac{i}{h}(x-y)\cdot \theta} P\left(\tfrac{x+y}{2},\theta;h\right) u(y)\, dy\,d\theta.
\end{equation}
Along the contour $\Gamma_{\Phi}(x)$, we have
\begin{equation}
\label{eq2.12}
\begin{split}
& -\Phi(x)+ {\rm Re}\,\left(i(x-y)\cdot \theta\right) + \Phi(y) \\
& \ \ \ \ \ \ = -\Phi(x) + {\rm Re}\,\left(2{\partial_x \Phi}(\tfrac{x+y}{2})\cdot(x-y)\right) + \Phi(y) -(\delta/2){\abs{x-y}^2}/{\langle{x-y\rangle}} \\
&  \ \ \ \ \ \  = -f(x) + \langle{\nabla f\left(\tfrac{x+y}{2}\right),x-y\rangle}_{{\mathbb R}^{2n}} + f(y) - (\delta/2)  {\abs{x-y}^2}/{\langle{x-y\rangle}}.
\end{split} \end{equation}
Here
\begin{equation}
\label{eq2.13}
-f(x) + \langle{\nabla f\left(\tfrac{x+y}{2}\right),x-y\rangle}_{{\mathbb R}^{2n}} + f(y) \leq 2 \|{\nabla f}\|_{L^{\infty}({\mathbb C}^n)}\abs{x-y},
\end{equation}
and an application of Taylor's formula gives that
\begin{equation}
\label{eq2.14}
\begin{split}
& -f(x) + \langle{\nabla f\left(\tfrac{x+y}{2}\right),x-y\rangle}_{{\mathbb R}^{2n}} + f(y)
 = \int_0^1 (1-t) f''\left(\tfrac{x+y}{2} - t\left(\tfrac{x-y}{2}\right)\right)\,\tfrac{(x-y)}{2}\cdot \tfrac{(x-y)}{2}\,dt \\
& \ \ \ \ \ \  - \int_0^1 (1-t) f''\left(\tfrac{x+y}{2} + t\left(\tfrac{x-y}{2}\right)\right)\,\tfrac{(x-y)}{2}\cdot \tfrac{(x-y)}{2}\,dt
  \leq \tfrac{1}{4} \|{\nabla^2 f}\|_{L^{\infty}({\mathbb C^n})}\abs{x-y}^2.
\end{split}
\end{equation}
(The Hessian and the scalar product are taken in the sense of $\mathbb R^{2n}$.)  It follows, in view of (\ref{eq2.1.2}), (\ref{eq2.8}), (\ref{eq2.12}), (\ref{eq2.13}), and (\ref{eq2.14}), that the absolute value of the effective kernel of the operator in (\ref{eq2.9}), (\ref{eq2.11}) does not exceed a multiple of
\begin{equation*}
\begin{split}
& h^{-n} m\left(\frac{x+y}{2}\right) \frac{\widetilde{m}(x)} {m(x) \widetilde m ( y )}
\exp\left(\frac{1}{h} \left(2\varepsilon \abs{x-y} \min (1,\abs{x-y}) - \frac \delta 2 \frac{\abs{x-y}^2}{\langle{x-y\rangle}} \right)\right)
\\
& \ \ \ \ \ \ \ \ \ \ \ \ \ \ \ \ \ \ \ \leq  C h^{-n} \exp \left(- \frac \delta {4 h} \frac{\abs{x-y}^2}{\langle{x-y\rangle}}\right)\langle{x-y\rangle}^{N},
\end{split}
\end{equation*}
for some $N\geq 0$, provided that $ \varepsilon > 0$ in (\ref{eq2.8}) is sufficiently small. The conclusion (\ref{eq2.9}) follows therefore from  an application of Schur's lemma \cite[(A.5.3)]{res}.
\end{proof}

The next proposition gives a way of approximating the action of $ P $ on $H_{\Phi}$ by multiplication.
The method goes back to the Cordoba--Fefferman proof of the sharp G{\aa}rding inequality and
the version here comes from \cite[Section 12.4]{Sj02}, see also \cite[Proposition 1.4.4]{HiSj15}.
For the reader's convenience we present the proof.
\begin{prop}
\label{prop:qm}
Suppose that $ \psi \in C^1( \mathbb C^n) $ is such that $ \psi, \nabla \psi \in L^\infty ( \mathbb C^n ) $,
and set  $ \xi(x) = (2/i) {\partial_x \Phi}(x)$. Then for order functions $ m_j $, $ j=1,2 $, satisfying
$ m_1 m_2 \geq m $,
\begin{equation}
\label{eq2.16}
\begin{split}
& \left(\psi P^{\rm{w}}(x,hD_x;h)u_1, u_2  \right)_{L^2_{\Phi}}  =
\int_{\mathbb C^n} \psi(x)\, p\left(x,\xi(x)\right)u_1(x)\overline{u_2 (x)}e^{-\frac{2}{h} \Phi(x)}\, dL(x) \\
& \ \ \ \ \ \ \ \ \ \ \ \ \ \ \ \ \ \ \ \ \ \ \  +
\mathcal{O}(h)\|u_1\|_{H_{\Phi, m_1}} \|u_2\|_{H_{\Phi, m_2 }},
\ \ \ u_j \in H_{\Phi,m_j} .
\end{split}
\end{equation}
Furthermore, for $ u\in H_{\Phi,m}$,
\begin{equation}
\label{eq2.17}
\begin{split}
& \left(\psi\, P^{\rm{w}}(x,hD_x;h)u,P^{\rm{w}}(x,hD_x;h)u\right)_{L^2_{\Phi}}
 = \\
 & \int_{\mathbb C^n} \psi(x)\, \abs{p\left(x,\xi(x)\right)}^2 \abs{u(x)}^2 e^{-2\Phi(x)/h}\,dL(x)
 + \mathcal O(h) \| u \|^2_{H_{\Phi,m}} .
\end{split}
\end{equation}
\end{prop}
\begin{proof}
Taylor's formula gives (recall that $\displaystyle \xi(x) = (2/i) {\partial_x \Phi}(x)$):
\begin{equation}
\label{eq2.15.1}
\begin{split}
p\left(\tfrac{x+y}{2},\theta\right) & = p(x,\xi(x)) + \tfrac{1}{2}\partial_x p(x,\xi(x))\cdot (y-x) \\
& \ \ \ \ \ + \partial_{\xi} p(x,\xi(x))\cdot (\theta - \xi(x)) + r(x,y,\theta),
\end{split}
\end{equation}
where
\begin{equation*}
r(x,y,\theta) = \int_0^1 (1-t)\partial_t^2 \left[ p\left(x + t\tfrac{y-x}{2},\xi(x) + t(\theta - \xi(x))\right) \right] dt.
\end{equation*}
Using (\ref{eq2.1.2}) together with the Cauchy estimates, as well as the fact that $\theta - \xi(x) = \mathcal O(\abs{x-y})$ along the contour (\ref{eq2.10}), we get
\begin{equation}
\label{eq2.15.3}
\abs{r(x,y,\theta)} \leq C m(x)\langle{y-x\rangle}^{N_0} \abs{y-x}^2,\quad (y,\theta)\in \Gamma_{\Phi}(x).
\end{equation}
Let
\begin{equation}
\label{eq2.15.4}
\begin{split}
Ru(x) & = \frac{1}{(2 \pi h)^n} \int\!\!\!\int_{\Gamma_{\Phi}(x)} e^{\frac{i}{h}(x-y)\cdot \theta} r\left(x,y,\theta\right) u(y)\, dy\,d\theta \\
& =: \int k(x,y;h)u(y) dL(y).
\end{split}
\end{equation}
The bound on $ R $ as an operator $ H_{\Phi,\widetilde{m}} \rightarrow L^2_{\Phi, \widetilde m/m} $
is given by the bound on the operator on $L^2(\mathbb C^n)$ with the kernel (see \cite[Discussion after (13.4.3)]{z12}):
$$
e^{-\Phi(x)/h} \frac{\widetilde{m}(x)}{m(x)\widetilde{m}(y)} k(x,y;h)e^{\Phi(y)/h}.
$$
Estimates (\ref{eq2.12}), (\ref{eq2.13}), (\ref{eq2.14}), and (\ref{eq2.15.3}) show that it is bounded by
a multiple of
\begin{equation}
\label{eq2.15.5}
   h^{-n} \exp\left(-\frac{\delta}{4h}\frac{\abs{x-y}^2}{\langle{x-y\rangle}}\right) \langle{y-x\rangle}^N \abs{y-x}^2.
\end{equation}
An application of Schur's inequality  \cite[(A.5.3)]{res} shows that the expression in (\ref{eq2.15.5}) is the integral kernel of a convolution operator on $L^2(\mathbb C^n)$ of norm $\mathcal O(h)$. It follows that
\begin{equation}
\label{eq2.15.6}
R = \mathcal O(h): H_{\Phi,\widetilde{m}} \rightarrow L^2_{\Phi, \widetilde m/m}.
\end{equation}

Inserting \eqref{eq2.15.1} in \eqref{eq2.11} (and using the boundedness of the lower order terms) we obtain that $P^{\rm {w}}u(x)$ is given by
\[ \begin{split} & p ( x, \xi ( x ) ) u ( x ) +
 \frac{1}{(2 \pi h)^n} \int\!\!\!\int_{\Gamma_{\Phi}(x)} e^{\frac{i}{h}(x-y)\cdot \theta}
   \partial_{\xi} p(x,\xi(x))\cdot (\theta - \xi(x) ) u (y) dy\, d \theta \\
    & \ \ \ \ \ \ \ \ \ \ + R_1 u ( x ) + R_2 u ( x ) ,
\end{split} \]
where $ R_1  = \mathcal O(h): H_{\Phi,\widetilde{m}} \rightarrow L^2_{\Phi, \widetilde m/m} $ (coming
from $ R $ in \eqref{eq2.15.4} and lower order terms), and
\[ \begin{split} R_2 u ( x ) & =
\frac{1}{2(2\pi h)^n} \int\!\!\!\int_{\Gamma_{\Phi}(x)} e^{\frac{i}{h}(x-y)\cdot \theta}
\partial_x p (x,\xi(x))\cdot (y-x)  u ( y ) dy\, d \theta \\
&   =
\frac{1}{2(2\pi h)^n} \int\!\!\!\int_{\Gamma_{\Phi}(x)}
\partial_x p (x,\xi(x))\cdot ih \partial_\theta ( e^{\frac{i}{h}(x-y)\cdot \theta} ) u ( y ) dy\, d\theta = 0 .
\end{split} \]
Noting that $ \theta_j e^{\frac{i}{h}(x-y)\cdot \theta} = - h D_{y_j}\left(e^{\frac{i}{h}(x-y)\cdot \theta}\right)$ we obtain
\begin{equation*}
\begin{split}
P^{\rm{w}}(x,hD_x;h) u(x) & = p(x,\xi(x))u(x) + \sum_{j=1}^n \partial_{\xi_j}p (x,\xi(x))(hD_{x_j} - \xi_{j}(x))u(x)
+ R_1 u(x).
\end{split}
\end{equation*}
To obtain \eqref{eq2.16} we note that
\[  ( h D_{x_j } + \xi_j (x ) ) e^{ -2 \Phi ( x )/h } =
( h D_{x_j } + 2 D_{x_j } \Phi (x ) ) e^{ -2 \Phi ( x )/h } = 0 , \]
so that for $ u_j \in H_{\Phi, m_j } $,
\[  \begin{split}
& \int_{ \mathbb C^n } \psi ( x ) \partial_{\xi_j } p ( x, \xi ( x ) ) (hD_{x_j} - \xi_{j}(x)) u_1(x) \overline{ u_2 ( x) }
e^{ - 2 \Phi ( x )/h } dL ( x )  \\
&  \
= - \int_{\mathbb C^n } u_1 ( x )  (hD_{x_j} + \xi_{j}(x)) [ \psi ( x ) \partial_{\xi_j } p ( x, \xi ( x ) )  \overline{ u_2 ( x) }
e^{ - 2 \Phi ( x )/h } ] d L  ( x) \\
& \
= - h  \int_{\mathbb C^n }  u_1(x) D_{x_j} (  \psi ( x ) \partial_{\xi_j } p ( x, \xi ( x ) )) \overline{ u_2 ( x) }
e^{ - 2 \Phi ( x )/h } d L ( x)\\
& \  = \mathcal O ( h ) \| u_1 \|_{ H_{\Phi, m_1} }  \| u_2 \|_{ H_{\Phi, m_2} } .
\end{split} \]
The same argument gives \eqref{eq2.17}.
\end{proof}

We conclude by applying \eqref{eq2.17} to obtain an elliptic estimate near infinity. It allows us to localize
the spectral analysis to a neighbourhood of the critical point:
\begin{prop}
\label{elliptic}
Suppose  $P(x,\xi;h)$ satisfies
\eqref{eq2.1.2}, \eqref{eq2.2}, \eqref{eq2.3}, and let $\Phi = \Phi_0 + f $ with $ f $ satisfying \eqref{eq2.8}. Then there exist $h_0 > 0$, $\eta > 0$, such that for $ 0 < h < h_0 $
 and $u\in H_{\Phi,m} $,
\begin{equation}
\label{eq2.22}
\| u \|_{L^2_{\Phi,m} (\mathbb C^n \setminus B (0, 2 R ) )}  \leq C \|P^{\rm{w}}(x,hD_x;h)u\|_{H_{\Phi}(\mathbb C^n \setminus B (0, R/2 ) )} +
 e^{-\eta/h}  \|{u}\|_{H_{\Phi,m}}.
\end{equation}
\end{prop}
\begin{proof}
Let $\chi \in C^{\infty}_0(\mathbb C^n;[0,1])$ satisfy $\chi(x) = 1$, $\abs{x}\leq R$, and ${\rm supp}\, \chi \subset B(0,2R )$. We shall apply (\ref{eq2.17}) with $\psi \in  C^{\infty}(\mathbb C; [0,1])$ such that
\begin{equation*}
\psi(x) = 1,\,\,\abs{x}\geq R, \quad \psi(x) = 0,\,\, \abs{x} < \frac{R}{2},
\end{equation*}
and with the weight $\Phi$ replaced by
\begin{equation}
\label{eq2.18}
\widetilde{\Phi} = \Phi - \eta(1-\chi), \quad 0 < \eta \ll 1,
\end{equation}
which still satisfies (\ref{eq2.8}), for $\eta > 0$ small enough. For $u\in H_{\Phi,m}(\mathbb C^n)$,
and $ \tilde{\xi}(x) = (2/i){\partial_x \widetilde{\Phi}}(x)$,
\begin{equation*}
\int_{|x| \geq R} | p (x,\tilde{\xi}(x) )| ^2 |u(x)|^2 e^{-2\widetilde{\Phi}(x)/h}\,dL(x)
 \leq \|P^{\rm{w}}(x,hD_x;h)u \|^2_{H_{\widetilde{\Phi}}(\abs{x} > R/2)}
+ \mathcal O(h) \|u \|^2_{H_{\widetilde{\Phi},m}}.
\end{equation*}
Combining this with the assumption \eqref{eq2.3} which gives
$$
|p(x,\tilde{\xi}(x))| \geq {m(x)}/C ,\quad \abs{x} \geq R,
$$
we get
\begin{equation*}
\int_{\abs{x} \geq R} |u(x)|^2\, m^2(x) e^{-2\widetilde{\Phi}(x)/h}\,dL(x) \\
\leq C \|P^{\rm{w}}(x,hD_x;h)u\|^2_{H_{\widetilde{\Phi}}(\abs{x} > R/2)} + \mathcal O(h) \|u \|^2_{H_{\widetilde{\Phi},m}},
\end{equation*}
and therefore, taking $h>0$ small enough,
\begin{equation*}
\int_{|x| \geq R} |u(x)|^2\, m^2(x) e^{-2\widetilde{\Phi}(x)/h}\,dL(x)
\leq e^{2\eta/h} \|P^{\rm{w}}(x,hD_x;h)u\|^2_{H_{\Phi}(\abs{x} > R/2)} + \mathcal O(h) \|{u}\|^2_{H_{\Phi,m}}.
\end{equation*}
Here we have also used that $\widetilde{\Phi}(x) = \Phi(x)$ for $\abs{x}\leq R$ and that $\widetilde{\Phi} \geq \Phi - \eta$ on $\mathbb C^n$. In the region $\abs{x}\geq 2R $, we have $\widetilde{\Phi} = \Phi - \eta$, and that gives \eqref{eq2.22}.
\end{proof}

\noindent
{\bf Remarks}. 1. In what follows, in order to enter the framework of  analytic $h$--pseudodifferential operators
, we shall strengthen condition (\ref{eq2.2}) by assuming that
\begin{equation}
\label{eq2.23}
\abs{p_k(x,\xi)} \leq C^{k+1} k^k\, m(x),\quad (x,\xi) \in \Lambda_{\Phi_0} + B_{\mathbb C^{2n}} (0, \delta), \quad k = 0,1,2,\ldots
\end{equation}

\noindent
2. We observe that rather than assuming (\ref{eq2.8}), it would have been sufficient to consider uniformly strictly plurisubharmonic weights $\Phi \in C^2(\mathbb C^n)$ such that $\Phi - \Phi_0 \in L^{\infty}(\mathbb C^n)$, $\norm{\nabla(\Phi - \Phi_0)}_{L^{\infty}(\mathbb C^n)}$ is small enough, and $\nabla^2 \Phi \in L^{\infty}(\mathbb C^n)$. Indeed, in the proof of Proposition \ref{p:neww}, instead of introducing the contour given in (\ref{eq2.10}), we deform to the contour
\[
\theta = \frac{2}{i} \int_0^1 \frac{\partial \Phi}{\partial x}(tx + (1-t)y)\, dt + \frac{i\delta}{2} \frac{\overline{(x-y)}}{\langle{x-y}\rangle}.
\]
Proposition \ref{prop:qm} and Proposition \ref{elliptic} are still valid under these weaker assumptions on the weight $\Phi$.

\subsection{Local Fourier integral operators}
\label{s:local}
The local theory is more subtle than the global theory reviewed in detail in \S \ref{s:global}. One key point
is that the action of operators is defined only up to exponentially small errors relative to the weights.
We cannot present all the details but precise references are provided.

To prove Theorem \ref{t:2} we will consider microlocal equivalence, with exponentially small errors, of   analytic pseudodifferential operators acting on Hilbert spaces of the form
\begin{equation}
\label{eq3.1}
H_{\Psi}(V) = \mathscr O (V) \cap L^2(V,e^{-2\Psi/h}\, dL(x)),
\end{equation}
where $V \Subset \mathbb C^n $ is a small open neighbourhood of $0\in \mathbb C^n $ and $\Psi = \Phi_0$, $\Phi$, with $\Phi_0$ quadratic strictly plurisubharmonic and $\Phi$ real analytic in $V$ satisfying $\Phi(x) = \Phi_0(x) + \mathcal O(x^3)$. In fact, in our applications the quadratic weight $\Phi_0$ will be strictly convex, and in the following discussion we shall make this stronger assumption,
\begin{equation}
\label{eq:assP}
\Psi  ( x) \ \text{ is strictly convex and } \
 \Psi(x) = \mathcal O(x^2), \ \ \ x\rightarrow 0. \end{equation}

We start by recalling $h$-pseudodifferential operators with classical analytic symbols acting on $H_{\Psi}(V) $
in (\ref{eq3.1}):
\begin{equation}
\label{eq3.13}
Pu(x;h) = \frac{1}{(2\pi h)^n} \underset{\Gamma(x)}{\int\!\!\!\int} e^{\frac{i}{h}\varphi(x,y,\theta)} P(x,\theta;h) u(y)\, dy\, d\theta, \ \ \ u \in H_\Psi (V ),
\end{equation}
where
\begin{equation}
\label{eq:obv}
\varphi ( x, y, \theta ) = ( x - y ) \cdot \theta,
\end{equation}
$P(x,\theta;h) = p(x,\theta) + \mathcal O(h)$ is a classical analytic symbol defined in a neighbourhood of $(0,0)\in \mathbb C^{2n}$, and $\Gamma(x) \subset \mathbb C^{2n}_{y,\theta}$ is a good contour for  $(y,\theta) \mapsto -{\rm Im}\, ((x-y)\cdot \theta) + \Psi(y)$, see \cite[\S 2.4.c]{HiSj15} for more on that terminology. Since we assumed
\eqref{eq:assP}, it can be taken independent of $ x $:
\begin{equation}
\label{eq3.17}
\Gamma : y \mapsto \theta = \tfrac{2}{i}{\partial_y \Psi}(y), \quad y\in {\rm neigh}_{\mathbb C^n}(0).
\end{equation}
In fact,
\[ \begin{split}
- \Im \left((x-y) \cdot \theta |_{ \Gamma }\right) + \Psi ( y ) & =
- \Im ( \tfrac{2}{i}{\partial_y \Psi}(y) \cdot (x - y ) ) + \Psi ( y ) \\
& =
\Psi ( y ) +  \partial_{ \Re y} \Psi ( y ) \cdot \Re ( x -y ) + \partial_{\Im y } \Psi ( y )
\cdot \Im ( x- y ) \\
& \leq \Psi ( x ) - |x -y|^2/C ,
\end{split}  \]
by the strict convexity of $ \Psi $.

We stress that if $ V $ is sufficiently small, then $Pu \in H_\Psi ( V )$ but the definition \eqref{eq3.13} depends on the realization of the
analytic symbol $ P ( x, \theta; h ) $ -- see \eqref{eq:reali}, \cite[Exemple 1.1]{sam} (or \cite[\S 2.2]{HiSj15}),
and the errors from the contour for $ |x - y | > \varepsilon_0 $ produce contributions
in $ H_{\Psi - \varepsilon_1 } (V )$, $ \varepsilon_1 > 0 $. Ultimately, this produces an overall
ambiguity which is however exponentially small, which is consistent with our goal. We also note that $ P : H_\Psi ( V ) \to H_\Psi ( V )$ is bounded uniformly in $ h $, in view of Schur's lemma.

The microlocal equivalence of analytic pseudodifferential operators is obtained using local {\em analytic Fourier integral operators} which quantize locally defined holomorphic symplectomorphisms:
\begin{equation}
\label{eq3.2}
\kappa: {\rm neigh}_{\mathbb C^{2n}} (0,0) \rightarrow {\rm neigh}_{\mathbb C^{2n}} (0,0),\quad \kappa(0,0) = (0,0).
\end{equation}
We assume that, in the notation of \eqref{eq2.1},
\begin{equation}
\label{eq3.6}
\kappa({\rm neigh}_{\Lambda_{\Phi_0}}(0,0)) = {\rm neigh}_{\Lambda_{\Phi}}(0,0).
\end{equation}
These assumptions on $ \kappa $ provide a particularly nice generating function: we
record that in the following lemma from \cite[Section 2]{MeSj2} (the proof there is geometric and is independent of the rest of the paper): 
\begin{lemm}
\label{l:gen}
Suppose that $ \kappa $ in \eqref{eq3.2} satisfies
$ \kappa^* ( d\xi \wedge d x ) = d\xi \wedge dx $ and \eqref{eq3.6}. Then (near $ 0 $)
\begin{equation}
\label{eq:graph} \{ ( y, \eta ), \kappa ( y, \eta ) ) \} = \{ (  y, -\varphi'_y(x,y,\theta), x,\varphi'_x(x,y,\theta) ) \, : \,  \varphi'_{\theta}(x,y,\theta) = 0 \} \end{equation}
where
\begin{equation}
\label{eq3.9.1}
\varphi(x,y,\theta) = \tfrac{2}{i}\left(F(x,\theta) - \Psi_0(y,\theta)\right), \ \ \
\Psi_0 \in \mathscr O ( \mathbb C^{2n} ), \ \ \Psi_0 ( x, \bar x ) = \Phi_0 ( x) ,
\end{equation}
and $F \in \mathscr O ( \neigh_{ \mathbb C^{2n} } (0))$, satisfies
${\rm det}\,F''_{x\theta} \neq 0,
$
and
\begin{equation}
\label{eq3.9.2}
- C \abs{x-\tilde{\kappa}(y)}^2 \leq  2{\rm Re}\, F(x,\overline{y}) - \Phi(x) - \Phi_0(y) \leq - \abs{x-\tilde{\kappa}(y)}^2/C ,
\end{equation}
where
$ \tilde{\kappa}(y) := \pi (\kappa (y, (2/i) \partial_y \Phi_0(y)))$ and
$ (x,y)\in {\rm neigh}_{\mathbb C^{2n}}(0)$.
\end{lemm}

\noindent
{\bf Remark.} When $ \Phi = \Phi_0 $ and $ \kappa$ and $ \tilde \kappa $
are the identity transformations on $ \mathbb C^{2n} $ and $ \mathbb C^n $ respectively, then
$ F ( x , \theta ) = \Psi_0 ( x , \theta ) $ and \eqref{eq3.9.2} is a standard consequence of the
strict plurisubharmonicity of $ \Phi_0 $. Of course in that case the obvious choice of $ \varphi $
for which \eqref{eq:graph} holds is given \eqref{eq:obv}.

We recall the following fundamental property of a generating function $ \varphi $ associated to $\kappa$ in (\ref{eq3.2}), (\ref{eq3.6}),
\begin{equation}
\label{eq:propgen}
\begin{gathered}
\nabla_{y,\theta} ( - \Im \varphi ( 0, y, \theta ) + \Phi_0 (y) ) |_{(y,\theta ) = (0,0)} = 0 , \\
\sgn \left( \nabla^2_{y,\theta }  ( - \Im \varphi ( 0, y, \theta ) + \Phi_0 (y) ) |_{(y,\theta ) = (0,0)} \right) = 0 .
\end{gathered}
\end{equation}
This implies that for $ x $ near $ 0 $, we have a unique critical point,
\[ \nabla_{y,\theta} ( - \Im \varphi ( x, y, \theta ) + \Phi_0 (y) ) |_{(y,\theta ) = (y_{\rm{c}} (x) , \theta_{\rm{c}} (x )  )} = 0 , \]
and
\[  - \Im \varphi ( x, y_{\rm{c}} (x), \theta_{\rm{c}} (x) ) + \Phi_0 (y_{\rm{c}} (x)) = \Phi ( x ) . \]
It also shows the existence of a good contour   $ \Gamma ( x ) $
passing through the critical point $(y_c(x),\theta_c(x))$, along which we have
\[ - \Im \varphi ( x, y, \theta ) + \Phi_0( y) - \Phi ( x ) \leq - |y-y_c(x)|^2 / C - |\theta - \theta_c(x)|^2/C,   \]
$ x \in \neigh_{\mathbb C^n} ( 0 ) $.
In the case of $ \varphi $ given in Lemma \ref{l:gen}, the unique critical point is given by $x = \tilde{\kappa}(y)$, $\theta = \bar y$, and we can simply take $ \Gamma(x)$ to be $x$--independent given by $\Gamma_V: y \mapsto \theta = \bar y $, $ y \in V $.


{In the notation of \eqref{eq3.9.2}, we have
$\tilde \kappa: V \rightarrow U$ is a diffeomorphism for some $V , U \subset \mathbb C^n $, small open neighborhoods of $0$.
Local Fourier integral operators taking $ H_{\Phi_0 } ( V )$ to $ H_{\Phi } (  U ) $ are defined as follows.
Let  $a = a(x,y,\theta;h) = a_0(x,y,\theta) + \mathcal O(h)$ be a classical analytic symbol in
$ \neigh_{\mathbb C^{3n} } ( 0 ) $. For $ u \in H_{\Phi_0 } ( V ) $ we put
\begin{equation}
\label{eq3.9.4}
\begin{split}
Au (x)  & = \frac{1}{(2\pi h)^{n}} \underset{\Gamma_V}{\int\!\!\!\int} e^{\frac{i}{h} \varphi ( x, y , \theta) }  a(x,y,\theta;h) u(y)\, dy\, d\theta \\
& = \frac{i^n}{(\pi h)^n} \int_V e^{\frac{2}{h} F(x,\overline{y})} a(x,y,\overline{y};h) u(y)\, e^{-2\Phi_0(y)/h}\, dL(y).
\end{split}
\end{equation}
As in \eqref{eq3.13},  $ Au \in H_\Phi ( U ) $ is well defined modulo errors in
$ H_{\Phi - \varepsilon} (U )  $, $ \varepsilon > 0 $.

The first formula for $ A u $ in (\ref{eq3.9.4}) makes sense for more general phase functions but we do not stress this point here. The advantage of the second representation based on Lemma \ref{l:gen} is the simplicity of the contour (integration over $ V \subset \mathbb C^n $) in the definition.

The properties of $ A $ are summarized in the following proposition. For the proofs, see
\cite[Theorem 4.5, Proposition 12.9]{sam},~\cite[Section III.7]{leb}.

\begin{prop}
\label{p:AB}
Suppose that $ V \Subset \mathbb C^n $ is a small open neighbourhood of $ 0$,
which is mapped diffeomorphically onto a small neighborhood $ U \Subset \mathbb C^n $ of $0$ by $ \tilde \kappa $, see \eqref{eq3.9.2}.
Then
\begin{equation}
\label{eq3.10}
{A = \mathcal O(1): H_{\Phi_0}(V) \rightarrow H_{\Phi}(U).}
\end{equation}
Moreover, if $ a $ in \eqref{eq3.9.4} satisfies $ a_0 ( 0 ) \neq 0 $, then there exists
a {classical} analytic symbol $ b $ with $ b_0 ( 0 ) \neq 0$ and $ \Gamma ( y) $, a good contour
for $ ( x, \theta ) \mapsto \Im \varphi(x,y,\theta) + \Phi(x) $, such that
\begin{equation}
\label{eq3.11}
(Bv)(y) := \frac{1}{(2\pi h)^n} \underset{\Gamma (y)}{\int\!\!\!\int} e^{-\frac{i}{h}\varphi(x,y,\theta)} b(x,y,\theta;h) v(x)\, dx\, d\theta,
\end{equation}
is a microlocal inverse of $ A $ in the sense that
\begin{equation}
{B = \mathcal O(1): H_{\Phi}(U) \rightarrow H_{\Phi_0}(V),}
\end{equation}
and for each $W_1 \Subset V$, $  W_2 \Subset U $, there exists $\eta > 0$ such that
\begin{equation}
\label{eq3.12}
\begin{split}
& I - AB = \mathcal O(1) e^{-\eta/h}: H_{\Phi}(U) \rightarrow H_{\Phi}(W_2), \\
& I - BA = \mathcal O(1) e^{-\eta/h}: H_{\Phi_0}(V) \rightarrow H_{\Phi_0}(W_1 ).
\end{split}
\end{equation}
\end{prop}

The composition of Fourier integral operators in the complex domain is presented in
~\cite[Chapter 4]{sam},~\cite[Section 2.5]{HiSj15}, In particular, in the notation above,
\begin{equation}
\label{eq3.18}
{P \circ A: H_{\Phi_0}(V) \rightarrow H_{\Phi}(U)}
\end{equation}
is a Fourier integral operator of the form \eqref{eq3.9.4}, with an  analytic symbol given by
$$
\frac{1}{(2\pi h)^n} e^{-i\varphi(x,y,\theta)/h} P \left(e^{i\varphi(\cdot, y, \theta)/h} a(\cdot, y, \theta;h)\right)(x).
$$
Assuming that $A$ is elliptic and letting $B$ be a microlocal inverse of $A$ from Proposition
\ref{p:AB}, we also obtain that
\begin{equation}
\label{eq3.19}
Q := B \circ P \circ A: H_{\Phi_0}(V) \rightarrow H_{\Phi_0}(V)
\end{equation}
is an  analytic $h$–pseudodifferential operator, with the principal symbol given by
\begin{equation}
\label{eq3.20}
q = p \circ \kappa.
\end{equation}
This is the form of Egorov's theorem we will use.

\section{A general semiclassical result}
\label{s:gen}

We present a general semiclassical result in one dimension generalizing the result in \cite{Hi04} by allowing more general behaviour at infinity and, perhaps more importantly,  describing eigenvalues with exponential accuracy in $ h$, uniformly in a neighbourhood of a non-degerate minimum of the classical observable.
Although not surprising, we do not know of a detailed reference for this even in the self-adjoint case.

Hence, we assume that $ n = 1$, and that $ H^{\rm{w}} ( x, h D; h ) $ is a semiclassical
pseudodifferential operator satisfying the following assumptions for some order function $ m \geq 1 $
(see \cite[\S 4.4]{z12} or \eqref{eq2.1.1}):
\begin{equation}
\label{eq:assH}
\begin{gathered}
H \in S ( \mathbb R^{2}, m ) , \ \ H(\rho;h)  \sim H_0 ( \rho ) + h H_1 ( \rho)  + \cdots,  \ \
\rho = ( x, \xi ), \ \ H_0 (0) = 0,   \\
\text{ $ d H_0 ( 0) = 0 $, $ d^2 H_0 ( 0 ) $ is elliptic, and
$ \{ \langle d^2 H_0(0) v, v \rangle; v \in \mathbb R^2 \} \neq \mathbb C$, }\\
\text{ $ H  ( \rho; h ) $ extends to $ \{ |\Im \rho | < c_0 \} \subset \mathbb C^2 $, $ c_0 > 0 $,
as an analytic symbol, so that }\\
\abs{H_j(\rho)} \leq C^{j+1} j^j\, m({\rm Re}\,\rho ),\quad |\Im \rho | < c_0,\,\, \text{and} \\
\text{ $ \forall \, \varepsilon > 0 \, \exists \,  \delta > 0 \
| \rho | > \varepsilon \ \Longrightarrow \ | H_0 ( \rho) | \geq \delta m ( \rho) $, $ \rho \in
\mathbb R^2.$ }
\end{gathered}
\end{equation}
Here the ellipticity of $ d^2 H_0 ( 0 ) $ means that $ \langle d^2 H_0 (0 ) v , v\rangle \neq 0 $
for all $ v \in \mathbb R^2 \setminus \{0\} $. This and the last condition on $ d^2 H_0 ( 0) $ above imply that there exists
$ \lambda \in \mathbb C $ such that $ \Re \langle \lambda  d^2 H_0(0) v, v \rangle > 0$, $ v \in \mathbb R^2 \setminus \{0\} $, see  Lemma \ref{l:quand}.

The conditions on the quadratic form $ \mathbb R^2 \ni v \mapsto \langle d^2 H_0 ( 0 ) v, v \rangle \in \mathbb C$ imply that there exists $ \theta_0 \in [0, 2 \pi ) $ such that
$|\langle d^2 H_0 ( 0 ) v, v \rangle + e^{ i \theta_0 }| \geq 1 + |v|^2/C$, $ v \in \mathbb R^2 $. That
in turn implies that if we take $ \chi \in C_0^\infty ( \mathbb R^2 , [0, 1 ])$,
$ \chi|_{ B_{\mathbb R^2 } ( 0 ,  \varepsilon )} = 1$,
then $| H_0(\rho) + (\delta/2) e^{ i \theta_0} \chi ( \rho)|
\geq m ( \rho )/{\mathcal O(1)}  $, $\rho \in \mathbb R^2$, if $ \varepsilon > 0$ is small enough and $ \delta $ is as in \eqref{eq:assH}. Similarly we get $| H_0(\rho) + (\delta/2) e^{ i \theta_0}| \geq m ( \rho )/{\mathcal O(1)}  $, $\rho \in \mathbb R^2$. Here we also use that $ m \geq 1 $. Hence, $ H^{\rm{w}} + \chi^{\rm{w}} - z $ and $ H^{\rm{w}} + (\delta/2) e^{ i \theta_0}-z  $ are invertible for small $ h $ and $ |z |$
(see \cite[Theorem 4.29]{z12}). Then
\[  \neigh_{\mathbb C} ( 0) \ni z \mapsto ( H^{\rm{w}}  - z)^{-1} = ( H^{\rm{w}}   +
\chi^{\rm{w}}  - z  )^{-1} ( I - \chi^{\rm{w}}  ( H^{\rm{w}}   +
\chi^{\rm{w}} - z  )^{-1}  )^{-1}  \]
is meromorphic in $ z$: $ T(z)  :=  \chi^{\rm{w}}  ( H^{\rm{w}}   +
\chi^{\rm{w}} - z  )^{-1} $ is compact (see \cite[Theorem 4.28]{z12}) and
the invertibility of $ H^{\rm{w}}  + (\delta/2) e^{ i \theta_0}  $ shows that $ I - T (-\frac{\delta e^{ i \theta_0}}{2})
$ is invertible. That means that analytic Fredholm theory applies \cite[\S D.4]{z12} and indeed
$ ( H^{\rm{w}} - z )^{-1} $ is meromorphic near $ 0 $. Considering $ H^{\rm{w}} $ as a closed unbounded operator on
$ L^2(\mathbb R)$, with the domain given by $ H_h ( m) $ (see \cite[\S 8.3]{z12}), we conclude that the spectrum of $ H^{\rm{w}} $ is the set of poles with the usual formula for multiplicity.

A typical example with $m(x,\xi) = 1 + |\xi|^2$ is given by $ H ( x, \xi ) = H_0 ( x, \xi ) + h^2 H_2 ( x, \xi )$,
where $ H_0 ( x, \xi ) = \xi^2 + W_0 ( x ) $, $ H_2 = W_1 ( x ) $, where $ W_j $
are holomorphic bounded in a strip around $ \mathbb R \subset \mathbb C$,
{such that $W_0^{-1}(0) = \{0\}$, $W'_0(0) = 0$, $W_0''(0) \notin (-\infty,0]$} (recall that $ W_j $ may be complex valued) -- see \S \ref{s:complexi}. The ellipticity assumption here means that
$ | H_0 ( x, \xi ) | \geq \delta ( 1 + |\xi|^2 ) $ for $ |(x,\xi)| > \varepsilon $.

Our general result is given in
\begin{theo}
\label{t:gen}
Suppose that $ H $ satisfies \eqref{eq:assH}. Then there exist constants $ r_0, h_0, c_0  > 0 $ such that
for $ 0 < h < h_0 $,
\begin{equation}
\label{eq:tgen}\begin{gathered}  \Spec ( H^{\rm{w}} ( x, h D; h ) ) \cap D ( 0 , r_0 ) = \{  \lambda_n ( h );
n = 0 , 1, \cdots \} \cap D( 0, r_0 ) , \\
\lambda_n ( h ) = G ( 2 \pi ( n + \tfrac12 ) h; h ) + \mathcal O ( e^{ - c_0/ h } ) ,
\end{gathered}
\end{equation}
where $ G ( z; h ) $ has an asymptotic expansion, $  G_0 ( z ) + h G_1 ( z ) + \cdots $,
$ G_j \in \mathscr O ( \neigh_{\mathbb C} ( 0 ) ) $, $ |G_j | \leq A^{j+1} j! $,
in the sense that {there exists $B>0$ such that for all $N \in \mathbb N$, we have}
\begin{equation}
\label{eq:Gzh} | G ( z; h ) - \sum_{ j=0}^{N-1} h^j G_j ( z ) | \leq B^{N+1} N! h^N  . \end{equation}
The leading term $ G_0 $ is the inverse of a complex action defined in Theorem {\rm \ref{t:action}}.
When $ H_1 \equiv 0$ (see \eqref{eq:assH})  then $ G_1 \equiv 0 $.
\end{theo}

In view of \eqref{eq:tgen} and \eqref{eq:Gzh} the theorem
gives an exponentially accurate description of eigenvalues of $ H^{\rm{w}} $ near $ 0$.
To obtain $ G ( z; h ) $ satisfying \eqref{eq:Gzh} from $ G_j $'s one can proceed in different ways.
One was used in \eqref{eq:reali}, see \cite[\S 2.2]{HiSj15}, and here we present a variant
in the spirit of exact WKB:
\begin{equation}
\label{eq:Borel} G( z, h ) = h^{-1} \int_0^\varepsilon  g ( z, t ) e^{ -t / h } dt , \ \ \ g ( z, h ) :=
\sum_{ j=0}^\infty \frac{ G_j ( z ) } {j!} h^j . \end{equation}
where $0 < \varepsilon < 1/ A $ is smaller that the radius of covergence of $ h \mapsto g ( z, h ) $.
The choice of $ \varepsilon $ produces exponentially small ambiguity, just as we saw in \eqref{eq:reali}.



The first step in the proof of Theorem \ref{t:gen} consists of passing to the FBI transform side by means of a suitable metaplectic FBI-Bargmann transform. When doing so, we write, in the notation of \eqref{eq:assH},
\begin{equation}
\label{eq4.1.0.1}
H_0(y,\eta) = q (y,\eta) + \mathcal O((y,\eta)^3), \quad (y,\eta) \rightarrow (0,0),
\end{equation}
where $ q $ is the quadratic form corresponding to $ d^2 H_0(0) $. Lemma \ref{l:quand} shows that
${\rm Re}\, (\lambda q )$ is positive definite on $\mathbb R^2$, for some $ \lambda \in \mathbb C$.  It follows from \cite[Proposition 2.4]{Hi04},~\cite[Proposition 2.1]{HiSjVi} that there exists a metaplectic unitary FBI transform, i.e. an FBI transform with a quadratic phase and a constant amplitude,
\begin{equation}
\label{eq4.1.0.2}
T: L^2(\mathbb R) \rightarrow H_{\Phi_1}(\mathbb C),
\end{equation}
such that the conjugated operator $T\circ H^{\rm w}(y,hD_y;h) \circ T^{-1}$ acting on $H_{\Phi_1}(\mathbb C)$, is of the form
\begin{equation}
\label{eq4.1.0.3}
T\circ H^{\rm w}(y,hD_y;h) \circ T^{-1} = P^{\rm w}(x,hD_x;h), \,\, P(x,\xi;h) \sim p_0(x,\xi) + hp_1(x,\xi) + \ldots,
\end{equation}
where
\begin{equation}
\label{eq4.1.0.4}
p_0(x,\xi) = H_0(\kappa_T^{-1}(x,\xi)) = \mu x\xi + \mathcal O((x,\xi)^3),   \ \ \ {\rm Im}\, (\lambda \mu) >0.
\end{equation}
Here $\kappa_T$ is the complex linear canonical transformation associated to $T$ and the quadratic weight $\Phi_1$ in (\ref{eq4.1.0.2}) is strictly convex. See also \cite[Theorem 1.4.2]{HiSj15}, \cite[Theorem 13.9]{z12}.  The structure and properties of the conjugated operator $ P^{\rm w} ( x, h D; h ) $ are reviewed in \S \ref{s:global}, and in particular, $P(x,\xi;h)$ is a holomorphic function in a tubular neighbourhood of $\Lambda_{\Phi_1}$ in $\mathbb C^2$, satisfying (\ref{eq2.1.2}), (\ref{eq2.2}), and (\ref{eq2.23}).

Following \cite{N}, \cite{Hi04}, we shall now modify the weight function $\Phi_1$ in a bounded region of $\mathbb C$, so that the modified weight agrees with the standard radial weight
\begin{equation}
\label{eq4.1}
\Phi_0(x) = \tfrac12 {\abs{x}^2}
\end{equation}
in a small but fixed neighbourhood of the origin, while keeping the ellipticity of $ p_0 $ away from $0$ (see the last condition in
\eqref{eq:assH}). This reduction is very convenient in what follows, as monomials (the eigenfunctions of $ x h D_x $) are then orthogonal in
$ L^2_{\Phi_0}(\mathbb C)$.
\begin{prop}
\label{prop:radial}
For every sufficiently small open neighborhood $U \subset \mathbb C$ of $0$ there exists
a strictly convex weight  $\widetilde{\Phi} \in C^{\infty}(\mathbb C)$ satisfying
\[ \widetilde{\Phi}|_{\mathbb C \setminus U }  = \Phi_1|_{\mathbb C \setminus U }, \ \
\widetilde{\Phi}|_{U_0 }  = \Phi_0|_{U_0 },  \ \
U_0 = \neigh_{\mathbb C} ( 0 ) \Subset U, \ \ \ \nabla^2 \widetilde{\Phi} \in L^{\infty}(\mathbb C),
\]
and such that
\begin{equation}
\label{eq4.1.1}
\forall \, \theta > 0 \ \exists \, \delta >0 \ \  (x,\xi) \in \Lambda_{\widetilde{\Phi}}, \,\,\, \abs{(x,\xi)} \geq \theta \Longrightarrow \abs{p_0(x,\xi)} \geq \delta m(x).
\end{equation}
\end{prop}
\begin{proof}
Following an argument of~\cite{N}, we let $\psi \in C^{\infty}(\mathbb R; [0,1])$ be such that $\psi(t) = 1$, $t\leq 1$, $\psi(t) = 0$, $t\geq 2$, and notice that the $C^{\infty}_0(\mathbb C)$--function $\chi_{\varepsilon}(x) = \psi(\varepsilon \log \abs{x})$ satisfies for $\varepsilon \in (0,1]$,
\begin{equation}
\label{eq4.1.1.1}
\nabla \chi_{\varepsilon}(x) = \mathcal O\left({\varepsilon} |x|^{-1} \right), \quad \nabla^2 \chi_{\varepsilon}(x) = \mathcal O \left({\varepsilon}{\abs{x}^{-2}}\right).
\end{equation}
It follows from (\ref{eq4.1.1.1}) that the $C^{\infty}(\mathbb C)$--function
\begin{equation}
\label{eq4.1.1.2}
\Phi_{\varepsilon,\eta}(x) = \chi_{\varepsilon}\left({e^{2/\varepsilon}x}/{\eta}\right)\Phi_0(x) + \left(1-\chi_{\varepsilon}\left({e^{2/\varepsilon}x}/{\eta}\right)\right)\Phi_1(x)
\end{equation}
is strictly convex for $\varepsilon > 0$ sufficiently small, uniformly in $\eta > 0$, and we have
\begin{equation}
\label{eq4.1.1.3}
\Phi_{\varepsilon,\eta}(x) = \Phi_0(x), \,\, \abs{x} < \eta\, e^{-1/\varepsilon}, \quad \Phi_{\varepsilon,\eta}(x) = \Phi_1(x), \,\, \abs{x} > \eta,
\end{equation}
Furthermore,
\[
\norm{\nabla(\Phi_{\varepsilon,\eta} - \Phi_1)}_{L^{\infty}(\mathbb C)} \leq \mathcal O(\eta),\quad \nabla^2 \Phi_{\varepsilon, \eta} \in L^{\infty}(\mathbb C),
\]
uniformly in $\varepsilon \in (0,1]$, $\eta > 0$. Using \eqref{eq2.1} we define $ \Lambda_{\Phi_{\varepsilon,\eta}} $ which coincides with $\Lambda_{\Phi_0}$ in a neighbourhood of the origin and agrees with $\Lambda_{\Phi_1}$ away from another fixed neighbourhood of $0$.
In view of (\ref{eq4.1.1.1}), (\ref{eq4.1.1.2}) we have
\begin{equation}
\label{eq4.1.1.4}
\begin{split}
\xi(x) & := \tfrac{2}{i}\partial_x \Phi_{\varepsilon,\eta}(x) \\
& = \chi_{\varepsilon}\left({e^{2/\varepsilon}x}/{\eta}\right)
\tfrac{2}{i}{\partial_x \Phi_0} (x) + \left(1-\chi_{\varepsilon}\left({e^{2/\varepsilon}x}/{\eta}\right)\right)\tfrac{2}{i}{\partial_x \Phi_1}(x) + \mathcal O(\varepsilon \abs{x}),
\end{split}
\end{equation}
uniformly in $\eta>0$, and therefore, for $\varepsilon > 0$ small enough,
\begin{equation}
\label{eq4.1.1.5}
\begin{split}
{\rm Re}\,(ix\xi(x)) & = \chi_{\varepsilon}\, {\rm Re}\, (2x\, \partial_x \Phi_0(x)) + (1-\chi_{\varepsilon})\, {\rm Re}\, (2x\, \partial_x \Phi_1(x)) + \mathcal O(\varepsilon \abs{x}^2) \\
& = \chi_{\varepsilon} \langle{x,\nabla \Phi_0(x)\rangle}_{\mathbb R^2} + (1-\chi_{\varepsilon})\langle{x,\nabla \Phi_1(x)\rangle}_ {\mathbb R^2} + \mathcal O(\varepsilon \abs{x}^2) \asymp \abs{x}^2,
\end{split}
\end{equation}
uniformly in $\eta>0$. Taking $\eta > 0$ sufficiently small but fixed, we obtain, using (\ref{eq4.1.0.4}), (\ref{eq4.1.1.5}),
\begin{equation}
\label{eq4.1.1.6}
\abs{p_0(x,\xi(x))} \asymp \abs{x}^2,\,\,x\in {\rm neigh}_{\mathbb C}(0).
\end{equation}
Recalling also (\ref{eq:assH}), (\ref{eq4.1.1.3}), we obtain the following ellipticity property along $\Lambda_{\Phi_{\varepsilon,\eta}}$: $\forall \theta > 0$ $\exists \delta >0$ such that
\begin{equation}
\label{eq4.1.1.7}
x\in \mathbb C,\,\, \abs{x} \geq \theta  \Longrightarrow \abs{p_0(x,\xi(x))} \geq \delta m(x).
\end{equation}
It follows that we can take $\widetilde{\Phi} = \Phi_{\varepsilon, \eta}$, for $\varepsilon > 0$, $\eta > 0$ small enough fixed, and this completes the proof.
\end{proof}

Proposition \ref{prop:radial} and the results of \S \ref{s:global}  show that
\begin{equation}
\label{eq4.1.1.8}
P^{\rm w}(x,hD_x;h) = \mathcal O(1): H_{\widetilde{\Phi},m}(\mathbb C) \rightarrow H_{\widetilde{\Phi}}(\mathbb C).
\end{equation}
The spectral analysis of $P^{\rm w}$ in the proof of Theorem \ref{t:gen} is carried out in a weighted space obtained from $H_{\widetilde{\Phi}}(\mathbb C)$ by an additional modification of the weight in a small but fixed neighbourhood of the origin, implementing the holomorphic canonical transformation $\kappa: {\rm neigh}_{\mathbb C^2}(0) \rightarrow {\rm neigh}_{\mathbb C^2}(0)$ given in Theorem \ref{t:vey}, with $p(x,\xi) = p_0(x,\xi)$ in (\ref{eq4.1.0.4}). We have $\kappa(\rho) = \rho + \mathcal O(\rho^2)$, and it follows that the I-Lagrangian R-symplectic manifold $\kappa({\rm neigh}_{\Lambda_{\Phi_0}}(0)) \subset \mathbb C^2$, where $\Phi_0$ is the standard weight function given in (\ref{eq4.1}), is of the form
\begin{equation}
\label{eq4.1.1.9}
\kappa({\rm neigh}_{\Lambda_{\Phi_0}}(0)) = {\rm neigh}_{\Lambda_{\Phi}}(0).
\end{equation}
Here $\Phi$ is real analytic near $0\in \mathbb C$ and we may assume that $\Phi(0) = 0$. It follows therefore that $\Phi(x) = \mathcal O(\abs{x}^2)$ as $x\rightarrow 0$, and using also that
\[
T_{0}\Lambda_{\Phi} = d\kappa(0)(\Lambda_{\Phi_0}) = \Lambda_{\Phi_0},
\]
we obtain a more precise description of the deformed weight,
\begin{equation}
\label{eq4.1.1.10}
\Phi(x) = \Phi_0(x) + \mathcal O(x^3).
\end{equation}
In particular, $\Phi$ is strictly convex near $0\in \mathbb C$. We then define, using the notation in the proof of Proposition \ref{prop:radial},
\begin{equation}
\label{eq4.1.1.11}
\widehat{\Phi}(x) = \chi_{\varepsilon}\left({e^{2/\varepsilon}x}/{\eta}\right)\Phi(x) + \left(1-\chi_{\varepsilon}\left({e^{2/\varepsilon}x}/{\eta}\right)\right)\Phi_1(x),
\end{equation}
for $\varepsilon >0$, $\eta > 0$ small enough fixed. It follows from (\ref{eq4.1.1.1}), (\ref{eq4.1.1.10}) that the function $\widehat{\Phi} \in C^{\infty}(\mathbb C)$ is strictly convex, satisfying $\widehat{\Phi}(x) = \Phi_1(x)$ for $\abs{x} \geq \eta$, $\widehat{\Phi}(x) = \Phi(x)$ for $\abs{x} < \eta e^{-1/\varepsilon}$, with $\norm{\nabla(\widehat{\Phi} - \Phi_1)}_{L^{\infty}(\mathbb C)}$ small enough, $\nabla^2 \widehat{\Phi} \in L^{\infty}(\mathbb C)$. We also get, using (\ref{eq4.1.0.4}), (\ref{eq4.1.1.5}), and (\ref{eq4.1.1.10}), that $\forall \theta > 0$ $\exists \delta >0$ such that
\begin{equation}
\label{eq4.1.1.1.12}
(x,\xi) \in \Lambda_{\widehat{\Phi}}, \,\,\, \abs{(x,\xi)} \geq \theta \Longrightarrow \abs{p_0(x,\xi)} \geq \delta m(x).
\end{equation}
This discussion shows that
\begin{equation}
\label{eq4.1.1.1.13}
P^{\rm w}(x,hD_x;h) = \mathcal O(1): H_{\widehat{\Phi},m}(\mathbb C) \rightarrow H_{\widehat{\Phi}}(\mathbb C),
\end{equation}
and it follows from the ellipticity property (\ref{eq4.1.1.1.12}) and the proof of Proposition \ref{elliptic} that for each $\varepsilon > 0$ there exist $h_0, \eta, \gamma , C  > 0$ such that for all $0 < h < h_0 $, $\abs{z} < \gamma$,  we have
\begin{equation}
\label{eq4.1.2}
\| u \|_{L^2_{\widehat{\Phi},m}(\abs{x} > \varepsilon)} \leq C \| (P^{\rm{w}}(x,hD_x;h)-z)u\| _{H_{\widehat{\Phi}}(\{\abs{x} > \varepsilon/4\})} +  C e^{-\eta/h} \| u \|_{H_{\widehat{\Phi},m}(\mathbb C)} ,
\end{equation}
for all $u\in H_{\Phi,m}(\mathbb C)$.

Let $U \Subset \mathbb C$ be a small open neighbourhood of the origin such that $\widehat{\Phi} = \Phi$ in a neighbourhood of $\overline{U}$. There exists a classical analytic $h$--pseudodifferential operator
\begin{equation}
\label{eq4.1.3}
P = P_U = \mathcal O(1): H_{\Phi}(U) \rightarrow H_{\Phi}(U),
\end{equation}
with the leading symbol $p_0$ in (\ref{eq4.1.0.4}), realized using a good contour of the form (\ref{eq3.17}), with $\Psi = \Phi$, such that for some $\delta > 0$, we have for all $u\in H_{\widehat{\Phi},m}(\mathbb C)$,
\begin{equation}
\label{eq4.1.4}
\| P^{\rm{w}}(x,hD_x;h)u - P_U(u|_U)\|_{H_{\Phi}(U)} \leq \mathcal O(e^{-\delta/h}) \| u \|_{H_{\widehat{\Phi},m}(\mathbb C)}.
\end{equation}
We may view $P_U$ as a local realization of the globally defined operator $P^{{\rm w}}(x,hD_x;h)$ in (\ref{eq4.1.1.1.13}), and the key point in the proof of Theorem \ref{t:gen} is the following precise normal for $ P = P_ U$.
\begin{prop}
\label{p:norm}
Suppose that $ P = P_U $ is defined above, so that in particular {\rm \eqref{eq4.1.3}} holds. Then there exist a small open disc $V = D(0,r) \Subset \mathbb C$ centered at the origin and an elliptic analytic Fourier integral operator $ A = \mathcal O(1): H_{\Phi_0}(V) \rightarrow H_{\Phi}(U)$, of the form {\rm (\ref{eq3.9.4})}, with a microlocal inverse $B = \mathcal O(1): H_{\Phi}(U) \rightarrow H_{\Phi_0}(V)$, such that the conjugated operator $Q = BPA: H_{\Phi_0}(V) \rightarrow H_{\Phi_0}(V)$ is of the form
\begin{equation}
\label{eq4.1.6}
Q = G( 2 i x h D_x + h; h ) + R,
\end{equation}
where $G = G(w;h)$ is a classical analytic symbol in a neighbourhood of $0\in \mathbb C$ such that $G_0(0) = 0$, $G'_0(0) \neq 0$,
$G(2 i xhD_x + h ;h)$ is an analytic $h$--pseudodifferential operator defined as in the remark after the statement of Lemma {\rm \ref{l:F2f}}, and
for each $W \Subset V$ there exists $\delta = \delta_W > 0$ such that
\begin{equation}
\label{eq4.1.6.5}
R = \mathcal O(e^{-\delta/h}) : H_{\Phi_0}(V) \rightarrow H_{\Phi_0}(W).
\end{equation}
\end{prop}
\begin{proof}
Let $A_0 = \mathcal O(1): H_{\Phi_0}(V) \rightarrow H_{\Phi}(U)$ be an elliptic analytic Fourier integral operator of the form (\ref{eq3.9.4}), associated to the holomorphic canonical transformation $\kappa: {\rm neigh}_{\mathbb C^2}(0) \rightarrow {\rm neigh}_{\mathbb C^2}(0)$ introduced in Theorem \ref{t:vey} with $p = p_0$ in (\ref{eq4.1.0.4}), such that (\ref{eq4.1.1.9}) holds. Let $B_0$ be a microlocal inverse of $A_0$ -- see Proposition \ref{p:AB}.

It follows from \eqref{eq3.19} that $Q_0 := B_0 P A_0: H_{\Phi_0}(V) \rightarrow H_{\Phi_0}(V)$ is a classical analytic $h$--pseudodifferential operator with the principal symbol given by $q_0(x,\xi) = p_0(\kappa(x,\xi)) = g(x\xi)$, for some $g \in \mathscr O ( \neigh_{\mathbb C }  (0 ) ) $, $g(0) = 0$, $g'(0) \neq 0$. Furthermore, Proposition \ref{p:ave} shows that there exists an analytic $h$--pseudodifferential operator of the form $e^{A_1}$, where $A_1$ is an analytic $h$--pseudodifferential operator of order $0$, such that the conjugated operator $e^{-A_1}\circ Q_0\circ e^{A_1}: H_{\Phi_0}(V) \rightarrow H_{\Phi_0}(V)$ is of the form (\ref{eq4.1.6}), (\ref{eq4.1.6.5}). We obtain the desired Fourier integral operator by putting  $A := A_0 \circ e^{A_1}$.
\end{proof}

 We now define
\begin{equation}
\label{eq4.1.8}
\varphi_j(x) = (\pi j!)^{-1/2} h^{-\frac12(j+1)} {x}^j, \quad j=0,1,2,\ldots
\end{equation}
noting that they form an orthonormal basis in the Bargmann space $H_{\Phi_0}(\mathbb C)$. To $ \varphi_j $'s we associate quasi-eigenvalues of (\ref{eq4.1.6}):
\begin{equation}
\label{eq4.1.9}
\left\{ G\left((2j +1)h;h\right);\,j\in \mathbb N\right\} \cap {\rm neigh}_{\mathbb C}(0).
\end{equation}
The first step is a localization result which shows that the spectrum of $P^{\rm{w}} $ in (\ref{eq4.1.1.1.13}) lies near these quasi-eigenvalues:

\begin{prop}
\label{p:warm_up}
There exist $ r_0, c_0 , h_0 > 0 $ such that
for
\begin{equation}
\label{eq:condz}  z \in D ( 0 , r_0 ),  \ \ \ | z - G ( ( 2j + 1 ) h; h) | \geq c_0 h, \ \ \ 0 < h < h_0 ,
\end{equation}
the operator
\begin{equation}
\label{eq:4.20}
P^{\rm{w}}(x,hD_x;h) - z: H_{\widehat{\Phi},m}(\mathbb C) \rightarrow H_{\widehat{\Phi}}(\mathbb C)
\end{equation}
is invertible. In particular, considering $ P^{\rm{w}}(x,hD_x;h)$ as a closed densely defined operator on $ H_{\widehat{\Phi}}(\mathbb C)$ with
the domain given by $ H_{\widehat{\Phi},m}(\mathbb C) $, we have
\begin{equation}
\label{eq:warm_up} \Spec (P^{\rm{w}}(x,hD_x;h) ) \cap D ( 0 , r_0 ) \subset \bigcup_j D (
 G ( ( 2j + 1 ) h, h ), c_0 h ) .
 \end{equation}
 \end{prop}
The proof is quite involved and we start with a few technical lemmas.
The first comes from the work of G\'erard--Sj\"ostrand \cite[Lemma 4.5]{GeSj} and is natural when
studying approximation by holomorphic polynomials \eqref{eq4.1.8}. Since the proof is simple we recall it for the convenience of the reader.
\begin{lemm}
\label{Taylor}
Suppose that $u$ is holomorphic in a neighbourhood of $\overline{D(0,1)}$ satisfying $u^{(j)}(0) = 0$ for $j<N$. Then for $ 0 < \rho < 1 $ we have
\begin{equation}
\label{eq4.2}
\| u \|_{L^{\infty}(D(0, \rho))} \leq N ( 1 - \rho)^{-1} \rho^N \| u \|_{L^{\infty}(D(0, 1))}.
\end{equation}
\end{lemm}
\begin{proof}
We have by Taylor's formula,
$$
u(x)  = \int_0^1 \frac{(1-t)^{N-1}}{(N-1)!} \left(\frac{d}{dt}\right)^N u(tx)\,dt.
$$
To estimate $ u ( x ) $ for $ |x| \leq \rho $, we may assume that $\abs{x}=\rho $. Cauchy's estimates
applied to the holomorphic function
$
D\left(0,1/\rho \right) \ni \zeta \mapsto u(\zeta x)
$
give
$$
\abs{\left(\frac{d}{dt}\right)^N u (tx)}  \leq \frac{N!}{(1/\rho - t)^N}\| u \|_{L^{\infty}(D(0,1))}.
$$
Since
$$
N \int_0^1 \frac{(1-t)^{N-1}}{(1/\rho - t)^N}\,dt
= N \rho^{N} \int_0^1 \frac{(1 - t)^{N-1}} { (1 - \rho t )^N} dt
\leq \frac{N}{1 - \rho}\rho^{N},
$$
we obtain \eqref{eq4.2}.
\end{proof}

The next lemma provides estimates on elements of $ H_{\Phi_0} ( V ) $ orthogonal to
$ x^j$'s:
\begin{lemm}
\label{a_priori}
Put $V = D(0,r)$, $r>0$, and define $  N ( h ) $ as the least integer such that
\begin{equation}
\label{eq4.2.5}
N(h) \geq (\tfrac 1 8 r^2 + 2\delta ) h^{-1}, \quad \delta > 0.
\end{equation}
Then for $u\in H_{\Phi_0}(V)$ and $ \varphi_j $ in \eqref{eq4.1.8} we have
\begin{equation}
\label{eq4.2.6}
(u,\varphi_j)_{H_{\Phi_0}(V)} = 0, \ \ j < N(h) \ \Longrightarrow \
\| u \|_{H_{\Phi_0}(D(0,r/16))} \leq C e^{-\delta/h} \| u \|_{H_{\Phi_0}(D(0,r))}.
\end{equation}
\end{lemm}
\begin{proof}
The monomials
\begin{equation}
\label{eq4.1.17}
f_j = {\varphi_j}/{\|{\varphi_j}\|_{H_{\Phi_0}(V)}},\quad j = 0,1,2\ldots,
\end{equation}
form an orthonormal basis in $H_{\Phi_0}(V)$, and we have for $u\in H_{\Phi_0}(V)$ (recall that
$ V $ is a disc),
\begin{equation}
\label{eq4.1.18}
(u,f_j)_{H_{\Phi_0}(V)} = \frac{u^{(j)}(0)}{j!}\, (x^j, f_j)_{H_{\Phi_0}(V)}, \quad j = 0,1,2\ldots.
\end{equation}
If  $ (u,f_j)_{H_{\Phi_0}(V)}  = 0 $, $ j < N $, and $ 0 < r_0 < r_1 < r_2 < r $,  then Lemma \ref{Taylor}
(applied with $ \rho = r_0/r_1 $ after a natural rescaling) and (\ref{eq4.1.18}) give,
\begin{equation}
\label{eq4.2.1}
\begin{split}
\| u \|^2_{H_{\Phi_0}(D(0,r_0))} & \leq \| u \|^2_{L^2(D(0,r_0))} \leq C r_0^2 \| u \|_{L^{\infty}(D(0,r_0))}^2 \\
& \leq  C r_1^2   N^2 ( 1 - r_0/r_1 )^{-2} (r_0/r_1)^{2N}  \| u \|_{L^{\infty}(D(0,r_1))}^2 \\
& \leq C N^2 r_1^2
( r_2 - r_1)^{-2}  ( 1 - r_0/r_1 )^{-2} (r_0/r_1)^{2N}   \| u \|_{L^{2}(D(0,r_2))}^2 \\
& \leq C N^2   r_1^2
( r_2 - r_1)^{-2}  ( 1 - r_0/r_1 )^{-2} (r_0/r_1)^{2N}  e^{r_2^2/h} \| u \|^2_{H_{\Phi_0}(D(0,r_2))}.
\end{split}
\end{equation}
Taking $r_0 = r/16$, $r_1 = er/16$, $r_2 = r/2$ gives
\begin{equation}
\label{eq4.2.2}
\| u \|^2_{H_{\Phi_0}(D(0,r/16))} \leq CN^2  e^{-2 N} e^{r^2/4 h} \| u \|^2_{H_{\Phi_0}(V)}.
\end{equation}
Choosing $ N $ satisfying \eqref{eq4.2.5} gives \eqref{eq4.2.6}.
\end{proof}

The final preparatory lemma concerns exponential decay of $ \varphi_j$'s away from $ 0 $:
\begin{lemm}
\label{decay_exp}
For $\rho > 0$ and $ \varphi_j $ given in \eqref{eq4.1.8},
\begin{equation}
\label{eq4.3}
j +1 \leq \frac{ \rho^2}{ 4 h} \ \Longrightarrow \
\int_{\abs{x} > \rho} \abs{\varphi_j(x)}^2 e^{-2\Phi_0(x)/h}\, L(dx) \leq e^{-\rho^2/4h}.
\end{equation}
\end{lemm}
\begin{proof}
The integral in (\ref{eq4.3}) is equal to
\begin{multline}
\label{eq4.4}
\frac{1}{\pi j!} \frac{1}{h^{1+j}}\int_{\abs{x} > \rho} \abs{x}^{2j} e^{-\abs{x}^2/h}\, L(dx) = \frac{1}{j!}\frac{1}{h^{1+j}}
\int_{\rho}^{\infty} 2r^{2j +1} e^{-r^2/h}\, dr \\
= \frac{1}{j!}\frac{1}{h^{1+j}} \int_{\rho^2}^{\infty} y^{j} e^{-y/h}\, dy \leq e^{-\rho^2/2h} \frac{1}{j!}\frac{1}{h^{1+j}}
\int_{0}^{\infty} y^{j} e^{-y/2h}\, dy = e^{-\rho^2/2h}\, 2^{j+1},
\end{multline}
which gives the estimate in (\ref{eq4.3}).
\end{proof}

We now move to
\begin{proof}[Proof of Proposition \ref{p:warm_up}]
We consider
\begin{equation}
\label{eq4.1.10}
(P^{\rm{w}}(x,hD_x;h) -z)u = v, \ \ \ u\in H_{\widehat{\Phi},m}(\mathbb C), \ v\in H_{\widehat{\Phi}}(\mathbb C), \ \ \ z\in {\rm neigh}_{\mathbb C}(0) .
\end{equation}
Restricting the attention to the neighbourhood $U$ in (\ref{eq4.1.3}), we rewrite this using  (\ref{eq4.1.4}) as
\begin{equation}
\label{eq4.1.11}
(P-z)u = v + w \quad \wrtext{in}\,\,U, \ \ \| w \|_{H_{\Phi}(U)} \leq C e^{-\delta/h} \| u \|_{H_{\widehat{\Phi},m}(\mathbb C)}.
\end{equation}
Following Proposition \ref{p:norm} we apply the operator $B$ to (\ref{eq4.1.11}) to get
\begin{equation}
\label{eq4.1.13}
(Q-z)\widetilde{u} = \widetilde{v} \quad \wrtext{in}\,\,V, \ \ \ \widetilde{u} = Bu \in H_{\Phi_0}(V),
\end{equation}
where
\begin{equation}
\label{eq4.1.14}
\widetilde{v} = B(v+w) - BP(1-AB)u
\end{equation}
satisfies
\begin{equation}
\label{eq4.1.15}
\|\widetilde{v}\|_{H_{\Phi_0}(V)} \leq C \|{v}\|_{H_{\widehat \Phi}(\mathbb C)} + C  e^{-\delta/h} \| u \|_{H_{\widehat{\Phi},m}(\mathbb C)} + C \| u \|_{H_{\Phi}(U\setminus \widetilde{U})},
\end{equation}
with $\widetilde{U} \Subset U$ being a small open neighbourhood of $0\in \mathbb C$. Here we recall that
\[ V = D ( 0 , r ) , \ \ \ 0 < r \ll 1. \]

In the notation of Lemma \ref{a_priori}, we define
\begin{equation}
\label{eq4.2.7}
\Pi_{N(h)} : H_{\Phi_0}(V) \rightarrow H_{\Phi_0}(V), \ \ \
\Pi_{N(h)} u :=  \sum_{j=0}^{N(h) -1} {u^{(j)}(0)} x^j/j!,
\end{equation}
which in view of \eqref{eq4.1.18} is
the orthogonal projection onto the span of $f_j$'s, $ 0 \leq j < N ( h ) $ in \eqref{eq4.1.17}. Lemma \ref{a_priori} then gives
\begin{equation}
\label{eq4.2.9}
1 - \Pi_{N(h)} = \mathcal O(1)\, e^{-\delta/h}: H_{\Phi_0}(V) \rightarrow H_{\Phi_0}(W), \ \
 \ W := D ( 0 , r/16) .
\end{equation}

Lemma \ref{decay_exp} and \eqref{eq4.2.5} show that
$\| {\varphi_j}\|_{H_{\Phi_0}(V)} = 1 + C e^{-r^2/4h}$, $ j < N(h)$, uniformly with respect to $j$, provided that $\delta>0$ in (\ref{eq4.2.5}) satisfies $0 < \delta < r^2/32$. Hence, in the notation of  (\ref{eq4.1.17}),
\begin{equation}
\label{eq4.4.2}
f_j = \varphi_j + \mathcal O_{ H_{\Phi_0} (V) }(1) \, e^{-r^2/4h}, \quad j < N(h).
\end{equation}
Another application of Lemma \ref{decay_exp} (with $ \rho= \alpha r $) together with (\ref{eq4.4.2}) gives that
\begin{equation}
\label{eq4.4.4}
\|f_j\|_{H_{\Phi_0}(\alpha r < \abs{x} < r)} = \mathcal O(1)\, e^{-\eta/h}, \quad \alpha = \left(\tfrac{2}{3}\right)^{\frac12}, \ \ \eta > 0,
\end{equation}
uniformly with respect to $j < N(h)$, provided that $\delta>0$ in (\ref{eq4.2.5}) satisfies $0 < \delta < r^2/48$. In what follows, we shall assume that $\delta >0$ is chosen small enough so that (\ref{eq4.4.2}), (\ref{eq4.4.4}) hold.

Applying the orthogonal projection (\ref{eq4.2.7}) to (\ref{eq4.1.13}), we get
\begin{equation}
\label{eq4.6}
\Pi_{N(h)} (Q-z) \Pi_{N(h)} \widetilde{u} = \Pi_{N(h)} \widetilde{v} - \Pi_{N(h)} (Q -z) (1-\Pi_{N(h)})\widetilde{u}.
\end{equation}
We first estimate the norm of the second term on the right hand side of \eqref{eq4.6}.
To this end, let $\widetilde{V} =D ( 0, r' ) $ with
$ r/16 < r' < r $ sufficiently close to $ r/16 $ so that
we still have
\begin{equation}
\label{eq4.8}
1 - \Pi_{N(h)} = \mathcal O(1)\, e^{-\delta/h}: H_{\Phi_0}(V) \rightarrow H_{\Phi_0}(\widetilde{V}).
\end{equation}
Then, (see Proposition \ref{p:norm} for the mapping properties of $ Q  - z $)
\begin{equation}
\label{eq4.9}
\begin{split}
& \|\Pi_{N(h)} (Q -z) (1-\Pi_{N(h)})\widetilde{u}\|_{H_{\Phi_0}(V)} \leq C \|(1-\Pi_{N(h)})\widetilde{u}\|_{H_{\Phi_0}(V)} \\
&  \ \ \ \ \ \leq C \| \indic_{\widetilde{V}} (1-\Pi_{N(h)})\widetilde{u}\|_{L^2_{\Phi_0}(V)}
+ C \| (1-\indic_{\widetilde{V}})(1-\Pi_{N(h)})\widetilde{u}\|_{L^2_{\Phi_0}(V)} \\
&  \ \ \ \ \  \leq C \|\widetilde{u}\|_{H_{\Phi_0}(V\backslash \widetilde{V})} + Ce^{-\delta/h} \|\widetilde{u}\|_{H_{\Phi_0}(V)} .
\end{split}
\end{equation}
It follows next from (\ref{eq4.1.6}) that the matrix $\mathcal D_{N(h)}$ of the operator $\Pi_{N(h)}(Q-z)\Pi_{N(h)}$ with respect to the orthonormal basis $f_j$, $0 \leq j \leq N(h)-1$, is of the form ($ \delta > 0 $)
\begin{equation}
\label{eq4.10}
\left( G((2j+1)h;h)-z)\delta_{j,k}\right)_{0\leq j,k\leq N(h)-1} + \mathcal O(e^{-\delta/h})
 = \mathcal O(1): \mathbb C^{N(h)} \rightarrow \mathbb C^{N(h)}.
\end{equation}
Here we equip $\mathbb C^{N(h)}$ with the Euclidean norm.

For $ z $ satisfying \eqref{eq:condz} we also have
\begin{equation}
\label{eq4.12}
\mathcal D_{N(h)}^{-1} = \mathcal O(h^{-1}): \mathbb C^{N(h)} \rightarrow \mathbb C^{N(h)}.
\end{equation}
Hence, for $z$ satisfying \eqref{eq:condz},  (\ref{eq4.6}), (\ref{eq4.9}), and (\ref{eq4.12}) give
\begin{equation*}
\begin{split}
\|\Pi_{N(h)} \widetilde{u}\|_{H_{\Phi_0}(V)} & \leq C h^{-1} \left(\|\Pi_{N(h)} \widetilde{v}\|_{H_{\Phi_0}(V)} + \|\Pi_{N(h)} (Q -z) (1-\Pi_{N(h)})\widetilde{u}\|_{H_{\Phi_0}(V)}\right) \\
& \leq  C h^{-1} \left(\|\widetilde{v}\|_{H_{\Phi_0}(V)} + \|\widetilde{u}\|_{H_{\Phi_0}(V\backslash \widetilde{V})}\right) + C e^{-\delta/h} \|\widetilde{u}\|_{H_{\Phi_0}(V)}, \quad \delta > 0.
\end{split}
\end{equation*}
Using this and (\ref{eq4.2.9}) we obtain that
\begin{equation}
\label{eq4.15}
\begin{split}
\|\widetilde{u}\|_{H_{\Phi_0}(W)} & \leq  \|\Pi_{N(h)} \widetilde{u}\|_{H_{\Phi_0}(W)} + \|(1-\Pi_{N(h)})\widetilde{u}\|_{H_{\Phi_0}(W)} \\
& \leq C h^{-1} \left(\|\widetilde{v}\|_{H_{\Phi_0}(V)} + \|\widetilde{u}\|_{H_{\Phi_0}(V\backslash \widetilde{V})}\right) + C e^{-\delta/h} \|\widetilde{u}\|_{H_{\Phi_0}(V)}.
\end{split}
\end{equation}
Recalling (\ref{eq4.1.15}) and the fact that $\widetilde{u} = Bu$, we get
\begin{equation}
\label{eq4.16}
\begin{split}
\|Bu\|_{H_{\Phi_0}(W)}
& \leq C h^{-1}\left(\|v\|_{H_{\widehat \Phi}(\mathbb C)} + \| u \|_{H_{\Phi}(U\backslash \widetilde{U})} + \|Bu\|_{H_{\Phi_0}(V\backslash \widetilde{V})}\right) + C e^{-\delta/h} \| u \|_{H_{\widehat{\Phi},m}(\mathbb C)} \\
& \leq C h^{-1}\left(\|v\|_{H_{\widehat \Phi}(\mathbb C)} + \| u \|_{H_{\Phi}(U\backslash \widehat{U})}\right) + C\, e^{-\delta/h} \| u \|_{H_{\widehat{\Phi},m}(\mathbb C)}.
\end{split}
\end{equation}
Here $\widehat{U} \Subset \widetilde{U} \Subset U$ and we have used the fact that
\begin{equation}
\label{eq4.17}
\|Bu\|_{H_{\Phi_0}(V\setminus \widetilde{V})} \leq C \| u \|_{H_{\Phi}(U\setminus \widehat{U})} +
C e^{-\delta/h} \| u \|_{H_{\Phi}(U)}, \quad \delta > 0,
\end{equation}
for $\widehat{U} \Subset \widetilde{\kappa}(\widetilde{V}) \Subset U$. Writing $u = ABu + (1-AB)u$ we obtain from (\ref{eq4.16}) for $U_1 \Subset \widetilde{\kappa}(W) \Subset U$,
\begin{equation}
\label{eq4.18}
\| u \|_{H_{\Phi}(U_1)} \leq C h^{-1}\left(\|v\|_{H_{\widehat \Phi}(\mathbb C)} + \| u \|_{H_{\Phi}(U\backslash \widehat{U})}\right) + C e^{-\delta/h} \| u \|_{H_{\widehat{\Phi},m}(\mathbb C)}.
\end{equation}
Combining (\ref{eq4.18}) with (\ref{eq4.1.2}) applied for $\varepsilon > 0$ sufficiently small but fixed, we obtain that for $z$ satisfying
\eqref{eq:condz}, equation (\ref{eq4.1.10}) implies that
\begin{equation}
\label{eq4.19}
\| u \|_{H_{\widehat{\Phi},m}(\mathbb C)} \leq C h^{-1}\|v\|_{H_{\widehat{\Phi}}(\mathbb C)} + C e^{-\delta/h} \| u \|_{H_{\widehat{\Phi},m}(\mathbb C)}, \quad \delta > 0.
\end{equation}
The operator \eqref{eq:4.20} is therefore injective, for all $h>0$ small enough, and hence bijective, as it is Fredholm of index $0$. This completes the proof.
\end{proof}

The main part of Theorem \ref{t:gen} follows from the next proposition about the spectrum of $ P^{\rm{w}}(x,hD_x;h)$: it follows from (\ref{eq4.1.1.11}) that $\widehat{\Phi} - \Phi_1$ is compactly supported so that the spectrum of $ P^{\rm{w}} $ on $ H_{\widehat{\Phi}}(\mathbb C) $ is the same as the spectrum on $ H_{\Phi_1 }(\mathbb C)$ (the two spaces are the same but the norms are different). We also recall that
the operators $ H^{\rm{w}} $ on $ L^2 ( \mathbb R ) $ and $ P^{\rm{w}} $ on $ H_{\Phi_1 }(\mathbb C)$ are unitarily equivalent using an FBI transform adapted to $ \Phi_1 $, see (\ref{eq4.1.0.2}), (\ref{eq4.1.0.3}).

\begin{prop}
\label{p:new}
Suppose that
\begin{equation}
\label{eq4.21}
\abs{z - G((2j+1)h;h)} <{h}/C, \ \ \
hj \leq \tfrac14{\rho^2}, \ \ 0 < \rho \ll 1.
\end{equation}
Then we have
\begin{equation}
\label{eq4.23}
z \in \Spec_{ H_{\widehat{\Phi}}}( P^{\rm{w}} ( x, h D, h ) )
\ \Longleftrightarrow \ z = G((2j+1)h;h) + \mathcal O(e^{-\delta/h}), \ \ \delta > 0,
\end{equation}
and the eigenvalue is simple.
\end{prop}
\begin{proof}
We will use the Schur complement formula and
a  globally well-posed Grushin problem for the operator {$P^{\rm{w}}(x,hD_x;h) -z: H_{\widehat{\Phi},m}(\mathbb C) \rightarrow H_{\widehat{\Phi}}(\mathbb C)$, see \cite[\S D.1]{z12} for a general introduction to this method and references.

It is convenient for us to assume, as we may, that the order function $m$ introduced in (\ref{eq2.1.1}) satisfies $m\in C^{\infty}(\mathbb C)$ with
\begin{equation}
\label{eq4.23.1}
\abs{\partial^{\alpha} m} \leq C_{\alpha}m,\quad \forall \alpha \in \mathbb N^2,
\end{equation}
see~\cite[Lemma 8.7]{z12}.

As before, let $V = D(0,r)$ and  $\displaystyle \rho \leq r/C_0$, where $C_0>0$ is sufficiently large. This
guarantees that $V$ is much larger than the region where $\varphi_j$ in (\ref{eq4.1.8}) is not exponentially small -- see Lemma \ref{decay_exp}.  Let $\chi \in C^{\infty}_0(\mathbb C; [0,1])$ satisfy
\[ \chi(x) = 1, \ \ \abs{x} \leq \rho, \ \ \ \ \chi(x) < 1, \ \   \abs{x} > \rho,  \]
and define
\begin{equation}
\label{eq4.24}
\widetilde{\Phi}_0(x) = \Phi_0(x) - \tfrac18{\rho^2}(1 - \chi(x)).
\end{equation}
We then have
\begin{equation}
\label{eq4.25}
\widetilde{\Phi}_0(x) = \Phi_0(x), \quad \abs{x} \leq \rho, \quad \widetilde{\Phi}_0(x) < \Phi_0(x), \quad \abs{x} > \rho.
\end{equation}
It follows from Lemma \ref{decay_exp} that for $ hj \leq \frac14{\rho^2}$,
\begin{equation}
\label{eq4.26}
\|\varphi_j\|_{H_{\widetilde{\Phi}_0}(V)}^2 \leq \int_{\mathbb C} \abs{\varphi_j(x)}^2 e^{-2\widetilde{\Phi}_0(x)/h}\, L(dx) \leq 2.
\end{equation}
Using (\ref{eq3.9.2}), (\ref{eq3.9.4}), and (\ref{eq4.26}) we see that
\begin{equation}
\label{eq4.27}
\| A \varphi_j\| _{H_{\widetilde{\Phi}}(U)} =  \mathcal O(1), \quad hj \leq \tfrac14{\rho^2},
\end{equation}
where $\widetilde{\Phi} \in C(U;\mathbb R)$ is a suitable weight satisfying
\begin{equation}
\label{eq4.28}
\widetilde{\Phi} \leq \Phi \quad \wrtext{in}\,\,U, \quad \widetilde{\Phi} < \Phi \quad \wrtext{in}\,\,U\backslash \overline{U_1}.
\end{equation}
Indeed, we can take, in the notation of (\ref{eq3.9.2}) and (\ref{eq3.9.4}),
$$
\widetilde{\Phi}(x) = \Phi(x) - \inf_{y\in V} \left(\frac{\abs{x-\widetilde{\kappa}(y)}^2}{2C} + \frac{\rho^2}{8} (1-\chi(y))\right).
$$
Here $U_1 \Subset U \Subset \mathbb C$ can be chosen arbitrarily small provided that $\rho > 0$ is small enough.
Let $ \Pi_{\widehat{\Phi}_m}: L^2_{\widehat{\Phi},m}(\mathbb C) \rightarrow H_{\widehat{\Phi},m}(\mathbb C)$ be the orthogonal
projection and let
\begin{equation}
\label{eq:props} \psi \in C^{\infty}_0(U;[0,1]), \ \ \ \psi (x) = 1,  \ \ x \in \neigh_{\mathbb C } (\overline{U_1}) .
\end{equation}
We define a continuous uniformly bounded map
\begin{equation}
\label{eq4.29}
R_- = \mathcal O(1): \mathbb C \ni u_- \longmapsto u_-\Pi_{\widehat{\Phi}_m}(\psi A \varphi_j) \in H_{\widehat{\Phi},m}(\mathbb C).
\end{equation}
Using \eqref{eq4.23.1} we get
\begin{equation}
\label{eq4.30.1}
\begin{gathered}
\widehat{\Phi}_m := \widehat{\Phi} - h\log m \in  C^{\infty}(\mathbb C), \ \ \ \partial^{\alpha} \widehat{\Phi}_m \in L^{\infty}(\mathbb C), \quad \abs{\alpha} = 2, \\
\Delta \widehat{\Phi}_m \geq \Delta \widehat{\Phi} - \mathcal O(h) \geq 1/{C} - \mathcal O(h) \geq {1}/{2C}, \ \
0 < h < h_0 ,
\end{gathered}
\end{equation}
and  $L^2_{\widehat{\Phi},m}(\mathbb C) = L^2_{\widehat{\Phi}_m}(\mathbb C)$, $H_{\widehat{\Phi},m}(\mathbb C) = H_{\widehat{\Phi}_m}(\mathbb C)$.

It is important for us to know that $R_- u_-$ is exponentially small away from the origin:
\begin{lemm}
\label{loc_R-}
There exists $\delta >0$ such that
\begin{equation}
\label{eq4.31}
\|R_- u_- - u_- \psi A\varphi_j\|_{L^2_{\widehat{\Phi},{m}}(\mathbb C)} \leq C e^{-\delta/h}\abs{u_-}, \ \
hj \leq \tfrac14{\rho^2}
\end{equation}
\end{lemm}
\begin{proof}
Definition (\ref{eq4.29}) gives
\begin{equation}
\label{eq4.32}
R_- u_- - u_- \psi A\varphi_j = -u_-(1-{\Pi_{\widehat{\Phi}_m})}(\psi A\varphi_j),
\end{equation}
where $ u := {(1- \Pi_{\widehat{\Phi}_m})(\psi A \varphi_j) \in L^2_{\widehat{\Phi}_m}(\mathbb C)}$ is the  minimal norm solution of
\[ \bar\partial u = \bar{\partial} (\psi A \varphi_j) =( \bar \partial \psi) A \varphi_j =  \mathcal O_{L^2_{\widehat{\Phi}_m}}( e^{-\delta/h}), \quad \delta >0,\]
where we use (\ref{eq4.27}), (\ref{eq4.28}), and \eqref{eq:props} to obtain the last estimate ($ \bar \partial \psi$ is supported in the region where we have uniform estimates with better weights).

An application of $L^2$-estimates for the $\overline{\partial}$ equation in the exponentially weighted space $L^2_{\widehat{\Phi},m}(\mathbb C) = L^2_{\widehat{\Phi}_m}(\mathbb C)$ (in the elementary one-dimensional case, see~\cite[\S 4.2]{Horm_Conv}) shows that $ \| u \|_{ L^2_{\widehat{\Phi}_m} } = \mathcal O ( e^{-\delta/h} ) $,
concluding the proof.
\end{proof}

Using $ R_-$ from \eqref{eq4.29} and defining
\begin{equation}
\label{eq4.35.1}
R_+: H_{\widehat{\Phi}}(\mathbb C) \ni u \mapsto (B(u|_U),\varphi_j)_{H_{\Phi_0}(V)} \in  \mathbb C,
\end{equation}
where $U\Subset \mathbb C$ is the open neighbourhood of $0$ given in (\ref{eq4.1.3}),
we pose the following global Grushin problem for $z \in \mathbb C$ satisfying (\ref{eq4.21}):
\begin{equation}
\label{eq4.36}
\mathcal P( z )  :=  \begin{pmatrix} P^{\rm{w}}  - z & R_- \\
R_+ & 0 \end{pmatrix} : H_{\widehat{\Phi},m}(\mathbb C)\times \mathbb C \to H_{\widehat{\Phi}}(\mathbb C) \times \mathbb C, \ \
\mathcal P ( z ) \begin{pmatrix} u \\
\ u_- \end{pmatrix} = \begin{pmatrix} v \\ \ v_+ \end{pmatrix}.
\end{equation}

Restricting $ u $ to $U$, using (\ref{eq4.1.4}), Lemma \ref{loc_R-}, (\ref{eq4.27}), (\ref{eq4.28}),
as well as the fact that $ \| (\psi-1)A\varphi_j\|_{ H_\Phi (U) } = \mathcal O(e^{-\delta/h})$ {in view of (\ref{eq:props})}, we get
\begin{equation}
\label{eq4.37}
\begin{gathered}
(P-z)u + u_- A\varphi_j = v + w \quad \wrtext{in}\,\,U, \quad R_+ u = v_+, \\
\| w \|_{H_{\Phi}(U)} \leq C e^{-\delta/h} (\| u \|_{H_{\widehat{\Phi},m}(\mathbb C)} + \abs{u_-}) ,  \ \ \delta > 0 .
\end{gathered}
\end{equation}
Applying the operator $B$ to the first equation in (\ref{eq4.37}) we get, similarly to (\ref{eq4.1.13}),
\begin{equation}
\label{eq4.39}
(Q-z) \widetilde{u} + u_- \varphi_j = \widetilde{v} \quad \wrtext{in}\,\,V, \quad (\widetilde{u},\varphi_j)_{H_{\Phi_0}(V)} = v_+.
\end{equation}
Here $\widetilde{u} = B(u|_U) \in H_{\Phi_0}(V)$.
With $ \widetilde{U} \Subset U$ close to $ U $, and similarly to (\ref{eq4.1.14}), (\ref{eq4.1.15}), we have
\begin{equation}
\label{eq4.40}
\begin{gathered}
\widetilde{v} = B(v+w) -BP(1-AB)u + u_-(1-BA)\varphi_j,\\
\|\widetilde{v}\|_{H_{\Phi_0}(V)} \leq C \|v\|_{H_{\widehat{\Phi}}(\mathbb C)} + C e^{-\delta/h} \left(\| u \|_{H_{\widehat{\Phi},m}(\mathbb C)} + \abs{u_-}\right) + C \| u \|_{H_{\Phi}(U\setminus \widetilde{U})}.
\end{gathered}
\end{equation}
Recalling (\ref{eq4.4.2}), we rewrite the microlocal Grushin problem (\ref{eq4.39}) as follows,
\begin{equation}
\label{eq4.42}
(Q-z) \widetilde{u} + u_- f_j = \widetilde{v} \quad \wrtext{in}\,\,V, \quad (\widetilde{u},f_j)_{H_{\Phi_0}(V)} = \widetilde{v}_+.
\end{equation}
Here $f_j \in H_{\Phi_0}(V)$ is given by (\ref{eq4.1.17}). The function $\widetilde{v}$ in (\ref{eq4.42}) is modified by an exponentially small quantity compared to the $\widetilde{v}$ in (\ref{eq4.40}) to take into account the difference between $ f_j $ and $ \varphi_j $, see \eqref{eq4.4.2}. It still satisfies the estimate in (\ref{eq4.40}) and we also have
\begin{equation}
\label{eq4.43}
\abs{\widetilde{v}_+} \leq \abs{v_+} + C e^{-\delta/h} \| u \|_{H_{\widehat{\Phi},m}(\mathbb C)}.
\end{equation}

The analysis of (\ref{eq4.42}) proceeds similarly to (\ref{eq4.1.13}), applying the orthogonal projection $\Pi_{N(h)}$ in (\ref{eq4.2.7}) to the first equation in (\ref{eq4.42}) and using that $\Pi_{N(h)} f_j = f_j$, for $hj \leq \rho^2/4$. We get
\begin{equation}
\label{eq4.44}
\Pi_{N(h)} (Q-z) \Pi_{N(h)} \widetilde{u} + u_- f_j = \Pi_{N(h)} \widetilde{v} - \Pi_{N(h)} (Q -z) (1-\Pi_{N(h)})\widetilde{u},
\end{equation}
\begin{equation}
\label{eq4.45}
(\Pi_{N(h)}\widetilde{u},f_j)_{H_{\Phi_0}(V)} = \widetilde{v}_+.
\end{equation}
Recalling that the matrix of the operator $\Pi_{N(h)} (Q-z) \Pi_{N(h)}$ acting on the range of $\Pi_{N(h)}$, with respect to the orthonormal basis $f_j$, $0\leq j \leq N(h) -1$, is given by (\ref{eq4.10}) and using (\ref{eq4.21}) we conclude that the problem (\ref{eq4.44}), (\ref{eq4.45}) enjoys the estimate
\begin{equation*}
\begin{split}
h \|\Pi_{N(h)} \widetilde{u}\|_{H_{\Phi_0}(V)} + \abs{u_-}
& \leq C \|\Pi_{N(h)} \widetilde{v} - \Pi_{N(h)} (Q -z) (1-\Pi_{N(h)})\widetilde{u}\|_{H_{\Phi_0}(V)} \\
& \ \ \ \ \ \ + C h \abs{\widetilde{v}_+}  + C\,e^{-\delta/h}\|\widetilde{u}\|_{H_{\Phi_0}(V)} \\
& \leq C \|\widetilde{v}\|_{H_{\Phi_0}(V)} + C \|\Pi_{N(h)} (Q -z) (1-\Pi_{N(h)})\widetilde{u}\|_{H_{\Phi_0}(V)}
\\ & \ \ \ \ \ \ + C h \abs{\widetilde{v}_+} + C\, e^{-\delta/h}\|\widetilde{u}\|_{H_{\Phi_0}(V)}, \quad \delta > 0.
\end{split}
\end{equation*}
Combining this with  (\ref{eq4.9}), (\ref{eq4.40}) and (\ref{eq4.43}) gives
\begin{equation}
\label{eq4.47}
\begin{split}
h \|\Pi_{N(h)} \widetilde{u}\|_{H_{\Phi_0}(V)} + \abs{u_-} & \leq C \|v\|_{H_{\widehat{\Phi}}(\mathbb C)} +  C\| u \|_{H_{\Phi}(U\backslash \widetilde{U})} + C  \|\widetilde{u}\|_{H_{\Phi_0}(V\backslash \widetilde{V})}  \\
& \ \ \ \ \ +C  h\abs{v_+} + C  e^{-\delta/h} \left(\| u \|_{H_{\widehat{\Phi},m}(\mathbb C)} + \abs{u_-}\right).
\end{split}
\end{equation}
Here $\widetilde{V} \Subset V$ and $\widetilde{U} \Subset U$. We obtain therefore, in view of (\ref{eq4.2.9}) and (\ref{eq4.47}),
\begin{equation}
\label{eq4.48}
\begin{split}
h \|\widetilde{u}\|_{H_{\Phi_0}(W)} + \abs{u_-} & \leq h \|\Pi_{N(h)} \widetilde{u}\|_{H_{\Phi_0}(W)} + \abs{u_-} + h \|(1-\Pi_{N(h)}) \widetilde{u}\|_{H_{\Phi_0}(W)} \\
& \leq C\left(\|v\|_{H_{\widehat{\Phi}}(\mathbb C)} + \| u \|_{H_{\Phi}(U\backslash \widetilde{U})} + \|\widetilde{u}\|_{H_{\Phi_0}(V\backslash \widetilde{V})} + h\abs{v_+}\right) \\
& \ \ \ \ \ \ \ \ \ \ + C\, e^{-\delta/h}\left(\| u \|_{H_{\widehat{\Phi},m}(\mathbb C)} + \abs{u_-}\right).
\end{split}
\end{equation}
We get from (\ref{eq4.48}), using that $\widetilde{u} = B(u|_U)$, similarly to (\ref{eq4.16}), (\ref{eq4.17}), (\ref{eq4.18}),
\begin{equation}
\label{eq4.49}
\begin{split}
h \| u \|_{H_{\Phi}(U_1)} + \abs{u_-}
& \leq C\left(\|v\|_{H_{\widehat{\Phi}}(\mathbb C)} + \| u \|_{H_{\Phi}(U\backslash \widehat{U})} + h\abs{v_+}\right) \\
& \ \ \ \ \ \ \ \ \ \ + C\, e^{-\delta/h}\left(\| u \|_{H_{\widehat{\Phi},m}(\mathbb C)} + \abs{u_-}\right).
\end{split}
\end{equation}
Here $U_1 \Subset \widehat{U} \Subset U$. We have using the elliptic estimate (\ref{eq4.1.2}) and the first equation in (\ref{eq4.36}),
\begin{equation}
\label{eq4.50}
\| u \|_{H_{\Phi}(U\backslash \widehat{U})} \leq C\left(\|v\|_{H_{\widehat{\Phi}}(\mathbb C)} + e^{-\delta/h}\left(\| u \|_{H_{\widehat{\Phi},m}(\mathbb C)} + \abs{u_-}\right)\right).
\end{equation}
Here we have also used that, in view of Lemma \ref{loc_R-}, the term $R_- u_-$ is exponentially small away in $L^2_{\widehat{\Phi},m}$ from an arbitrarily small neighbourhood of $0$, provided that $\rho > 0$ in (\ref{eq4.21}) is small enough. We get, injecting (\ref{eq4.50}) into (\ref{eq4.49}),
\begin{equation}
\label{eq4.51}
h \| u \|_{H_{\Phi}(U_1)} + \abs{u_-} \leq C\left(\|v\|_{H_{\widehat{\Phi}}(\mathbb C)} + h\abs{v_+}\right) + C\, e^{-\delta/h}\left(\| u \|_{H_{\widehat{\Phi},m}(\mathbb C)} + \abs{u_-}\right).
\end{equation}}
Using the elliptic estimate again,
\begin{equation}
\label{eq4.52}
\| u \|_{H_{\widehat{\Phi},m}(\mathbb C \backslash U_1)} \leq C\left(\|v\|_{H_{\widehat{\Phi}}(\mathbb C)} + e^{-\delta/h}\left(\| u \|_{H_{\widehat{\Phi},m}(\mathbb C)} + \abs{u_-}\right)\right),
\end{equation}
valid for $\rho > 0$ small enough, we get, adding (\ref{eq4.51}) and (\ref{eq4.52}),
\begin{equation}
\label{eq4.53}
h \| u \|_{H_{\widehat{\Phi},m}(\mathbb C)} + \abs{u_-} \leq C\left(\|v\|_{H_{\widehat{\Phi}}(\mathbb C)} + h\abs{v_+}\right) + C e^{-\delta/h}\left(\| u \|_{H_{\widehat{\Phi},m}(\mathbb C)} + \abs{u_-}\right).
\end{equation}
We have proved therefore that for all $h>0$ small enough, the operator $ \mathcal P ( z ) $ in
\eqref{eq4.36}
is injective, for $z\in \mathbb C$ satisfying (\ref{eq4.21}) for $\rho > 0$ sufficiently small but fixed.
Since $ \mathcal P ( z) $  is a Fredholm operator of index zero (it is a
finite rank perturbation of the operator in which $ R_{\pm } $ are replaced by $ 0 $ and that operator
has that property thanks to the Fredholm property of $P^{{\rm w}} -z$), it is consequently bijective.

Let
\begin{equation}
\label{eq4.55}
\mathcal E ( z) =
\left( \begin{array}{ccc}
E(z) & E_+ ( z) \\
E_- ( z) & E_{-+}(z)
\end{array} \right): H_{\widehat{\Phi}}(\mathbb C) \times \mathbb C \rightarrow
H_{\widehat{\Phi},m}(\mathbb C) \times \mathbb C,
\end{equation}
be the inverse of $\mathcal P( z) $ in (\ref{eq4.36}), and let us recall that $z\in \mathbb C$ in the set (\ref{eq4.21}) belongs to spectrum of $P^{\rm{w}}(x,hD_x;h)$ precisely when $E_{-+}(z)=0$. We shall now compute $E_{-+}(z)$ up to an exponentially small term. For that we set
\begin{equation}
\label{eq4.56}
\widetilde{E}_+ = R_-, \quad \widetilde{E}_{-+}(z) = z - G ((2j +1)h;h).
\end{equation}
Using (\ref{eq4.29}), (\ref{eq4.35.1}), (\ref{eq4.37}), (\ref{eq4.28}), and Lemma \ref{loc_R-}, we see that uniformly in $j$, we have for $v_+ \in \mathbb C$,
\begin{equation}
\label{eq4.57}
\begin{split}
R_+ \widetilde{E}_+ v_+ & = R_+(v_+\Pi_{\widehat{\Phi}_m}(\psi A\varphi_j)) = R_+(v_+ \psi A\varphi_j) + \mathcal O(e^{-\delta/h}) v_+  \\
& = R_+(v_+ A\varphi_j) + \mathcal O(e^{-\delta/h}) v_+ = v_+ + \mathcal O(e^{-\delta/h}) v_+,
\end{split}
\end{equation}
that is, $ R_+ \widetilde{E}_+ - 1 = \mathcal O  (e^{-\delta/h}): \mathbb C \rightarrow \mathbb C$.

Next, we consider
\begin{equation}
\label{eq4.59}
\begin{split}
\left(P^{\rm{w}} -z\right)\widetilde{E}_+\,v_+  + R_- \widetilde{E}_{-+}v_+
& = \left(P^{\rm{w}} -z\right)R_- \,v_+
+  R_-\left(z- G ((2j +1)h;h) \right)v_+ .
\end{split}
\end{equation}
We first observe that in view of  Lemma \ref{loc_R-},
\begin{equation}
\label{eq4.60}
\| R_- v_+\|_{ L^2_{\widehat{\Phi},m} ( \mathbb C\backslash \widetilde{U} )} =
\mathcal O( e^{-\delta/h})\abs{v_+},
\end{equation}
where $\widetilde{U} \Subset U$ is a neighborhood of the origin, and thus the $L^2_{\widehat{\Phi}}(\mathbb C \setminus U)$--norm of the function in (\ref{eq4.59}) is $\mathcal O(e^{-\delta/h}) |v_+|$.
Restricting the attention to the neighbourhood $U$, we see that modulo an error term of the $H_{\Phi}(U)$--norm $\mathcal O( e^{-\delta/h})\abs{v_+}$, the holomorphic function in (\ref{eq4.59}) is equal to
\begin{equation}
\label{eq4.62}
v_+ \left((P-z)A\varphi_j + A\varphi_j \left(z- G ((2j +1)h;h)\right)\right) = v_+ \left(PA-AQ\right)\varphi_j \\
= \mathcal O( e^{-\delta/h}) |v_+|,
\end{equation}
in $H_{\Phi}(U)$. It follows that
\begin{equation}
\label{eq4.63}
\left(P^{\rm{w}}(x,hD_x;h) -z\right)\widetilde{E}_+ + R_- \widetilde{E}_{-+} = \mathcal O( e^{-\delta/h}): \mathbb C \rightarrow H_{\widehat{\Phi}}(\mathbb C).
\end{equation}
Combining (\ref{eq4.57}) and (\ref{eq4.63}) with the well-posedness of the Grushin problem, we see that we have computed $E_{-+}(z)$ in (\ref{eq4.55}) up to an exponentially small error term,
\begin{equation}
\label{eq4.64}
E_{-+}(z) = z - G ((2j +1)h;h) + \mathcal O( e^{-\delta/h}),
\end{equation}
uniformly in $j$ satisfying (\ref{eq4.21}), for $\rho > 0$ small enough but fixed. The proof of Theorem \ref{t:gen} is complete.
\end{proof}

We finish by indicating the vanishing of $ G_1 $ when the subprincipal symbol of $ H^{{\rm w}}$ in Theorem \ref{t:gen} vanishes. For this we can work in the setting of \cite{Hi04} and consider expansions modulo $ \mathcal O(h^\infty) $ and $ C^\infty $ Fourier integral operators.
In the proof of Proposition \ref{p:norm}, we can follow the proof of \cite[(a.3.11)]{HelSj} which shows that the Fourier integral operator $A_0$ can be chosen so that the conjugation by $ A_0 $ gives a normal form up to $ \mathcal O ( h^2) $ terms, see also~\cite[Section 2]{HitSj1}. The pseudodifferential conjugation does not introduce any $ \mathcal O ( h )$ terms and that shows that $ G_1 \equiv 0 $.}

\section*{Appendix: Vey's Morse Lemma in dimension 2}
\renewcommand{\theequation}{A.\arabic{equation}}
\refstepcounter{section}
\renewcommand{\thesection}{A}
\setcounter{theo}{0}
\newtheorem{thm}{Theorem}[section]
\setcounter{equation}{0}

For the reader's convenience we present a self-contained account of Vey's Morse Lemma in dimension two:
\begin{thm}
\label{t:vey}
Suppose that \[ p \in \mathscr O ( \neigh_{\mathbb C^2 }  (0,0) ), \ \ \
 p ( x, \xi ) = q ( x, \xi )  + \mathcal O (( x, \xi)^3 ) ),  \]
where $ q $ is a
quadratic form such that $ \{ q ( x, \xi ) :
( x, \xi ) \in \mathbb R^2 \} \neq \mathbb C $ and $ q|_{\mathbb R^2 \setminus 0 } \neq 0 $.
 Then there exists
a\ biholomorphic map $ F : \neigh_{\mathbb C^2 } ( 0 ) \to \neigh_{\mathbb C^2 } ( 0 ) $
and $ f \in \mathscr O ( \neigh_{\mathbb C }  (0 ) ) $, $ f ( 0 ) = 1 $,
such that
\begin{equation}
\label{eq:vey}
F^* p ( z, \zeta ) = q( z ,\zeta),  \ \ \ F^* ( d\xi \wedge d x ) = f ( q(z ,\zeta )) d \zeta \wedge d z ,
\end{equation}
{$ F ( 0 )  = 0$, $ dF ( 0 ) = I_{\mathbb C^2} $}. Equivalently, we can find a biholomorphic map
\[
\kappa : \neigh_{\mathbb C^2 } ( 0 ) \to \neigh_{\mathbb C^2 } ( 0 )
\]
such that
\begin{equation}
\label{eq:vey1}
\begin{gathered}
\kappa^* p ( z, \zeta ) = g( q (z, \zeta))  ,  \ \ \ \ \kappa^* ( d\xi \wedge d x ) = d \zeta \wedge d z ,\\
{\kappa(0)=0, \ \  d\kappa ( 0) = I_{\mathbb C^2} }, \ \ \
g'(t ) = f ( g ( t ) ) ^{-1} , \ \ \  g ( 0 ) = 0 .
\end{gathered}
\end{equation}
If $ p|_{ \mathbb R^2 } $ is real valued and $ q $ is positive definite, then $ F, \kappa$ and $f, g $ are real valued in the real domain.
We also note that while $ F$ and $ \kappa $ are not unique, $ f $ and $ g $ are.
\end{thm}

We start by recalling the following fact (see \cite[Proposition 2.1.10]{pravda} for a slightly different version
and more general statements):
\begin{lemm}
\label{l:quand}
\noindent
Any 
holomorphic quadratic form on $ \mathbb C^2 $ such that
$ \{ q ( x, \xi ) : ( x, \xi ) \in \mathbb R^2 \} \neq \mathbb C $ and $ q|_{\mathbb R^2\setminus 0 }
\neq 0 $
is complex symplectically equivalent to
 $ \frac 1 {2i} \mu ( x^2 + \xi^2 ) $, $ \mu \in \mathbb C \setminus \{ 0 \}$, (or equivalently to
 $ \mu x \xi $).
\end{lemm}
\begin{proof}
If $ q ( x, \xi ) = a x^2 + 2 b x\xi + c \xi^2 $ and, since $ q|_{\mathbb R^2 \setminus 0 } \neq 0 $,
we can rescale and assume $ a = 1 $. If $ a_\pm $ are roots of $ z^2 + 2bz  + c = 0 $ then
$ q ( x, \xi ) = ( x - a_+ \xi ) ( x - a_- \xi ) $, and again, $ a_+ , a_- \notin \mathbb R $. By
a real linear change of variables we can assume that $ a_+ = i $ and $ a_- \notin \mathbb R $.
Since the range of $ q $ is not all of $ \mathbb C$, $ a_- \neq i $.
Writing $ z = x + i \xi $, we then get
\[ q ( x, \xi ) = \alpha \bar z ( z + \beta \bar z ), \ \ \alpha := - i (i - a_-)  , \ \
\beta = \frac{ i + a_- }{ i - a_- },   \ \ | \beta | \neq 1 . \]
If $ | \beta | < 1 $ then $ \Re ( \bar \alpha q ) = | \alpha|^2 ( |z|^2 - \Re ( \beta \bar z^2 ) ) > 0 $, $ z \neq 0 $.
Otherwise,  $ q ( x, \xi ) = \alpha \beta ( \bar z^2 + |z|^2/\beta )
= \alpha |\beta| r^2 ( e^{  i \theta } - 1/|\beta | )$, $ \theta = - \frac12 \arg z - \arg \beta
 $,  takes call complex values.

Hence (by multiplying $ q $ by a complex number ) we can assume that  $ \Re q|_{\mathbb R^2\setminus 0 } > 0 $.  We can then take a real linear canonical transformation such that
 $ \Re q |_{\mathbb R^2 } = \lambda (x^2 + \xi^2 ) $, $ \lambda \in \mathbb R \setminus 0 $. Now we use an orthogonal transformation (with
 determinant $ 1 $; it preserves our new real part) to diagonalize $ \Im q |_{\mathbb R^2 } =
 \mu_1 x^2 + \mu_2 \xi^2  $, $ \mu_j \in \mathbb R$. Complexifying these $ \mathbb R$-linear transformations
 gives us a complex symplectic reduction to $ ( \lambda + i \mu_1) x^2 + ( \lambda + i \mu_2 ) \xi^2 $.
 A complex symplectic scaling finally produces $ \tfrac1{2i} \mu ( x^2 + \xi^2 ) $, $ \mu \in \mathbb C \setminus \{0 \} $.  \end{proof}

In our setting the proof is a simple adaptation of the elegant argument presented by Colin de Verdi\`ere and Vey \cite{CV}. An alternative
argument was given (independently of \cite{vey}, \cite{CV}) by Helffer--Sj\"ostrand \cite{HelSj}. The advantage we stress here is the algorithmic nature of the method which allows computation of the functions $ f $ and $ g $.

We also record the following simple but useful
\begin{thm}
\label{t:action}
In the notation of Theorem \ref{t:vey}, suppose that $ q ( x, \xi ) = \frac12 \mu ( x^2 + \xi^2 ) $, $ \mu \in \mathbb C\setminus \{ 0 \} $  and that for $ t \in \neigh_{\mathbb R } ( 0) $,
\[  S ( t \mu ) := \int_{ \gamma ( t \mu ) } \xi d x  \]
where $ \gamma (t \mu ) $ is a closed cycle on the complex curve $ p ( x, \xi ) = t \mu $ close to
the positively oriented (real) circle $ \frac12 (x^2 + \xi^2)  =  t $. Then $ S \in \mathscr O ( \neigh_{\mathbb C}  (0 ) ) $ and in the notation of
\eqref{eq:vey} and \eqref{eq:vey1}
\begin{equation}
\label{eq:f2S} \mu S' ( w ) =  2 \pi  f  ( w ) , \ \ S ( 0 ) = 0 ,   \ \   \mu S ( g ( w ) ) = 2 \pi w  . \end{equation}
\end{thm}
\begin{proof}
In coordinates given by $ F $ in \eqref{eq:vey}, $ p = \frac12 \mu ( z^2+ \zeta^2) $ and
$ \gamma ( t \mu ) $ is the positively oriented real circle. By Stokes's theorem we have
(in the real domain and using polar coordinates)
\[  S ( t \mu ) = \int_{ \frac12 (x^2 + \xi^2) \leq t } f ( \tfrac12 \mu ( x^2 + \xi^2 ) ) d x d \xi = 2 \pi \int_{
\frac12 r^2 \leq t } f ( \tfrac12 \mu r^2 ) r dr = 2 \pi \int_{0}^t f ( \mu \tau ) d\tau . \]
This gives $  \mu S' ( t \mu ) =2  \pi f ( \mu t ) $ and the first statement in \eqref{eq:f2S} follows. Since $ f $ is holomorphic
near $ 0 $ so is $ S$. The relation between $ S $ and $ g $ now follows from the relation between
$ g $ and $ f $ given in \eqref{eq:vey1}.
\end{proof}

\begin{proof}[Proof of Theorem \ref{t:vey}]
We start by showing how \eqref{eq:vey1} follows from \eqref{eq:vey}.
If $ p ( x, \xi ) = q ( x, \xi ) $ and $ \omega := f ( q ( x, \xi ) ) d\xi \wedge dx $,
we need to find $ ( x, \xi ) = H ( z, \zeta ) $ such that $ H^* q = g ( q ) $ and
$ H^* \omega = d\zeta \wedge d z$.  Then $ \kappa = F \circ H $ gives \eqref{eq:vey1}.

We define $ H $ by
\[   x = \sqrt{a ( q ( z, \zeta ) ) } z, \ \ \  \xi = \sqrt{a ( q ( z , \zeta ) )} \zeta , \ \ \  a (0) = 1.  \]
so that
\[  H^* q  = q a ( q) , \ \ \  H^* \omega =
f ( a ( q ) q ) ( a( q ) +  q a'(q)  ) d \zeta \wedge d z  , \ \  q = q ( z, \zeta ) . \]
We then put
$ g ( t ) = a(t ) t $ where $ a $ is chosen so that,
\[  g' ( t ) = a(t) + t a'(t)  = \frac{1}{ f ( a(t) t ) } = \frac{ 1}{ f ( g (t)) } ,\]
which is the condition in \eqref{eq:vey1}. We also note that $dH(0)= I_{\mathbb C} $ and hence
$ d\kappa (0) = dF( 0) = I_{\mathbb C^2}$.

Following \cite{CV}, we start by establishing the following fact: for any $ R
\in \mathscr O ( \neigh_{\mathbb C^2 } ( 0 , 0))  $
there exists $ f \in \mathscr O ( \neigh_{\mathbb C } (0))  $ and $
\eta \in  \mathscr O ( \neigh_{\mathbb C^2 } ( 0 , 0)) $ such that
\begin{equation}
\label{eq:CVL}   R ( z , \zeta ) d \zeta \wedge dz = f ( q ( z, \zeta ) )d\zeta \wedge d z +
d q \wedge d \eta , \ \ \ \eta (0,0) = 0 .
\end{equation}
In dimension two and for holomorphic $ R $ (and unlike in the general case considered in
\cite{CV}) this is very simple. We can assume (by a complex linear symplectic
change of variables) that $ q ( z, \zeta ) = i \mu z \zeta $ so that \eqref{eq:CVL} becomes
\[  R ( z, \zeta ) = f ( i \mu z \zeta ) + i \mu ( z \partial_z \eta ( z, \zeta ) - \zeta \partial_\zeta \eta ( z, \zeta ) ), \ \ \eta (0, 0) = 0   . \]
If $ R ( z , \zeta ) = \sum_{ n,m } R_{nm} z^n \zeta^m $ then the solutions are given by
\begin{equation}
\label{eq:F2f_app}    f ( w ) = \sum_{ n=0}^\infty ( - i)^n \mu^{-n} R_{nn} w^n, \ \ \
\eta ( z, \zeta ) = (i \mu)^{-1} \sum_{ n\neq m } \frac{ R_{nm} z^n w^m}{ n-m} .
\end{equation}

To obtain \eqref{eq:vey} we first apply the standard holomorphic Morse lemma to obtain
 a biholomorphism $ \widetilde F $ such that
 \[  \widetilde F^* p ( z , \zeta ) = q ( z, \zeta ) , \ \ \ \widetilde F^* ( d \xi \wedge d x ) = R (z, \zeta ) d \zeta \wedge dz . \]
We then use \eqref{eq:CVL} and define
\[ \omega_t := f ( q ( z, \zeta ) )d\zeta \wedge d z + t d q \wedge d \eta . \]
We now search for a family of biholomorphic maps $ F_t $ such that $ F_0  ( z, \zeta ) = ( z, \zeta ) $
and
\begin{equation}
\label{eq:Ft}  F_t^* \omega_t = \omega_0 , \ \  F_t^* q = q . \end{equation}
Finding $ F_t $ is equivalent to finding a holomorphic family of vector fields $ X_t $ which define
$ F_t $ by $ X_t ( F_t ( x ) ) = \partial_t F_t ( x ) $.  Then \eqref{eq:Ft} is equivalent to
\begin{equation}
\label{eq:Xt}    \mathscr L_{X_t} \omega_t + \partial_t \omega_t = 0 , \ \ \
X_t q = 0 . \end{equation}
Cartan's formula and the definition of $ \omega_t $ show that that this is equivalent to
\[    d ( \iota_{X_t} \omega_t ) =  dq \wedge d \eta = - d ( \eta  d q ) . \]
Since $ d q \wedge d \eta $ vanishes at $ ( 0 , 0 )$, $ \omega_t $ is nondegenerate
in a neighbourhood of $ ( 0, 0 ) $ and we can find $ X_t $ such that
$ \iota_{X_t} \omega_t =  - \eta  d q$. We have
\[ 0 = \omega_t  ( X_t , X_t ) = - \eta dq (X_t ) = - \eta X_t q . \]
If $ \eta \equiv 0 $ then $ X_t \equiv 0 $. Otherwise, $ X_t q $ vanishes on a open set and
analyticity
shows that $ X_t q \equiv 0$. This means that  \eqref{eq:Xt}, and hence \eqref{eq:Ft}, hold
and we obtain \eqref{eq:vey} with $ F = \widetilde F \circ F_1 $.

Finally, we note that once we have $ F$ satisfying \eqref{eq:vey}, we can can compose it
with a linear symplectic transformation so that $ d F ( 0 ) = I_{\mathbb C^2} $. In fact,
once \eqref{eq:vey} holds, we have the $ dF ( 0 ) $ is symplectic (we have $ f ( 0 ) = 1 $)
and $ dF(0)^* q = q $. But then then $ \widetilde F := dF ( 0 )^{-1} \circ F $ satisfies
\eqref{eq:vey} and $ d \widetilde F ( 0 ) = I_{\mathbb C^2}$.
\end{proof}

\end{document}